\documentclass[english,aps]{revtex4}
\usepackage[T1]{fontenc}
\usepackage[latin9]{inputenc}
\usepackage{amsthm}
\usepackage{amsmath}
\usepackage{amssymb}
\usepackage{graphicx}

\makeatletter
\@ifundefined{textcolor}{}
{%
 \definecolor{BLACK}{gray}{0}
 \definecolor{WHITE}{gray}{1}
 \definecolor{RED}{rgb}{1,0,0}
 \definecolor{GREEN}{rgb}{0,1,0}
 \definecolor{BLUE}{rgb}{0,0,1}
 \definecolor{CYAN}{cmyk}{1,0,0,0}
 \definecolor{MAGENTA}{cmyk}{0,1,0,0}
 \definecolor{YELLOW}{cmyk}{0,0,1,0}
}
  \theoremstyle{plain}
  \newtheorem{prop}{\protect\propositionname}

\newcommand{\be }{\begin {equation}} \newcommand{\ee }{\end {equation}}

\newcommand{\ket}[1]{|#1\rangle}
\newcommand{\bra}[1]{\langle #1|}

\usepackage{subfigure}

\usepackage{enumitem}

\usepackage{babel}
\providecommand{\propositionname}{Proposition}

\makeatother

\usepackage{babel}
  \providecommand{\propositionname}{Proposition}

\begin{document}

\title{Coherence in quantum estimation}

\author{Paolo Giorda}

\email{magpaolo16@gmail.com}

\affiliation{Consorzio Nazionale Interuniversitario per la Scienze fisiche della
Materia (CNISM), Italy}

\author{Michele Allegra}

\email{mallegra@sissa.it}

\affiliation{Scuola Internazionale Superiore di Studi Avanzati (SISSA), I-34136
Trieste, Italy}
\begin{abstract}
The geometry of quantum states provides a unifying framework for estimation
processes based on quantum probes, and it allows to derive the ultimate
bounds of the achievable precision. We show a relation between the
statistical distance between infinitesimally close quantum states
and the second order variation of the coherence of the optimal measurement
basis with respect to the state of the probe. In Quantum Phase Estimation
protocols, this leads to identify coherence as the relevant resource
that one has to engineer and control to optimize the estimation precision.
Furthermore, the main object of the theory i.e., the Symmetric Logarithmic
Derivative, in many cases allows to identify a proper factorization
of the whole Hilbert space in two subsystems. The factorization allows:
to discuss the role of coherence vs correlations in estimation protocols;
to show how certain estimation processes can be completely or effectively
described within a single-qubit subsystem; and to derive lower bounds
for the scaling of the estimation precision with the number of probes
used. We illustrate how the framework works for both noiseless and
noisy estimation procedures, in particular those based on multi-qubit
GHZ-states. Finally we succinctly analyze estimation protocols based
on zero-temperature critical behavior. We identify the coherence that
is at the heart of their efficiency, and we show how it exhibits the
non-analyticities and scaling behavior proper of a large class of
quantum phase transitions. 
\end{abstract}
\maketitle

\section{Introduction}

Precision in single parameter estimation processes can be strikingly
enhanced with the use of quantum probes \cite{Wineland_AtomicClock1,Wineland_AtomicClock2,Wineland_AtomicClock3}.
Therefore the search for new and increasingly efficient quantum estimation
schemes is at the basis of the development of several technologies,
and it is an arduous theoretical and experimental challenge \cite{ParisEstimationBook,Wisemanbook}.
Two paradigmatic examples are Quantum Phase Estimation (QPE) and Criticality-Enhanced
Quantum Estimation (CEQE). In the first case the goal is to determine
the phase $\lambda$ of a unitary evolution $e^{-i\lambda G}$ generated
by a fixed operator $G$. QPE is essential for several applications
such as interferometry \cite{Inteferometry1,Interferometry2,Interferometry3},
spectroscopy \cite{spectroscopy1,spectroscopy2}, magnetic sensing
\cite{Sensing1,Sensing2,Sensing3,Sensing4} and atomic clocks \cite{Clocks1,Clocks2}.CEQE
instead exploits the critical behavior of systems undergoing a quantum
phase transition (QPT) to drastically enhance the estimation precision
of the parameter driving the transition; the latter is in general
a dynamic parameter (such as a coupling constant) of a complex many-body
quantum system \cite{ZanardiCEQE,estimationcriticality,ParisDicke}.
In both cases the precision's ultimate bounds can be established by
means of quantum estimation\emph{ }theory \cite{ParisEstimationBook,Wisemanbook,BraunsteinCaves,HayashiInference}.

A fundamental open question is whether it is possible to identify
a single relevant resource underlying the optimal efficiency of all
these estimation tasks.

In answering this question one is led to consider different aspects.
First of all the estimation processes are dynamical, and one may expect
that rather than the static properties of the state probe, what matters
is their dynamical change. Secondly, since the probes are quantum,
one should focus on the prominent resources that distinguish quantum
from classical systems: coherence and correlations. The choice of
coherence is a natural and intuitive one: many estimation protocols
are indeed interference experiments \cite{Inteferometry1,Interferometry2,Interferometry3}.
But which is the \textit{relevant} coherence? And how to quantify
the latter in a consistent and general way?

On the other hand, as for QPE, many authors have focused on quantum
correlations \cite{Quantum enahnced measurements,MetrologyReview,Pezze}.
In particular, entanglement has been often indicated as the key to
achieve a better asymptotic scaling of the sensitivity with the number
of probes used. However, the relevance of quantum correlations can
be and has been questioned in different ways. For example, when the
estimation process is not affected by any noise, protocols based on
completely uncorrelated probes (e.g. multi-round single qubit protocols)
are able to reach the same sensitivity achieved by protocols based
on highly quantum correlated probes (such as GHZ multi-qubit states)
\cite{Quantum metrology,MetrologyReview,Multiroundexp,Multiroundtheory1,Multiroundtheory2,Multiroundtheory4,Mutliroundtheory3}.
This indicates that in the noiseless case quantum correlations are
not an intrinsic prerequisite for efficient QPE. Furthermore, the
presence of noise in the evolution of the system is in general highly
detrimental for the estimation processes, and the sole presence of
quantum correlations is not in general a sufficient condition to counteract
its effects and achieve an enhancement of the estimation sensitivity
\cite{Multiroundtheory4,Acin,AcinAsymptoticEntanglement}. Finally,
there is a more conceptual difficulty in identifying quantum correlations
as a key resource for estimation. Indeed, correlations are typically
defined once a specific tensor product structure (TPS) i.e., a factorization
of the Hilbert space $\mathcal{H}$ in subsystems, is chosen to describe
the whole system. But in general there can be many inequivalent TPSs
and correspondingly many different kinds of (multipartite) quantum
correlations; unless the problem at hand allows to identify a specific
TPS in a unique way, it is not clear which among the various possibilities
should be the relevant one. A possible way of ``ruling out'' the
relevance of certain (quantum) correlations could be the following.
Suppose there is at least one partition of the Hilbert space $\mathcal{H}=\mathcal{H}_{A}\otimes\mathcal{H}_{B}$
such that the whole estimation procedure can be described within the
subsystem $\mathcal{H}_{A}$ alone. Then one can trace out subsystem
$\mathcal{H}_{B}$ and (quantum) correlations between $\mathcal{H}_{A}$
and $\mathcal{H}_{B}$ never show up in the description of the whole
procedure. If the description of the process in the chosen TPS is
equivalent to the original one, then one may argue that those bipartite
correlations do not play any role for the estimation task, and one
has therefore to look for other resources at the basis of the estimation
sensitivity.

As for CEQE, on the other hand, the estimation protocols are based
on the occurrence of zero temperature transitions in the underlying
many-body system. While QPTs have been thoroughly characterized via
(bipartite/multipartite) quantum correlations \cite{entanglement&QPTs,entanglementcriticality,bipartiteentanglementQPT,entanglementdivergence1,entanglementdivergence2},
when one focuses on the use of criticality for enhancing the estimation
sensitivity the main concept of the theory is the information-geometric
notion of statistical distinguishability between neighboring ground
states, rather than correlations \cite{GuFidelityReview,ParisDicke,estimationcriticality,ZanardiCEQE}.

The notion of statistical distinguishability is ubiquitous in all
the above mentioned processes. Indeed, a unified description of both
noiseless/noisy QPE tasks and QPTs is provided by the powerful mathematical
language of information geometry \cite{Zyczkowski,HayashiInference,GuFidelityReview}
in terms of infinitesimal state discrimination. It is therefore natural
to look for connections between the bounds on the sensitivity achievable
in quantum estimation and some fundamental feature of the underlying
quantum system within the information geometry framework. The main
quantity that allows to connect geometry, estimation and QPTs is the
Quantum Fisher Information (QFI). The latter on one hand is proportional
to the statistical geometrical distance between neighboring quantum
states, and on the other hand it provides, via the Quantum Cramer
Rao theorem, the ultimate bounds on parameter estimation with quantum
probes.

In this work, we explore the existing connection between QFI and a
primary resource of quantum probes: \emph{coherence} \cite{speakablecoherence,wintercohernce,coherenceresource1,coherenceresource2,OurCoherence,zanardicoherence,MarvianSpekkens,frozencohernce,Piani,StreltsovGenuineCoherence,StreltsovReviewCoherence}.
In particular, we find a relation between the QFI and the \emph{curvature
of the coherence of the measurement basis that gives the optimal discrimination.
}\textit{\emph{Indeed, coherence is a basis dependent feature and
we show that }}\emph{ the relevant basis is given by the eigenvectors
of}\textit{ the main object of the Cramer-Rao approach to quantum
estimation: the Symmetric Logarithmic Derivative (SLD) $L$}\textit{\emph{.}}\emph{
}The relation found allows in the first place to highlight a possibly
new physical interpretation of the statistical-geometric distance
between infinitesimally close pure and mixed quantum states. On the
other hand, it allows to \textit{identify and quantify the relevant
resource that must be engineered, controlled, exploited and preserved
in QPE and CEQE protocols in order to achieve the highest possible
precision}. The main focus of our work will be QPE. In developing
our theoretical framework, we will show how it can be applied to instances
of both noiseless and noisy QPE protocols. In all treated cases the
relation between coherence and QFI holds independently of any Hilbert
space partition.However, we will argue that in many cases it is possible
to select in a unique way a proper TPS tailored to the problem at
hand. The relevant TPS ($TPS^{R}$) is again suggested by the Symmetric
Logarithmic Derivative.The factorization of the Hilbert space induced
by $L$ allows to neatly examine different aspects of estimation protocols.
In the first place it is possible to find a connection between the
QFI and bipartite classical (rather than quantum) correlations between
the set of observables that are relevant for the process \cite{MIMO}.\textsl{\emph{
Within our perspective the relation between the achievable estimation
precision and coherence vs correlations can be easily discussed. While
coherence is in general fundamental for the process, we will argue
that in many cases quantum correlations such as entanglement and discord
can be seen as irrelevant or even detrimental. Parallelly, we will
show that}} upon adopting the $TPS^{R}$, in many relevant examples
the whole estimation procedure can entirely or effectively be described
only within a subsystem. In such cases, whatever the dimension of
the original quantum system, the estimation process is seen to be
equivalent to a single qubit (multi-round) one. Notable examples are
procedures based on multi-qubit GHZ states or certain class of NOON
states \cite{Vogel}. In particular we will discuss why highly entangled
states, such as GHZ, can lead to a substantial enhancement of the
estimation sensitivity even in presence of some kind of noisy processes.
Finally, we will show that the description given by the $TPS^{R}$
may allow to easily derive meaningful lower bounds on the scaling
of the precision with the number of probes $M$ used or with the dimension
$N$ of the Hilbert space.

In the last part of our work we will turn our attention to CEQE. Here
we succinctly explore the consequences of the found link between information
geometry and coherence for the fidelity approach to Quantum Phase
Transitions, and for the estimation protocols based or zero-temperature
criticality. We show that the non-analiticities that characterize
the approach of a many-body system to a critical point, and that are
at the basis of CEQE, can be interpreted as the divergent behavior
of a specific coherence function.

The paper is organized as follows. In Section \ref{sec:QFI-and-coherence.}
we derive the relation between QFI and coherence. In Section \ref{sec:Coherence-in-phase}
we apply the relation found to two- and $N$-dimensional systems and
we introduce the $TPS^{R}$. With the aid of the latter we discuss
the role of coherence vs (quantum) correlations, and we analyze specific
relevant cases such as protocols based on GHZ and NOON states. In
Section \ref{sec:Coherence-and-noise} we discuss how our framework
can be applied to noisy estimation processes. In particular we thoroughly
discuss a relevant example based on multi-qubit GHZ states. In Section
\ref{sec:Coherence,-Quantum-Phase} we finally assess the role of
coherence in CEQE. For the sake of clarity and conciseness, the details
of the calculations leading to our main results can be found in the
Appendices.

\section{QFI and coherence.\label{sec:QFI-and-coherence.} }

The problem of identifying the ultimate precision in the estimation
of a given parameter $\lambda$ can be described within the Quantum
Crámer-Rao formalism. In this context, an unknown $\lambda\in\mathbb{R}$
parametrizes a family of quantum states $\rho_{\lambda}$ of an $N-$dimensional
quantum system. Given any (unbiased) estimator $\hat{\lambda}$ of
$\lambda$, the ultimate bound in terms of the variance of $\hat{\lambda}$
reads 
\begin{equation}
\Delta^{2}\hat{\lambda}\ge\left(M\ QFI\right)^{-1}
\end{equation}
where $M$ is the number of independent copies of $\rho_{\lambda}$
and 
\begin{equation}
QFI\left(\rho_{\lambda}\right)=Tr\left[\rho_{\lambda}L_{\lambda}^{2}\right]
\end{equation}
is the Quantum Fisher Information. The latter is expressed in terms
of the Symmetric Logarithmic Derivative (SLD) $L_{\lambda}$, the
Hermitian operator that satisfies the equation 
\begin{equation}
\partial_{\lambda}\rho_{\lambda}=\left(\rho_{\lambda}L_{\lambda}+L_{\lambda}\rho_{\lambda}\right)/2
\end{equation}
The bound in general can be attained by implementing an experiment
projecting the state onto the eigenbasis $\mathcal{B}_{\boldsymbol{\mathbf{\alpha}}}^{\lambda}=\left\{ \ket{\alpha^{\lambda}}\right\} _{\alpha=1}^{N}$
of $L_{\lambda}$. The $QFI$ is the maximum of the Fisher Information
$FI\left(\mathcal{B}_{\boldsymbol{\mathbf{x}}};\rho_{\lambda}\right)=\sum_{x}\left(\partial_{\lambda}p_{x}^{\lambda}\right)^{2}/p_{x}^{\lambda}$
over all possible experiments (orthonormal bases) $\mathcal{B}_{\boldsymbol{\mathbf{x}}}=\left\{ \ket{x}\right\} _{x=1}^{N}$,
where $p_{x}^{\lambda}=\left\langle x|\rho_{\lambda}|x\right\rangle $
is the probability of obtaining the outcome $x$: 
\begin{equation}
QFI(\rho_{\lambda})=\max_{\mathcal{B}_{\boldsymbol{\mathbf{x}}}}FI\left(\mathcal{B}_{\boldsymbol{\mathbf{x}}};\rho_{\lambda}\right)=FI\left(\mathcal{B}_{\boldsymbol{\mathbf{\alpha}}}^{\lambda};\rho_{\lambda}\right)
\end{equation}
The above formalism is common to all single-parameter quantum estimation
processes, and the $QFI$ was shown to bear a fundamental information-geometrical
meaning \cite{HayashiInference,Zyczkowski,BraunsteinCaves} since
it is proportional to the Bures metric $g_{\lambda}^{Bures}$: 
\begin{equation}
QFI(\rho_{\lambda})=4g_{\lambda}^{Bures}.
\end{equation}
$g_{\lambda}^{Bures}$ provides the infinitesimal geometric distance
$ds^{2}(\rho_{\lambda},\rho_{\lambda+\delta\lambda})$ between two
neighboring quantum states and their statistical distinguishability;
thus, it measures how well the states, and thus the parameter $\lambda$,
can be discriminated.

We now show that $QFI(\rho_{\lambda})$ can in general be connected
to the variation of the coherence of the basis $\mathcal{B}_{\boldsymbol{\mathbf{\alpha}}}^{\lambda}$
with respect to the state $\rho_{\lambda}$ when the latter undergoes
an infinitesimal change $\rho_{\lambda}\rightarrow\rho_{\lambda+\delta\lambda}$,
with $\delta\lambda\ll1$. In general, the coherence of given basis
$\mathcal{B}_{\boldsymbol{\mathbf{x}}}$ with respect to a state $\rho$
can be measured by the relative entropy of coherence \cite{coherenceresource1,OurCoherence}
\begin{equation}
Coh_{\mathcal{B}_{\boldsymbol{\mathbf{x}}}}(\rho)=-\mathcal{V}(\rho)+H\left(p_{x}\right)\label{eq:coherence_definition}
\end{equation}
where $\mathcal{V}(\rho)$ is the von Neumann entropy of $\rho$ and
$H\left(p_{x}\right)$ is the Shannon entropy of the probability distribution
$p_{x}=\left\langle x|\rho|x\right\rangle $. Now, it is well known
that the QFI can be expressed as 
\begin{equation}
QFI(\rho_{\lambda})=\left[\partial_{\delta\lambda}^{2}D(p_{\alpha}^{\lambda+\delta\lambda}|p_{\alpha}^{\lambda})\right]_{\delta\lambda=0}
\end{equation}
where $D(p_{\alpha}^{\lambda+\delta\lambda}|p_{\alpha}^{\lambda})$
is the relative entropy between the probability distributions $p_{\alpha}^{\lambda}=\left\langle \alpha^{\lambda}|\rho_{\lambda}|\alpha^{\lambda}\right\rangle $
and $p_{\alpha}^{\lambda+\delta\lambda}=\left\langle \alpha^{\lambda}|\rho_{\lambda+\delta\lambda}|\alpha^{\lambda}\right\rangle $;
the latter being the probabilities of the measurement defined by $\mathcal{B}_{\boldsymbol{\mathbf{\alpha}}}^{\lambda}$
realized on $\rho_{\lambda}$ and $\rho_{\lambda+\delta\lambda}$
respectively. The latter equation can be also written as

\begin{equation}
QFI(\rho_{\lambda})=\left[-\partial_{\delta\lambda}^{2}H(p_{\alpha}^{\lambda+\delta\lambda})+\partial_{\delta\lambda}^{2}\mathcal{X}(p_{\alpha}^{\lambda+\delta\lambda}|p_{\alpha}^{\lambda})\right]_{\delta\lambda=0}\label{eq: QFI Relative entropy}
\end{equation}
where $\mathcal{X}(p_{\alpha}^{\lambda+\delta\lambda}|p_{\alpha}^{\lambda})$
is the cross entropy of the two distributions. On the other hand,
from (\ref{eq:coherence_definition}) we obtain 
\begin{equation}
-\left[\partial_{\delta\lambda}^{2}Coh_{\mathcal{B}_{\boldsymbol{\mathbf{\alpha}}}^{\lambda}}(\rho_{\lambda+\delta\lambda})\right]_{\delta\lambda=0}=-\left[\partial_{\delta\lambda}^{2}H(p_{\alpha}^{\lambda+\delta\lambda})\right]_{\delta\lambda=0}+\left[\partial_{\delta\lambda}^{2}\mathcal{V}(\rho_{\lambda+\delta\lambda})\right]_{\delta\lambda=0}\label{eq: Coherence second derivative}
\end{equation}
By comparing (\ref{eq: QFI Relative entropy}) and (\ref{eq: Coherence second derivative})
we obtain the following
\begin{prop}
For a general estimation processes, the QFI is related to the second
order variation of the coherence of the Symmetric Logarithmic Derivative
eigenbasis $\mathcal{B}_{\boldsymbol{\mathbf{\alpha}}}^{\lambda}$
as: 
\begin{equation}
-\left[\partial_{\delta\lambda}^{2}Coh_{\mathcal{B}_{\boldsymbol{\mathbf{\alpha}}}^{\lambda}}(\rho_{\lambda+\delta\lambda})\right]_{\delta\lambda=0}=f(\rho_{\lambda})+QFI(\rho_{\lambda})\label{eq: QFI as second order variation of coherence}
\end{equation}
with $f(\varrho_{\lambda})=-\left[\partial_{\delta\lambda}^{2}\mathcal{X}(p_{\alpha}^{\lambda+\delta\lambda}|p_{\alpha}^{\lambda})\right]_{\delta\lambda=0}+\left[\partial_{\delta\lambda}^{2}\mathcal{V}(\rho_{\lambda+\delta\lambda})\right]_{\delta\lambda=0}$. 
\end{prop}
Proposition 1 is central in our analysis, as it establishes a link
between the optimal precision in an estimation process, the geometry
of quantum states and the coherence of the optimal measurement basis.
The found connection is further strengthened whenever $f(\rho_{\lambda})=0$
and $Coh_{\mathcal{B}_{\boldsymbol{\mathbf{\alpha}}}^{\lambda}}(\rho_{\lambda+\delta\lambda})$
has a critical point in $\delta\lambda=0$, so that the\textit{ $QFI$
can be expressed as the curvature of $Coh_{\mathcal{B}_{\boldsymbol{\mathbf{\alpha}}}^{\lambda}}(\rho_{\lambda+\delta\lambda})$
around a maximum}. It turns out that these conditions are always verified
for pure states. As proven in Appendix \ref{sec: Appendix Fubini-Study-metric-and}
for pure states $\rho_{\lambda}=|\psi_{\lambda}\rangle\langle\psi_{\lambda}|$
the following holds:
\begin{prop}
If $\rho_{\lambda}$ is pure, then

\begin{equation}
-\left[\partial_{\delta\lambda}^{2}Coh_{\mathcal{B}_{\boldsymbol{\mathbf{\alpha}}}^{\lambda}}(\rho_{\lambda+\delta\lambda})\right]_{\delta\lambda=0}=4g_{\lambda}^{FS}\label{eq: Cohernce vs Fubini Study -pure states}
\end{equation}
i.e., the second order variation of $-Coh_{\mathcal{B}_{\boldsymbol{\mathbf{\alpha}}}^{\lambda}}(\rho_{\lambda+\delta\lambda})$
is proportional to the Fubini-Study metric $g_{\lambda}^{FS}$, the
restriction of the Bures metric to the projective space $\mathcal{PH}$
of pure states. 
\end{prop}
Therefore $-\left[\partial_{\delta\lambda}^{2}Coh_{\mathcal{B}_{\boldsymbol{\mathbf{\alpha}}}^{\lambda}}(\rho_{\lambda+\delta\lambda})\right]_{\delta\lambda=0}$
determines the geometry of pure quantum states and the relative estimation
bounds that can be derived within the Cramer-Rao formalism. For pure
states $L_{\lambda}$ has only two nonzero eigenvalues, corresponding
to the eigenvectors $|\psi{}_{\lambda}\rangle$ and $\frac{d}{d\lambda}|\psi{}_{\lambda}\rangle$
. Thus the estimation process in fact happens in the single-qubit
space spanned by $|\psi{}_{\lambda}\rangle$ and $\frac{d}{d\lambda}|\psi{}_{\lambda}\rangle$,
and what matters is the change of the coherence in this subspace.
For mixed states, in general $f(\rho_{\lambda})\neq0$. However, as
we will prove below by means of specific examples, in many cases of
interest one has $f(\rho_{\lambda})\ll QFI(\rho_{\lambda})$ so that
the relation $-\left[\partial_{\delta\lambda}^{2}Coh_{\mathcal{B}_{\boldsymbol{\mathbf{\alpha}}}^{\lambda}}(\rho_{\lambda+\delta\lambda})\right]_{\delta\lambda=0}\thickapprox QFI(\rho_{\lambda})\approx4g_{\lambda}^{Bures}$
approximately holds and the variation of coherence is the leading
term that determines the geometry of mixed quantum states and the
relative estimation bounds.

We finally observe that Eq. (\ref{eq: QFI as second order variation of coherence})
and Eq. (\ref{eq: Cohernce vs Fubini Study -pure states}) are very
general since they hold whatever the process that induce the infinitesimal
change $\rho_{\lambda}\rightarrow\rho_{\lambda+\delta\lambda}$. In
the following Section we specialize our analysis to unitary phase
estimation processes in which $\rho_{\lambda+\delta\lambda}=U_{\delta\lambda}\rho_{\lambda}U_{\delta\lambda}^{\dagger}$,
where $U_{\delta\lambda}$ is a unitary operator. In Section \ref{sec:Coherence-and-noise}
we extend the discussion to some relevant non-unitary evolutions.
\\

\section{Coherence in phase estimation; pure and mixed probe states \label{sec:Coherence-in-phase}}

In this section we analyze the above results in the case of unitary
Quantum Phase Estimation processes where 
\begin{equation}
\rho_{\lambda}=\exp\left(-i\lambda G\right)\rho_{0}\exp\left(i\lambda G\right),
\end{equation}
$G$ is a Hermitian traceless operator, and the unknown phase $\lambda$
is the parameter to be estimated. In this case $QFI$ is independent
of $\lambda$, and it is sufficient to address the estimation problem
for $\lambda=0$ \cite{ParisEstimationBook}.

\subsection{The single-qubit case}

We start by analyzing the single-qubit case in which $\rho_{0}$,
$G$ and the generic measurement basis $\mathcal{B}$ can be defined
in the Bloch sphere formalism in terms of the vectors $\vec{z},\hat{\gamma},\hat{b}$
respectively as follows. Without loss of generality we choose the
single qubit state 
\begin{equation}
\varrho_{0}=(1+\vec{z}\cdot\boldsymbol{\sigma})/2
\end{equation}
where $\vec{z}=z\hat{z}=z(0,0,1)$, $0\le z\le1$ and $\boldsymbol{\sigma}=\left(\sigma_{x},\sigma_{y},\sigma_{z}\right)$
is the vector of Pauli matrices; and the phase generator 
\begin{equation}
G=\gamma\left(\hat{\gamma}\cdot\boldsymbol{\sigma}\right)
\end{equation}
with $\hat{\gamma}=(\sin\delta,0,\cos\delta)$, such that its eigenbasis
lies in the $\hat{x}\hat{z}$ plane, forming an angle $0\leq\delta\leq\frac{\pi}{2}$
with $\hat{z}$. \\
A generic measurement basis $\mathcal{B}_{\hat{b}}$ is defined by
the projectors $\Pi_{\pm}^{\hat{b}}=(1\pm\hat{b}\cdot\boldsymbol{\sigma})/2$
with $\hat{b}=\{\sin\theta\cos\phi,\sin\theta\sin\phi,\cos\theta\}$.
For a mixed state ($z<1$) the FI in $\lambda=0$ is given by (for
the proof, see appendix \ref{sec:Appendix II -Results for N=00003D00003D2}):
\begin{equation}
FI(\mathcal{B}_{\hat{b}},\rho_{0},G)=4\frac{(\vec{\gamma}\times\vec{z}\cdot\hat{b})^{2}}{1-(\vec{z}\cdot\hat{b})^{2}}.\label{eq: Fisher Qbit general measurement}
\end{equation}
The maximisation of $FI$ over the measurement basis has a unique
solution (\cite{BraunsteinCaves}) and leads to the choice of the
eigenbasis $\mathcal{B}_{\hat{\alpha}}$ of the SLD, that in our case
corresponds to choosing $\hat{b}=\hat{\alpha}=\{0,1,0\}\propto\hat{\gamma}\times\hat{z}$.
As for the coherence of $\mathcal{B}_{\hat{\alpha}}$ one has that
\begin{equation}
\left[\partial_{\lambda}Coh_{\mathcal{B}_{\alpha}}(\rho_{\lambda})\right]_{\lambda=0}=0
\end{equation}

\begin{equation}
FI(\mathcal{B}_{\alpha},\rho_{0})=QFI(\rho_{0},G)=-\left(\partial_{\lambda}^{2}Coh_{\mathcal{B}_{\alpha}^{0}}(\rho_{\lambda})\right)_{\lambda=0}
\end{equation}
and therefore one obtains result (\ref{eq: QFI as second order variation of coherence}),
with $f(\rho_{\lambda})=0$ (proof in Appendix \ref{sec:Appendix II -Results for N=00003D00003D2}).
Since $\mathcal{B}_{\hat{\alpha}}$ is unique one has that \emph{for
single qubit mixed states }\textit{the necessary and sufficient condition
for attaining the Cramer-Rao bound is the maximization of} \textit{the
coherence of the measurement basis with respect to the state $\rho_{0}$,
and the QFI coincides with its second order variation. }\\
The optimization of $QFI(\rho_{0},G)$ with respect to $G$ leads
to the choice of $\hat{\gamma}\cdot\hat{z}=0$ that corresponds to
the maximization of the coherence of the eigenbasis of $G$ with respect
to $\rho_{0}$, as has been highlighted in \cite{zanardicoherence,speakablecoherence,MarvianSpekkens}.
Taking into account both maximizations, one has $QFI=2z^{2}Tr[G^{2}]$.
Therefore our treatment allows to interpret the optimal estimation
procedure as the one that takes advantage of the full strength of
$G$ in order to variate the relevant coherence i.e., that of the
basis $\mathcal{B}_{\alpha}$. \\
For pure states the measurements axes $\hat{b}$ leading to $QFI$
are not unique. However, $\hat{b}=\hat{\alpha}$ is the only choice
that allows to attain the highest sensitivity in $\lambda$, and at
the same time the lowest sensitivity with respect to small changes
in the measurment angles $\delta\theta,\delta\phi$, possibly due
to imperfections of the measurement apparatus, or to the impurity
of the initial state (see Appendix \ref{sec:Appendix II -Results for N=00003D00003D2}).

\subsection{The $N$-dimensional case\label{sub: QFI and Coherence in  N dimensional-case}}

We now pass to analyze the general case $N=\dim\mathcal{H}>2$, where
there is room for discussing the role of coherence vs correlations
in estimation processes such as QPE. In order to do so we need to
find a direct connection between $QFI$ and the correlations relevant
for the estimation. In this subsection we first describe the main
points of our approach and give general formal results. We will then
illustrate the results by means of specific examples in Sections \ref{sub: Example: states with Max QFI},
\ref{sub: Main Example discordant} and \ref{sub: Main Examples:-GHZ-states}.
\\

The relation between coherence and the $QFI$ given by (\ref{eq: QFI as second order variation of coherence})
holds for any $N$ and irrespectively of any \emph{local} structure
of the given Hilbert space $\mathcal{H}_{N}$. Correlations, instead,
are typically defined between subsystems i.e., when a specific tensor
factorization of the Hilbert space is chosen. An $N-$dimensional
Hilbert space in general admits several (possibly inequivalent) factorizations
in subsystems $\mathcal{H}_{n_{i}}$ 
\[
\mathcal{H}_{N}=\otimes_{_{i}}\mathcal{H}_{_{i}}
\]
with $\Pi_{i}n_{i}=N,\ n_{i}=dim\mathcal{H}_{i}$. Accordingly different
definitions of bi- or multi-partite (quantum) correlations are possible.
The different factorizations are called \emph{tensor product structures}
(TPS). Which decomposition is ``relevant'' or useful should in general
be suggested by the problem at hand.

The first step of our approach is therefore to identify such relevant
factorization. It turns out that, under some hypotheses, for QPE one
can use the eigendecomposition of $L_{0}$ to uniquely identify a
factorization of the Hilbert space into the product of a single qubit
and an $N/2-$dimensional subspace 
\begin{equation}
\mathcal{H}_{N}\sim\mathcal{H}_{2}\tilde{\otimes}\mathcal{H}_{N/2}\label{eq:referenceTPS}
\end{equation}
We will refer to (\ref{eq:referenceTPS}) as the \textit{reference
TPS} , indicate it with $TPS^{R}$ and, for sake of clarity, use for
the corresponding tensor product operator the symbol $\tilde{\otimes}$
to distinguish it from other TPS (e.g., the standard TPS on $M$ qubits).
As we will show below, in $TPS^{R}$, the SLD can be written as $L_{0}=O_{2}\tilde{\otimes}O_{N/2}$
where $O_{2}$, $O_{N/2}$ are operators acting locally on $\mathcal{H}_{2}$
and $\mathcal{H}_{N/2}$ respectively. Therefore the eigenvectors
of $O_{2}$, $O_{N/2}$ form a product basis $\mathcal{B}_{2}\tilde{\otimes}\mathcal{B}_{N/2}$
and the relative projectors $\left\{ \Pi_{\pm}\otimes\Pi_{k}\right\} ,\ k=1,..,N/2$
define: a global von Neumann experiment on $\mathcal{H}_{N}$, whose
outcomes are distributed according a joint probability distribution
$p_{\pm,k}^{\lambda}=Tr\left[\Pi_{\pm}\tilde{\otimes}\Pi_{k}\rho_{\lambda}\right]$;
and local experiments with outcomes distributed according to the marginals
$p_{\pm}^{\lambda}=Tr\left[\Pi_{\pm}\tilde{\otimes}\mathbb{I}_{N/2}\rho_{\lambda}\right]$
and $p_{k}^{\lambda}=Tr\left[\mathbb{I}_{2}\tilde{\otimes}\Pi_{k}\rho_{\lambda}\right]$.
\\

The second step of our approach is based on the use of a relation
between coherence and (classical) correlations that can be found by
applying the definition of coherence (\ref{eq:coherence_definition})$ $$ $
to the case of product bases \cite{MIMO}. In particular for $\mathcal{B}_{2}\tilde{\otimes}\mathcal{B}_{N/2}$
the coherence function can be written as :

\begin{eqnarray}
Coh_{\mathcal{B}_{\pm,k}}(\rho_{\lambda}) & = & -\mathcal{V}(\rho_{\lambda})+H(p_{\pm}^{\lambda})+H(p_{k}^{\lambda})-\mathcal{M}_{L_{0}}^{\lambda}\left(p_{\pm,k}^{\lambda}\right)\label{eq: Coh in terms of MI}
\end{eqnarray}
where 
\begin{equation}
\mathcal{M}_{L_{0}}^{\lambda}\left(p_{\pm,k}^{\lambda}\right)=H\left(p_{\pm}^{\lambda}\right)+H\left(p_{k}^{\lambda}\right)-H\left(p_{\pm,k}^{\lambda}\right)\label{eq:definition of mutualinfo}
\end{equation}
is the classical mutual information for the probability distribution
$p_{\pm,k}^{\lambda}$. We notice that expression (\ref{eq: Coh in terms of MI})
was used in \cite{MIMO} where it was shown that the efficiency of
a communication protocol such as remote state preparation requires
the maximization of the correlations between some relevant observables
i.e., the maximization of the relative mutual information. For certain
kinds of two-qubit states, relation (\ref{eq: Coh in terms of MI})
expresses a general trade-off between correlations and coherence:
when the former is maximized the latter is correspondingly minimized.
For remote state preparation the relation is between static resources
stored in the system state. Here we will see that in dynamical processes
such as quantum estimation protocols the relation is between the changes
of correlations and coherence. 

Let us now proceed to derive the general results that allow to define
the $TPS^{R}$ and to state the relations between $QFI$, coherence
and correlations that follows from (\ref{eq: Coh in terms of MI}).
The $TPS^{R}$ construction builds on the properties of the eigendecomposition
of $L_{0}$ when some hypotheses on $\rho_{0},G,L_{0}$ are satisfied: 
\begin{prop}
Under the following hypotheses: $i)$ $N$ is even; $ii)$ the initial
diagonal state $\rho_{0}=\sum_{n}p_{n}\ket{n}\bra{n}$ is full rank;
$iii)$ $\langle n|G|m\rangle\in\mathbb{R}\ \forall n,m$ i.e., $G$
has purely real matrix elements when expressed in the eigenbasis of
$\rho_{0}$; $iv)$ and $L_{\lambda=0}$ is full rank then: 

3.1 $L_{0}$ is diagonal in a basis 
\begin{equation}
\mathcal{B}_{\boldsymbol{\mathbf{\alpha}=}(\pm,k)}=\{\ket{\alpha_{i,k}}\},\qquad i=\pm,k=1,..,N/2\label{eq:SLDbasis-1}
\end{equation}
with eigenvalues that are opposite in pairs 
\begin{equation}
\alpha_{\pm,k}=\pm\alpha_{+,k}\in\mathbb{R}\backslash\left\{ 0\right\} \label{eq:oppositeEigenvalues-1}
\end{equation}

3.2 The Hilbert space can be decomposed as $\mathcal{H}_{N}=\mathcal{H}_{2}\tilde{\otimes}\mathcal{H}_{N/2}$;
the eigenvectors of the SLD can be written as \textup{ 
\begin{equation}
\ket{\alpha_{i,k}}=\ket{i}\tilde{\otimes}\ket{k},\qquad i=\pm,k=1,..,N/2
\end{equation}
} and the SLD in its diagonal form can be written as 
\begin{equation}
L_{0}=S_{y}\tilde{\otimes}\sum_{k=1,..,N/2}\alpha_{+}^{k}\Pi_{k}\label{eq: SLD TPSR}
\end{equation}
where $S_{y}=\Pi_{+}-\Pi_{-}$ is a Pauli matrix acting locally on
the single qubit sector $\mathcal{H}_{2}$ and $O{}_{N/2}=\sum_{k=1,..,N/2}\alpha_{+}^{k}\Pi_{k}$
is an operator that depends on the eigevalues of $L_{0}$ and acts
locally onsubsystem $\mbox{\ensuremath{\mathcal{H}}}_{N/2}$ 
\end{prop}
The proof of $3.1)$ is given in the Appendix \ref{sec Appendix :SLD-and-coherence}
while the proof of $3.2)$ is given in Appendix \ref{sec: Appendix SLD-induced-TPS-for N dimensional states}.
The first result depends on the fact that under the stated hypotheses
$L_{0}$ is a Hermitian anti-symmetric operator. Result $3.2)$ relies
on the fact that a possible way to induce a TPS is based on the observables
of the system \cite{Virtual,ObservableInduced}. Indeed, suppose one
has a set of sub-algebras of Hermitian operators $\mathcal{A}_{2}$,
$\mathcal{A}_{N/2}$ satisfying the following conditions: $i)$ commutativity,
i.e., $\left[\mathcal{A}_{2},\mathcal{A}_{N/2}\right]=0$; $ii)$
completeness, i.e. the product of the observables belonging to $\mathcal{A}_{2},\mathcal{A}_{N/2}$
allows to generate the full set of Hermitian operators over $\mathcal{H}_{N}$.
Then $\mathcal{A}_{2}\vee\mathcal{A}_{N/2}\cong\mathcal{A}_{N}$ and
one can induce a factorization $\mathcal{H}_{N}=\mathcal{H}_{2}\tilde{\otimes}\mathcal{H}_{N/2}$
such that: each $O_{2}\in\mathcal{A}_{2}$, $O_{N/2}\in\mathcal{A}_{N/2}$
acts locally on $\mathcal{H}_{2}$ and $\mathcal{H}_{N/2}$ respectively;
the composition (product) of operators can be written as $O_{2}O_{N/2}\left(\mathcal{H}_{2}\tilde{\otimes}\mathcal{H}_{N/2}\right)=O_{N/2}O_{2}\left(\mathcal{H}_{2}\tilde{\otimes}\mathcal{H}_{N/2}\right)=\left(O_{2}\mathcal{H}_{2}\right)\tilde{\otimes}\left(O_{N/2}\mathcal{H}_{N/2}\right)$.
In Appendix \ref{sec: Appendix SLD-induced-TPS-for N dimensional states}
we show how the algebras $\mathcal{A}_{2},\mathcal{A}_{N/2}$ can
be explicitly constructed by taking appropriate sums of projectors
$\Pi_{\pm,k}=|\alpha_{i,k}\rangle\langle\alpha_{i,k}|$ onto the eigenbasis
$\mathcal{B}_{\boldsymbol{\mathbf{\alpha}=}(\pm,k)}$ given by result
$3.1)$.

Having defined $TPS^{R}$, we can now enounce the main result of this
section. The relation between $QFI$, coherence and $\mathcal{M}_{L_{0}}^{\lambda}$
can be stated in the following way : 
\begin{prop}
Under the same hypotheses of Proposition 3 one has: 

4.1 The relation $QFI=\left[-\partial_{\lambda}^{2}Coh_{\mathcal{B}_{\pm,k}}(\rho_{\lambda})\right]_{\lambda=0}+f(\rho_{\lambda})$
is attained in correspondence of a critical point of $Coh_{\mathcal{B}_{\pm,k}}(\rho_{\lambda})$
i.e., $\left[\partial_{\delta\lambda}Coh_{\mathcal{B}_{\pm,k}}(\rho_{\lambda})\right]_{\delta\lambda=0}=0;$ 

4.2 $\mathcal{M}_{L_{0}}^{\lambda=0}=0$ and $\left(\partial_{\lambda}\mathcal{M}_{L_{0}}^{\lambda}\right)_{\lambda=0}=0$
i.e., the observables $S_{y}$ and $O_{N/2}$ are uncorrelated for
$\lambda=0$, and therefore $\lambda=0$ is a minimum for $\mathcal{M}_{L_{0}}^{\lambda}$
and $\left(\partial_{\lambda}^{2}\mathcal{M}_{L_{0}}^{\lambda}\right)_{\lambda=0}\ge0$ 

4.3 the QFI can be written as 
\begin{eqnarray}
QFI & = & FI_{2}+\left(\partial_{\lambda}^{2}\mathcal{M}_{L_{0}}^{\lambda}\right)_{\lambda=0}\label{eq: QFI in terms of partial2 MI}
\end{eqnarray}
where $FI_{2}=\left[\sum_{i=\pm}\frac{\left(\partial_{\lambda}p_{i}\right)^{2}}{p_{i}}\right]_{\lambda=0}$ 

4.4 Given the single qubit reduced density matrix $\xi_{\lambda}=Tr_{\mathcal{H}_{N/2}}\left[\rho_{\lambda}\right]$
one has 
\begin{eqnarray*}
FI_{2}\le & QFI\left(\xi_{\lambda}\right) & \le QFI
\end{eqnarray*}

\end{prop}
The proof is given in the last part of Appendix \ref{sec: Appendix SLD-induced-TPS-for N dimensional states}.
Proposition 4 allows to give a new interpretation of the $QFI$ and
has the following several different consequences. \\
Result $4.1)$ shows that the Quantum Cramer-Rao bound is achieved
in correspondence of a critical point of the coherence of the eigenbasis
of the SLD with respect to $\rho_{\lambda}$. This typically corresponds
to a\textit{ maximum.} On the other hand, result $4.2)$ shows that
the correlations between the relevant observables defined by $L_{0}=S_{y}\tilde{\otimes}O_{N/2}$
are minimized. These results mirror the general trade-off between
correlations and coherence mentioned above \cite{MIMO}. Here the
variation of coherence is maximized in correspondence of a minimum
of the correlations between the relevant observables. As we discuss
below and in the following examples the minimization of correlations,
and in particular their complete absence, has some relevant consequences
for the representation of the estimation procedure and its efficiency.
 \\
As for result $4.3)$, Equation (\ref{eq: QFI in terms of partial2 MI})
shows that the $QFI$ can be expressed in terms of two contributions.
The first term $FI_{2}=\left[\sum_{i}\left(\partial_{\lambda}p_{i}\right)^{2}/p_{i}\right]_{\lambda=0}$
is the Fisher information of a single qubit. Indeed, since 
\begin{eqnarray*}
p_{\pm}^{\lambda} & = & Tr_{\mathcal{H}_{N}}\left[\Pi_{\pm}\otimes\mathbb{I}_{N/2}\rho\right]=Tr_{\mathcal{H}_{2}}\left[\Pi_{\pm}\xi_{\lambda}\right]
\end{eqnarray*}
$FI_{2}$ is the Fisher Information corresponding to the measurement
of the local $S_{y}$ onto the reduced density matrix $\xi_{\lambda}=Tr_{\mathcal{H}_{N/2}}\left[\rho_{\lambda}\right]$.
The other term is given by the second order variation of the correlations
$\mathcal{M}_{L_{0}}^{\lambda}=\mathcal{M}\left(p_{\pm,k}^{\lambda}\right)$
between the observables defined by the eigendecompostion of $L_{0}$
via $TPS^{R}$. We notice that on one hand the connection found is
between the $QFI$ and specific\textit{ classical correlations} rather
then quantum ones. On the other hand, since the estimation process
is a dynamical one, the connection involves a variation of those correlations
with $\lambda$. Indeed, since $\mathcal{M}\left(p_{\pm,k}^{\lambda=0}\right)=0$,
the relevant correlations have a minimum in $\lambda=0$, and the
efficiency of the estimation protocol depends on the ``acceleration''
with which those correlations are changed by the unitary evolution
that impresses the phase onto the state. \\
The last result $4.4)$ derives on one hand from the fact that in
general the basis defined by $S_{y}$ \textit{\emph{does not correspond}}
to eigenbasis of the SLD for $\xi_{\lambda}$, therefore in general
$FI_{2}\le QFI(\xi_{\lambda})$; and on the other hand from the fact
that an estimation process realized on a subsystem gives in general
a lower precision than an estimation realized on the whole system,
hence $QFI(\xi_{\lambda})\le QFI$. The relevance of result $4.4)$
stems from the fact that, if one is able to evaluate $\xi_{\lambda}$,
then by working on a single qubit system, independently on the dimension
$N$, one can easily found a lower bound on the $QFI$ i.e., $QFI(\xi_{\lambda})\le QFI$.
This becomes quite relevant whenever one is interested in evaluating
the scaling behavior of $QFI$ with $N$. We will see an example of
how property $4.4)$ can be fruitfully used in Section \ref{sub: Non-commuting-noise},
where we show that for certain noisy estimation processes based on
GHZ-sates, $\xi_{\lambda}$ can be evaluated and the scaling behavior
of $QFI$ can be deduced by means of $QFI(\xi_{\lambda})$.

We close this subsection with a few comments. Overall, the decomposition
given by (\ref{eq: QFI in terms of partial2 MI}), allows to unambiguously
express the $QFI$ in terms of a single qubit Fisher Information and
classical correlations. This result is particularly interesting whenever
$\left(\partial_{\lambda}^{2}\mathcal{M}_{L_{0}}^{\lambda}\right)_{\lambda=0}=0$,
so that $QFI=FI_{2}=QFI(\xi_{\lambda})$. In such cases the $TPS^{R}$
construction shows that the estimation process is effectively a single-qubit
one, whatever the dimension $N$ of the original Hilbert space; and
no (quantum) correlations in the probe state are involved in the process.
This turns out to be the case in some relevant examples we discuss
in the next section. 

As a final remark, let us discuss the generality of the results obtained
in Proposition 3 and 4. We notice that although seemingly restrictive,
the hypotheses stated in Proposition 3 are actually quite general.
Indeed, on one hand $N$ even includes all multi-qubit states. On
the other hand, as it can be seen from the following general expression
for $QFI$ (\cite{ParisEstimationBook}) 
\[
QFI=2\sum_{n\neq m}\frac{(p_{n}-p_{m})^{2}}{(p_{n}+p_{m})}\left|\langle n|G|m\rangle\right|^{2}
\]
where $\left\{ p_{n},\ket{n}\right\} $ are the eigenvalues and eigenvectors
of $\rho_{0}$, for each Hemitian operator $\tilde{G}\in\mathbb{C}^{N}\times\mathbb{C}^{N}$
one can always find a corresponding $G$ such that $\langle n|G|m\rangle=|\langle n|\tilde{G}|m\rangle|\in\mathbb{R}\ \forall n,m$,
and $Tr\left[\tilde{G}^{2}\right]=Tr\left[G^{2}\right]$, and such
that both operators have the same $QFI$. Therefore, as for the estimation
process, the choice of $\tilde{G}$ or $G$ is equivalent and by using
$G$ our results can be applied. As for the requirement that $L_{0}$
is full rank, it can be easily relaxed as it will be shown in the
specific examples below.

\subsection{Examples: states with maximal QFI\label{sub: Example: states with Max QFI}}

We first focus on the situation in which given $\rho_{0}$ one seeks
the operator with fixed norm $\gamma,$ 
\[
G\in\mathcal{O}_{\gamma}=\left\{ G\ |\ Tr[G^{2}]=2\gamma^{2}\right\} 
\]
that allows to obtain the best estimation precision i.e., $QFI_{Max}=\max_{G\in\mathcal{O}_{\gamma}}(\rho_{0},G)$.
We now state the main results, while the detailed calculations are
reported in Appendix \ref{sec: Appendix Examples N dimensional-mixed states Maximal QFI}.
Suppose $\rho_{0}=\sum_{n=1}^{N}p_{n}\ket{n}\bra{n}$ such that its
eigenvalues are arranged in decreasing order, then the optimal operator
reads 
\[
\bar{G}=G_{1N}\left(\ket{1}\bra{N}+\ket{N}\bra{1}\right)
\]
and 
\[
QFI_{Max}=4\gamma^{2}\frac{(p_{1}-p_{N})^{2}}{(p_{1}+p_{N})}
\]
In this case $L_{0}$ is not full-rank, so the basis $\mathcal{B}_{\mathbf{\alpha=\pm},k}$
and the relative $TPS^{R}$ are no longer unique. However, for all
choices of bases $\mathcal{B}_{\mathbf{\alpha=\pm},k}$, Eq. (\ref{eq: QFI as second order variation of coherence})
holds with $f(\rho_{\lambda=0})=0$, and the maximal $QFI$ can be
expressed in terms of the variation of the coherence of any of such
bases as: 
\begin{eqnarray*}
QFI_{Max} & = & -\left[\partial_{\lambda}^{2}Coh_{\mathcal{B}_{\alpha}}(\rho_{\lambda})\right]_{\lambda=0}=4g_{\lambda}^{Bures}
\end{eqnarray*}
As for the decomposition defined in Proposition 4, $QFI_{Max}$ can
be expressed in terms of two contributions whose values, due to the
non-uniqueness of $\mathcal{B}_{\mathbf{\alpha=\pm},k}$ depend on
the specific basis chosen. For all choices of basis, one has that
\begin{eqnarray*}
FI_{2} & \ge & QFI_{Max}\cdot(p_{1}+p_{N})\\
\left(\partial_{\lambda}^{2}\mathcal{M}_{L_{0}}^{\lambda}\right)_{\lambda=0} & \leq & QFI_{Max}\cdot(1-p_{1}-p_{N})
\end{eqnarray*}
In general, when the initial state is mixed, the action of $\bar{G}$
which is relevant for the estimation is partially to change a single
qubit coherence, as measured by $FI_{2}$, and partially to change
the relevant correlations such that $\left(\partial_{\lambda}^{2}\mathcal{M}_{L_{0}}^{\lambda}\right)_{\lambda=0}\neq0$.
In the pure state limit $p_{1}\rightarrow1,p_{N}\to0$, the contribution
of the correlations drops to zero, $FI_{2}\rightarrow QFI$ and one
recovers the result (\ref{eq: Cohernce vs Fubini Study -pure states}).
For pure states the estimation process is therefore in essence a single-qubit
one, where the action of $\bar{G}$ that is relevant for the estimation
is only focused to change  a single qubit coherence, while the correlations
disappear from the picture; and the  $QFI_{Max}=4\gamma^{2}$  is
the maximal $QFI$ one can obtain at fixed $N$ and $Tr\left[G^{2}\right]$.

\subsection{Examples: class of separable states \label{sub: Main Example discordant}}

We now introduce a family of states that allow us to discuss different
aspects of estimation protocols. On one hand we identify coherence
as the relevant resource for the estimation and we discuss the possible
role of different kinds of quantum correlations such as entanglement
and discord \cite{DiscordReveiw}. On the other hand the following
analysis will also give us the opportunity of studying estimation
protocols based on mixtures of GHZ and NOON states and to lay the
foundation for our later discussion about noisy estimation processes
based on GHZ probes (section \ref{sub: Non-commuting-noise}).

The states we focus on can be written as 
\begin{equation}
\rho_{0}=\sum_{k=1}^{N/2}p_{k}\tau_{k}\otimes\Pi_{k}\label{eq: General discordant state}
\end{equation}
defined by $\left\{ \tau_{k},p_{k}\right\} _{k=1}^{N/2}$, where 
\[
\tau_{k}=\left(\mathbb{I}+\vec{n}_{k}\cdot\vec{\sigma}\right)/2,\qquad\vec{n}_{k}=h_{k}(\cos\phi_{k},\sin\phi_{k},0)
\]
are single qubit states, with $\vec{n}_{k}=h_{k}(\cos\phi_{k},\sin\phi_{k},0)$
the relative Bloch vector lying in the $xy$ plane and $\left|\vec{n}_{k}\right|=h_{k}\le1\ \forall k$.

We now choose as shift operator $G=\sigma_{z}\otimes\mathbb{I}_{N/2}$.
In order to identify the resource that is relevant for the estimation
of $\lambda$ we can work in the original TPS in which the state is
defined and later discuss the results in $TPS^{R}$ (details of the
calculations can be found in Appendix \ref{sec: Appendix Examples N dimensional-mixed states Maximal QFI}.
Since the states are block diagonal, the SLD simply reads $ $$L=\oplus_{k}L_{k}$
\cite{Wang SLD and BlockDiagonalStates} with single-qubit contributions
given by $L_{k}=2i\eta\hat{\alpha}_{k}\cdot\bar{\sigma}$ and $\hat{\alpha}_{k}=\hat{n}_{k}\times\hat{z}$.
One has that for such states $f(\rho_{\lambda})=0$ and 
\begin{eqnarray}
QFI & = & \sum_{k}p_{k}QFI_{k}\nonumber \\
QFI_{k} & = & -\partial_{\lambda}^{2}Coh_{\mathcal{B}_{k}}(\tau_{k})=4h_{k}^{2}\label{eq: QFI discordant states}
\end{eqnarray}
i.e., the $QFI$ \textit{is the weighted sum of the second order variations
of single qubit coherences}. Therefore the estimation process is in
essence a single qubit one and \textit{it can be seen as an estimation
process carried over in parallel on $N/2$ qubits}. \\
In general $\vec{n}_{k}\neq\vec{n}_{h},\ h\neq k$, the state $\rho_{0}$,
with respect to the original TPS has zero bipartite entanglement but
non-zero discord. However, the essential resource is coherence rather
than discord (as also suggested in \cite{DiscordPhaseEstimation}).
Indeed, if one fixes the value of $QFI$, the latter can be achieved
by a whole class of states defined by $\left\{ \tau_{k},p_{k}\right\} _{k=1}^{N/2}$,
with different $N$, different purities, very different amount of
quantum correlations (as measured by the discord) between the single
qubit and the $N/2-$dimensional system; and most importantly\textit{
irrespectively of the latter}. Moreover, suppose one compares two
states $\rho_{1}$ and $\rho_{2}$ such that they differ only for
the direction of the Block vector $\vec{n}_{h}$ pertaining to a given
$\tau_{h}$. Suppose for example that in $\rho_{2}$, $\vec{n}_{h}$
does not lye in the $xy$ plane, then $\vec{n}_{h}\cdot\hat{z}\neq0$
and $QFI(\rho_{2})\le QFI(\rho_{1})$ (see Appendix \ref{sec: Appendix Examples N dimensional-mixed states Maximal QFI}):
the presence of the kind of discord implied by this choice of $\vec{n}_{h}$
in $\rho_{2}$ would therefore be \textit{detrimental} for the estimation
process.

As for the interpretation of the result in the reference $TPS^{R}$
we restrict to the simple case in which all single qubit states $\tau_{k}$
are pure. Given the eigenvectors $\ket{\alpha_{\pm,k}}$ corresponding
to each $L_{k}$ one can define 
\begin{eqnarray*}
\ket{\alpha_{\pm,k}} & = & \ket{\pm}\tilde{\otimes}\ket{k}.
\end{eqnarray*}
and one has that in general $\tilde{\otimes}\neq\otimes$ i.e., the
$TPS^{R}$ induced by $L_{0}$ is in general different from the original
TPS. In order to derive the decomposition (\ref{eq: QFI in terms of partial2 MI})
it is sufficient to evaluate $p_{\pm,k}^{\lambda}=\bra{\alpha_{\pm,k}}\rho_{\lambda}\ket{\alpha_{\pm,k}}=\bra{\pm}\tilde{\otimes}\bra{k}\rho_{\lambda}\ket{\pm}\tilde{\otimes}\ket{k}$;
then one can easily find (Appendix \ref{sec: Appendix Examples N dimensional-mixed states Maximal QFI})
that for this general class of states $\left(\partial_{\lambda}^{2}\mathcal{M}_{L_{0}}^{\lambda}\right)_{\lambda=0}=0$
so that $FI_{2}=QFI$ and no correlations are involved in the estimation
process.

In the following we will focus on two subclasses of states of the
type (\ref{eq: General discordant state}) that have been proposed
for use in phase estimation setups. \\

\subsection{Examples: GHZ states\label{sub: Main Examples:-GHZ-states}}

Pure GHZ states are a prototypical case often presented in the literature
\cite{MetrologyReview,Quantum metrology}, in which the precision
in the estimation of $\lambda$ is shown to have an Heisenberg scaling
\cite{MetrologyReview}, and in which the role of entanglement has
been often discussed. GHZ-like states are of the kind 
\begin{equation}
\ket{GHZ_{k}^{\pm}}=\left(\ket{k}_{M}\pm\ket{\bar{k}}_{M}\right)/\sqrt{2}
\end{equation}
where $\ket{k}_{M}\equiv\ket{k_{M},..,k_{1}},\ k_{i}\in\{0,1\}$ is
an $M$-qubit state, $k=\sum_{i=1}^{M}2^{(i-1)}k_{i}$, $k=0,..,2^{M-1}-1$
and $\ket{\bar{k}}_{M}=\ket{\bar{k}_{M},..,\bar{k}_{1}}$ represents
the binary logical negation of $k$. The set $\{\ket{GHZ_{k}^{\pm}}\}$
forms the GHZ-basis for $\mathcal{H}_{2^{M}}$. In order to apply
our framework we discuss the following general mixed state

\begin{equation}
\sum_{k=0}^{N/2-1}p_{k}\ket{GHZ_{k}^{+}}\bra{GHZ_{k}^{+}}\label{eq: GHZ diagonal mixed state}
\end{equation}
The state is a mixture of the GHZ basis states $\ket{GHZ_{k}^{+}}$
and for $p_{0}=1$ one has the typical example of pure GHZ state $\ket{GHZ_{0}^{+}}=\left(\ket{00...0}+\ket{11...11}\right)/\sqrt{2}$
discussed in the literature. As shift operator one typically chooses
$U_{\lambda}=\exp-i\lambda G$ with $G=\sum_{h=1}^{M}\sigma_{z}^{h}$.
The above class of states allows us to develop our treatment of the
estimation problem directly in the $TPS^{R}$ representation, to relate
the $QFI$ to the relevant coherence variations and to analyze its
decomposition according to (\ref{eq: QFI in terms of partial2 MI}).

We start our analysis by noticing that $G\ket{GHZ_{k}^{+}}=\left(M-2|k|\right)\ket{GHZ_{k}^{-}},$
thus the evolution does not couple the various sectors $k$. Here
$|k|$ is the number of ones in the binary representation of $k$
and $\left(M-2|k|\right)$ is the difference between the number of
zeros and ones; since $k=0,..,2^{M-1}-1$, then $|k|\in\left\{ 0,..,M-1\right\} $.
Correspondingly the SLD has a block diagonal form $L=\oplus_{k}L_{k}$
with 
\begin{equation}
L_{k}=2ip_{k}(M-2|k|)\left(\ket{GHZ_{k}^{+}}\bra{GHZ_{k}^{-}}-\ket{GHZ_{k}^{-}}\bra{GHZ_{k}^{+}}\right).
\end{equation}
In each block the eigenvectors are $\ket{\alpha_{\pm,k}}=\left(\ket{GHZ_{k}^{+}}\pm i\ket{GHZ_{k}^{-}}\right)/\sqrt{2}$.
By writing $\ket{\alpha_{\pm,k}}\doteq\ket{\pm}\tilde{\otimes}\ket{k}$
we define the $TPS^{R}$ and the whole Hilbert space can be written
as $\mathcal{H}_{2}\otimes\mathcal{H}_{N/2}$ with $N/2=2^{M-1}$.Each
of the GHZ basis states appearing in (\ref{eq: GHZ diagonal mixed state})
can be written as: 
\begin{equation}
\ket{GHZ_{k}^{+}}=\frac{\left(\ket{+}+\ket{-}\right)}{\sqrt{2}}\tilde{\otimes}\ket{k}=\ket{0}_{z}\tilde{\otimes}\ket{k}\label{eq: GHZ TPSR}
\end{equation}
In this $TPS^{R}$ one has that: 
\begin{eqnarray*}
G & = & S_{x}\tilde{\otimes}\sum_{k}\left(M-2|k|\right)\Pi_{k}\\
L & =2S_{y}\tilde{\otimes} & \sum_{k}\left(M-2|k|\right)\Pi_{k}
\end{eqnarray*}
where $S_{x},S_{y},S_{z}$ are the Pauli operators acting on $\mathcal{H}_{2}$
(see Appendix \ref{sec: Appendix Examples N dimensional-mixed states Maximal QFI}
for a derivation). The initial state reads 
\begin{eqnarray}
\rho_{0} & = & \ket{0}_{zz}\bra{0}\tilde{\otimes}\sum_{k}p_{k}\Pi_{k},\label{eq: GHZ mixed state in TPS}
\end{eqnarray}
(where $S_{z}\ket{0}_{z}=\ket{0}_{z}$), it is a product state in
the $TPS^{R}$, it is in general mixed and it is therefore a special
kind of the states (\ref{eq: General discordant state}).

The evolved state is 
\begin{eqnarray*}
\rho_{\lambda} & =\sum_{k=0}^{N/2-1} & p_{k}\left\{ \exp\left[-i\lambda\left(M-2|k|\right)\ S_{x}\right]\ket{0}_{zz}\bra{0}\exp\left[-i\lambda\left(M-2|k|\right)\ S_{x}\right]\right\} \tilde{\otimes}\Pi_{k}.
\end{eqnarray*}
We start by analyzing the pure state case $p_{0}\rightarrow1$, i.e.,
$\ket{GHZ_{0}^{+}}=\left(\ket{00...0}+\ket{11...11}\right)/\sqrt{2}$.
By tracing out the $\mathcal{H}_{N/2}$ system 
\begin{equation}
Tr_{\mathcal{H}_{N/2}}\left[U_{\lambda}\ket{\psi_{0}}\bra{\psi_{0}}U_{\lambda}^{\dagger}\right]=\exp\left(-i\lambda MS_{x}\right)\ket{0}_{zz}\bra{0}\exp\left(i\lambda MS_{x}\right).\label{eq: GHZ evolution pure state}
\end{equation}
one can clearly see that estimation process based on the multi-qubit
$\ket{GHZ_{0}^{+}}$ is completely equivalent to a single qubit multi-round
strategy \cite{Multiroundtheory4} where a single qubit operator $\exp\left(-i\lambda\sigma_{x}\right)$
is applied $M$ times to the initial state $\ket{0}_{z}$. In both
cases the link with coherence can be found within our approach and
\begin{equation}
QFI=\left[-\partial_{\lambda}^{2}Coh_{\mathcal{B}_{\alpha}}(\ket{0^{\lambda}}_{zz}\bra{0^{\lambda}})\right]_{\lambda=0}=4M^{2}
\end{equation}
where $\ket{0^{\lambda}}_{z}=\exp\left(-i\lambda MS_{x}\right)\ket{0}_{z}$.
The Heinsenberg scaling with $M$ therefore has the very same root
both in the single qubit multi-round protocols and multi-qubit GHZ
pure state case: \textit{it is the ability of $G$ to vigorously change
the relevant single-qubit coherence that allows for such scaling}.
\\
As for the decomposition (\ref{eq: QFI in terms of partial2 MI})
we notice that the evolution takes place in the $k=0$ sector thus
no correlations between subsystem $\mathcal{H}_{2}$ and $\mathcal{H}_{N/2}$
are created and one has that $QFI=FI_{2}$ and $\left(\partial_{\lambda}^{2}\mathcal{M}_{L_{0}}^{\lambda}\right)_{\lambda=0}=0$.
Given the previous discussion one may therefore argue that correlations,
and in particular entanglement, play no role in the estimation process.
As shown by this simple example, by selecting the proper $TPS^{R}$
one can formally establish an equivalence between two seemingly different
estimation procedures and find the common resource that is at the
basis of their efficiency.

We now pass to analyze the general mixed state (\ref{eq: GHZ diagonal mixed state}).
Within the given $TPS^{R}$ representation it is manifest that during
the whole evolution the operator $U_{\lambda}$ does correlate subsystems
$\mathcal{H}_{2}$ and $\mathcal{H}_{N/2}$, but it never creates
entanglement. Since $L=\oplus L_{k}$ one again obtains

\begin{equation}
QFI=\sum_{k}p_{k}QFI_{k}=4\sum_{k}\left(M-2|k|\right)^{2}p_{k}=4\Delta^{2}G
\end{equation}
where $\Delta^{2}G$ is the variance of $G$. As for the relation
with coherence, for each sector $k$ one has 
\begin{eqnarray*}
QFI_{k} & = & \left[-\partial_{\lambda}^{2}Coh_{\mathcal{B}_{\alpha}}(\tau_{k}^{\lambda})\right]_{\lambda=0}=4\left(M-2|k|\right)^{2}
\end{eqnarray*}
where $\tau_{k}^{\lambda}=\exp\left(-i\lambda\left(M-2|k|\right)S_{x}\right)\ket{0}_{zz}\bra{0}\exp\left(i\lambda\left(M-2|k|\right)S_{x}\right)$.
Therefore, in each sector the $QFI_{k}$ is given by the second order
variation of a single qubit coherence, and the estimation process
can be seen as a parallelized version of a single qubit multi-round
one, where in each of the $2^{M-1}$ sectors labeled by $k$ the single
qubit $\ket{0}_{z}$ is rotated in the $zy$ plane by $\exp\left(-i\lambda(M-2|k|)\ S_{x}\right)$.
The global $QFI$ is therefore an average of the sectors' single qubit
$QFI_{k}$'s or equivalently of sectors' single qubit second order
coherences variations.

We now discuss the decomposition of the $QFI$ according to Proposition
4. Here $L_{0}$ is full rank and one can uniquely write (see Appendix
\ref{sec: Appendix Examples N dimensional-mixed states Maximal QFI})

\begin{equation}
FI_{2}=\left(M-2\sum_{k}|k|p_{k}\right)^{2}
\end{equation}
such that in general $\left(\partial_{\lambda}^{2}\mathcal{M}_{L_{0}}^{\lambda}\right)_{\lambda=0}=QFI-FI_{2}\ge0$.
In order to give a physical interpretation of the various terms composing
the $QFI$ we define the following operators: 
\begin{equation}
O_{a}=S_{a}\otimes\sum_{k}\left(M-2|k|\right)\Pi_{k}
\end{equation}
with $S_{a},\ a=x,y,z$ Pauli matrices on $\mathcal{H}_{2}$, and
$Tr\left[O_{a}O_{b}^{\dagger}\right]=0$ when $a\neq b$. We have
that $O_{x}=G$, $2O_{y}=L_{0}$ and it holds $QFI=4\Delta^{2}O_{x}=4\Delta^{2}O_{y}$.
As for the operator $O_{z}$, the latter is diagonal in the GHZ-basis
and therefore it commutes with $\rho_{0}$. Furthermore $O_{z}$ has
the same set of eigenvalues of $O_{x}=G$ i.e, $\left\{ M-2|k|\right\} $.
Since $k=0,..,N/2-1=2^{M-1}-1$ we now can interpret the $\mathcal{H}_{2^{M-1}}=span\left\{ \ket{k}\right\} $
sector of the of the $TPS^{R}$ as an $M-1$ qubit system, such that
$\ket{k}=\ket{k_{M-1},..,k_{1}}$ where $\left(k_{M-1},..,k_{1}\right)$
is the $\left(M-1\right)$-digits binary representation of $k$. In
this way the operator $\sum_{k}\left(M-2|k|\right)\Pi_{k}$ can be
written as 
\begin{eqnarray*}
\sum_{k}\left(M-2|k|\right)\Pi_{k} & = & \mathbb{I}_{2^{M-1}}+\hat{S}_{z}
\end{eqnarray*}
where $\hat{S}_{z}=\sum_{t=0}^{M-1}\sigma_{z}^{t}$ is the total angular
momentum along the $z$ direction for an $\left(M-1\right)$ qubits
system. $\hat{S}_{z}$ is diagonal in the $\ket{k}$ basis and its
eigenvalues are $\left\{ M-1-2|k|\right\} _{k=0}^{M-1}$. With this
representation the shift operator can written as 
\begin{eqnarray}
G & = & S_{x}\otimes\mathbb{I}_{2^{M-1}}+S_{x}\otimes\hat{S}_{z}\label{eq: GHZ G as Interaction Hamiltonian}
\end{eqnarray}
while 
\begin{eqnarray}
O_{z} & = & S_{z}\otimes\mathbb{I}_{2^{M-1}}+S_{z}\otimes\hat{S}_{z}\label{eq: Oz as interaction qubit-spin system}
\end{eqnarray}
The above picture based on $TPS^{R}$ allows to see that the dynamical
evolution enacted by $G$ is fully equivalent to as ``system-bath''
interaction where: the role of the system is played by the single
qubit, its initial state being $\ket{0}_{z}$; the role of the bath
by the $(M-1)$-qubits system defined above, its initial state being
$\sum_{k}p_{k}\Pi_{k}$; and $G$ as described in (\ref{eq: GHZ G as Interaction Hamiltonian})
can be seen as a system-bath interaction Hamiltonian. If one is able
to prepare the state (\ref{eq: GHZ mixed state in TPS}) and to realize
the interaction Hamiltonian (\ref{eq: GHZ G as Interaction Hamiltonian})
the equivalence given by the description in the $TPS^{R}$ provides
\textit{an alternative way of enacting the estimation procedure that
is based on the preparation of (mixtures) of $M$ qubits product states
and does not require the preparation of initial states that are highly
entangled such as the $GHZ$ ones}. While the kind of interaction
described by (\ref{eq: GHZ G as Interaction Hamiltonian}) may be
difficult to realize in a $M$-qubit system, the above picture suggest
that in order to implement the described estimation process one could
resort to an analogous interaction between a single qubit and a spin-$j$
systems prepared in an eigenstate of the corresponding $\hat{S}_{z}$.

As for  the decomposition of $QFI$, it turns out that the relevant
quantities can be simply expressed in terms of $O_{z}$ as 
\begin{eqnarray*}
FI_{2} & = & 4\left\langle O_{z}\right\rangle ^{2}\\
\left(\partial_{\lambda}^{2}\mathcal{M}_{L_{0}}^{\lambda}\right)_{\lambda=0} & = & 4\Delta^{2}O_{z}
\end{eqnarray*}
and by using (\ref{eq: Oz as interaction qubit-spin system}) one
has 
\begin{eqnarray*}
\left\langle O_{z}\right\rangle  & = & 1+\left\langle S_{z}\otimes\hat{S}_{z}\right\rangle \\
\Delta^{2}O_{z} & = & \Delta^{2}\left(S_{z}\otimes\hat{S}_{z}\right)
\end{eqnarray*}
where we have used the fact that for the state (\ref{eq: GHZ mixed state in TPS})
$\left\langle \left(\mathbb{I}_{2}\otimes\hat{S}_{z}\right)\right\rangle =\left\langle S_{z}\otimes\hat{S}_{z}\right\rangle .$Therefore
the value of $FI_{2}$ is determined by the average of interaction
operator $S_{z}\otimes\hat{S}_{z}$ while $\left(\partial_{\lambda}^{2}\mathcal{M}_{L_{0}}^{\lambda}\right)_{\lambda=0}$
is determined by its variance thus providing a physical interpretation
of these quantities. In the above picture, the effect of the system-bath
interaction is to correlate the two subsystems $\mathcal{H}_{2}$
and $\mathcal{H}_{2^{M-1}}$ and the effect on the single qubit is
a conditional rotation around the $\hat{x}$ axis that depends on
$\left(M-2|k|\right)$. Furthermore, when $p_{0}=1$ the initial state
is pure and it can be represented as $\ket{\psi_{0}}=\ket{GHZ_{0}^{+}}$
or in the $TPS^{R}$ as $\ket{\psi_{0}}=\ket{0}_{z}\tilde{\otimes}\ket{0}$.
In this case the bath does not evolve under the action of $G$ and
the two subsystems \textit{remains uncorrelated during the evolution},
such that in particular $\left(\partial_{\lambda}^{2}\mathcal{M}_{L_{0}}^{\lambda}\right)_{\lambda=0}=\Delta^{2}\left(S_{z}\otimes\hat{S}_{z}\right)=0$.
The evolution can now be represented in the single qubit sector only
as (\ref{eq: GHZ evolution pure state}) such that 
\begin{eqnarray*}
QFI & = & FI_{2}=4\left\langle O_{z}\right\rangle ^{2}
\end{eqnarray*}
and the $QFI$ can be alternatively interpreted as: the square of
the average value of $O_{z}$; the variance of the system-bath interaction
Hamiltonian $O_{x}=G$; or the second order variation of a single
qubit coherence. Therefore, once the evolution has taken place, as
for the estimation process \textit{one only needs to realize the single
qubit measurement} defined by the $L_{0}$ eigenbasis $\left\{ \ket{0}_{y},\ket{1}_{y}\right\} $,
and the estimation precision is again provided by the second order
variation of a single qubit coherence.

We conclude this section by discussing the possibility of achieving
a quasi-Heisenberg scaling of $QFI$ with GHZ mixed states of the
kind (\ref{eq: GHZ diagonal mixed state}) and large number of qubits
$M$ or alternatively, thanks to the $TPS^{R}$ representation, with
a single qubit system coupled to a $M-1$ qubit one. We first notice
that the scaling can be achieved whenever the distribution $\{p_{k}\}$
is mostly concentrated in the sectors $k$ such that $|k|\ll M$ and/or
$|k|\approx M$ . This happens despite the fact that the initial state
is mixed and, provided it satisfies the previous conditions whatever
the distribution $\{p_{k}\}$ i.e., whatever the (quantum) correlations
that may be built by the operator $U_{\lambda}$ between the subsystems
$\mathcal{H}_{2}$ and $\mathcal{H}_{N/2}$. This is exactly the kind
of situation we will encounter when in Section \ref{sec:Coherence-and-noise}
where we discuss noisy estimation schemes base on GHZ states.\\

\subsection{Examples: (mixed) NOON states}

Another class of states of the type (\ref{eq: General discordant state})
that are worth mentioning are the following general (mixed) NOON states

\begin{equation}
\rho=p_{0}\ket{0,0}\bra{0,0}+\sum_{k=1}^{\infty}p_{k}\left\{ \left[\ket{k,0}\bra{k,0}+\ket{0,k}\bra{0,k}\right]+\left[\eta_{k}\ket{k,0}\bra{0,k}+\eta_{k}^{\ast}\ket{0,k}\bra{k,0}\right]\right\} 
\end{equation}
where $\ket{k,0}=\ket{k}_{a}\ket{0}_{b},\ket{0,k}=\ket{0}_{a}\ket{k}_{b}$
are Fock states of a two-modes $(a,b)$ quantum optical system with
$k$-photons in each mode respectively and $\ket{0,0}$ is the vacuum
contribution. These states have been proposed for use in quantum phase
estimation and have been thoroughly analyzed in Ref. \cite{Vogel}.
Within each sector $k$ the state can be represented as $\tau_{k}\otimes\Pi_{k}$
with 
\begin{equation}
\tau_{k}=\left(\begin{array}{cc}
1/2 & \eta_{k}\\
\eta_{k}^{*} & 1/2
\end{array}\right)
\end{equation}
i.e. as a single qubit state lying in the $xy$ plane with Bloch vector
$\bar{\eta}_{k}=|\eta_{k}|\left(\cos\arg\eta_{k},\sin\arg\eta_{k},0\right)$;
where $|\eta_{k}|$ determines the purity of state \cite{GenoniNatureSinleQubit}.
When now one applies a phase shift by means, for example, of the single
mode operator $U(\lambda)=e^{-i\lambda a^{\dagger}a}$, one has 
\[
\eta_{k}\ket{k,0}\bra{0,k}\rightarrow e^{-ik\lambda}\eta_{k}\ket{k,0}\bra{0,k}\qquad\eta_{k}^{\ast}\ket{k,0}\bra{0,k}\rightarrow e^{ik\lambda}\eta_{k}^{\ast}\ket{k,0}\bra{0,k}
\]
such that, aside for the vacuum contribution, the evolved state is
analogous to (\ref{eq: General discordant state}) and reads 
\begin{equation}
\rho_{\lambda}=p_{0}\ket{0,0}\bra{0,0}+\sum_{k=1}^{\infty}p_{k}\left(\begin{array}{cc}
1/2 & e^{-ik\lambda}\eta_{k}\\
e^{ik\lambda}\eta_{k}^{*} & 1/2
\end{array}\right)\otimes\ket{k}\bra{k}
\end{equation}
Again the process can be seen as a single qubit multi-round like one.
Indeed, in each sector $k$ the phase shift amounts to a rotation
of the intial Bloch vector $\bar{\eta}_{k}$ of an angle $k\lambda$
in the $xy$ plane. In general the estimation process is realized
by enacting measurements represented by operators of the kind $A_{M}=\ket{0,M}\bra{M,0}+\ket{M,0}\bra{0,M}$,
which can be realized by means of interference measurements in a Mach-Zender
interferometer set up. In our picture $A_{M}$ corresponds to measuring
the operator $S_{x}\otimes\ket{M}\bra{M}$. In this set up the measurement
is supposed to be fixed; in terms of the single qubit system, measuring
$A_{M}$ corresponds to measuring along the $\hat{x}$ axis. This
measurement is optimal only when for the state $\tau_{M}^{\lambda}$
it projects onto the eigenstates of $L_{\lambda}^{M}$ i.e., the SLD
pertaining to the sub-block $M$. But this happens only for specific
values of the phase $\lambda$ i.e., when the condition $\left(\arg\eta_{k}-k\lambda\right)\mod\pi/2\approx0$
is verified with sufficiently good approximation\cite{Vogel}, and
the Bloch vector of the rotated state $\tau_{M}^{\lambda}$ is sufficiently
close to the $\widehat{y}$ axis. The efficiency of the protocol depends
on $QFI=p_{M}QFI_{M}$ where $QFI_{M}$ is the usual single qubit
$QFI$ that corresponds to the second order variation of the coherence
of the $L_{\lambda}^{M}$ eigenbasis and $p_{M}$ is the probability
of projecting the state onto the sector $M$. In order to obtain a
super-SQL scaling on one hand $p_{M}\left|\bar{\eta}_{k}\right|^{2}$
must be sufficiently greater than $1/M$; this in particular is true
for a pure NOON state. On the other hand the condition $\left(\arg\eta_{k}-k\lambda\right)\mod\pi/2\approx0$
should be satisfied with sufficient accuracy.

\section{Coherence in phase estimation in presence of noise\label{sec:Coherence-and-noise}}

We now pass to analyze the situation in which the estimation protocols
are affected by noise. In general the estimation processes may be
``noisy'' for different reasons: the imperfections of the state
preparation procedure or the coupling of the probe state to the surrounding
environment. Since in experiments typically one or both of these situations
actually occur, noisy estimation processes have been and are being
object of intense study\cite{AlipourNoisyMetrology,Acin,davidovich,Demkowicz2,Demkowicz,DurNoisyMetrology,Kolodynski,TsangMetrology,MetrologyReview,WangGHZNoise,GenoniNoisyEsitmation}.
In the following, we extend the framework laid down in previous sections
to some specific cases of noisy evolution. This will allow us on one
hand to describe the connection between phase estimation and coherence
in presence of noise, and on the other hand to exemplify our approach
in specific relevant examples.

\subsection{Noise in state preparation and ``commuting noise''}

The extension of our approach is straightforward in at least two cases.
The first case is when the noise acts \emph{before} the phase encoding
starts, i.e., the shift operator is applied to a mixed state $\rho_{0}$
that is the result of a noisy process or an imperfect state preparation
procedure. Examples of this situation may well be the mixed states
such as mixed GHZ and NOON states, analyzed in the previous section.
In this case one can straightforwardly apply the results of Sec. \ref{sec:QFI-and-coherence.}
and \ref{sec:Coherence-in-phase} for mixed states, which hold for
any dimension $N$: in particular, Eqs. (\ref{eq: QFI as second order variation of coherence})
and (\ref{eq: QFI in terms of partial2 MI}) hold without change.
A second case where our approach applies is when the system is coupled
to a noisy environment but the map describing the overall process
$\Lambda_{\lambda,\gamma}\left[\rho_{0}\right]$ is given by the composition
of two different \textit{commuting} maps $\Lambda_{\lambda},\Lambda_{\gamma}$,
enacting the coherent ($\Lambda_{\lambda}$) and decoherent ($\Lambda_{\gamma}$)
part of the evolution (here $\gamma$ is a generic parameter characterizing
the decoherent process). Then $\Lambda_{\lambda,\gamma}\left[\rho_{0}\right]=\Lambda_{\gamma}\left[\Lambda_{\lambda}\left[\rho_{0}\right]\right]=\Lambda_{\lambda}\left[\Lambda_{\gamma}\left[\rho_{0}\right]\right]$
and the estimation process is equivalent to applying the coherent
shift $\lambda$ to the decohered state $\rho_{\gamma}=\Lambda_{\gamma}\left[\rho_{0}\right]$.
This scenario includes several nontrivial cases of noisy evolutions
(neat examples are given in \cite{Demkowicz,Kolodynski,Demkowicz2}).
It also includes the situation in which the unitary evolution is encoded
in a decoherence free subsystem\cite{LidarDFS,ZanardiSFS}, i.e.,
in which the unitary evolution enacting the phase shift $\Lambda_{\lambda}$
takes places in a subspace that is left invariant by the noisy map
$\Lambda_{\gamma}$. In this case, if the initial pure state $\rho_{0}^{enc}$
belongs to the decoherence-free subsystem, it is left untouched by
the noise ($\rho_{0}^{enc}=\Lambda_{\gamma}\left[\rho_{0}^{enc}\right]$)
and the estimation process can take place by means of a proper encoding
$\Lambda_{\lambda}^{enc}$ of the phase shift. The latter corresponds
to a unitary process and our arguments on the role of coherence can
be easily applied

\subsection{``Non-commuting'' noise\label{sub: Non-commuting-noise}}

The above examples show that the approach developed in previous sections
can be straightforwardly extended to whole classes of noisy evolutions.
But what happens when the noisy and the coherent map do not commute?
In such general case the picture changes, mainly because the eigenvalues
$\left\{ \epsilon_{i}^{\lambda}\right\} $ of the output state of
the evolution $\rho_{\lambda}=\Lambda_{\lambda,\gamma}(\rho_{0})$
possibly depend on $\lambda$, and therefore $\mathcal{V}(\rho_{\lambda})$
is not conserved. Then formula (\ref{eq: QFI as second order variation of coherence})
holds with 
\begin{equation}
f(\rho_{\lambda})=f^{\chi}(\rho_{\lambda})+f^{\epsilon}(\epsilon_{i}^{\lambda})-QFI^{c}\label{eq: second order variation Coh - general noisy case}
\end{equation}
where $f^{\chi}(\rho_{\lambda})=-\left[\partial_{\delta\lambda}^{2}\mathcal{X}(p_{\alpha}^{\lambda+\delta\lambda}|p_{\alpha}^{\lambda})\right]_{\delta\lambda=0}$is
the usual term appearing in the noise-free case while the additional
terms 
\begin{eqnarray*}
f^{\epsilon}(\epsilon_{i}^{\lambda}) & = & -\sum_{i}\left(\partial_{\delta\lambda}^{2}\epsilon_{i}^{\lambda+\delta\lambda}\right)_{\delta\lambda=0}\log\epsilon_{i}^{\lambda}\\
QFI^{c} & = & \sum_{i}\left(\partial_{\delta\lambda}^{2}\epsilon_{i}^{\lambda+\delta\lambda}\right)_{\delta\lambda=0}^{2}/\epsilon_{i}^{\lambda}
\end{eqnarray*}
come from the variation of the eigenvalues of $\rho_{\lambda}$ with
$\lambda$. The term $QFI^{c}$ is the ``classical contribution''
to the $QFI$ due to the first order variation of the $\epsilon_{i}^{\lambda}$
with $\lambda$\cite{BraunsteinCaves}. Since $QFI=QFI^{c}+QFI^{Q}$
where $QFI^{Q}$ is the part depending on the variation of the eigenvectors
of $\rho_{\lambda}$ with $\lambda$, relation (\ref{eq: QFI as second order variation of coherence})
can be further simplified as 
\begin{equation}
-\left(\partial_{\delta\lambda}^{2}Coh\right)_{\delta\lambda=0}=f^{\epsilon}(\epsilon_{i}^{\lambda})+f^{\chi}(\rho_{\lambda})+QFI^{Q}\label{eq: Simplified second order variation Coh - general noisy case}
\end{equation}
Discussing the previous relation in the most general case and drawing
general conclusions on complex decoherent processes is a compelling
but rather hard task. In the following we focus on a relevant example
of noisy estimation processes that has been diffusely explored in
the literature: phase estimation based on the $M$-qubit GHZ state
in the presence of noise (\cite{WangGHZNoise}).

The relevance of this case stems on one hand from the fact that, as
previously discussed, GHZ states in the noise-free case allow to attain
the Heisenberg limit $QFI\propto M^{2}$. On the other hand, it is
has been shown \cite{Acin} that GHZ states allow to overcome the
Standard Quantum Limit (SQL) $QFI\propto M$ even in the presence
of noise (for specific kinds of noisy maps). In the following we show
how our approach based on coherence applies and how the results seen
in Section \ref{sub: Main Examples:-GHZ-states} can be extended.

The specific setting we describe in the following has been put forward
in \cite{Acin} and we now briefly review it. Assume that the $M$-qubits
systems undergoes a coherent phase shift and is subject to Pauli diagonal
noise. Then its evolution is governed by the Markovian master equation
\begin{equation}
\partial_{t}\rho=\mathcal{H}(\rho)+\mathcal{L}(\rho)\label{eq: MasterEquationLocal}
\end{equation}
where the unitary part enacting the phase shift is the same used for
the pure state case and it is given by 
\begin{equation}
\mathcal{H}(\rho)=-\frac{i\omega}{2}[\sum_{h}\sigma_{z}^{h},\rho].\label{eq: Coherent part GHZ multi qubit}
\end{equation}
while the decoherent part is given by 
\begin{equation}
\mathcal{L}(\rho)=-\frac{\gamma}{2}\sum_{h}\left[\rho-\alpha_{x}\sigma_{x}^{h}\rho\sigma_{x}^{h}-\alpha_{y}\sigma_{y}^{h}\rho\sigma_{y}^{h}-\alpha_{z}\sigma_{z}^{h}\rho\sigma_{z}^{h}\right]\label{eq: Decoherent part GHZ multi qubit}
\end{equation}
where $\sigma_{a}^{h},\ a=x,y,z$ are Pauli matrices acting on the
$h$-th qubit, $\alpha_{a}\ge0$ and $\sum_{a}\alpha_{a}=1$. Here
the goal is to estimate the frequency $\omega$ with the best possible
precision. The latter satisfies the Quantum Cramer-Rao bound 
\[
\delta^{2}\omega\ T\ge\left(QFI/t\right)^{-1}
\]
where the limits to the precision depend on the instant of time $t$
in which the estimation takes place. There are two extremal cases
of noise relevant for the discussion. In the ``parallel'' case ($\alpha_{z}=1$)
the noise acts locally on each qubit along the $z$ direction. In
the ``transverse'' case ($\alpha_{x}=1$) the noise acts locally
on each qubit along the $x$ direction.

In \cite{Acin} the Authors showed, in agreement with \cite{Demkowicz},
that in the parallel case one can only achieve a precision that scales
in terms of the number of qubits at most as $M^{-1}$, the SQL. On
the other hand, the Authors were able to show that, in the purely
transverse case, by optimizing the time at which the estimation takes
place such that $t_{opt}=(3/\gamma\omega^{2}M)^{1/3}$, then one can
overcome the SQL and obtain $\delta^{2}\omega\ T\ge M^{-5/3}$. In
the following we show how these two cases can be cast in our framework.

In order to analyze the dynamics we use the same $TPS^{R}$ Eq. (\ref{eq: GHZ TPSR})
defined for the noise-free GHZ case in Section (\ref{sub: Main Examples:-GHZ-states}).
Indeed, as we show in Appendix \ref{sec:Appendix F QFI GHZ states noisy},
for any direction of the noise (i.e., for any value of $\alpha_{x},\alpha_{y},\alpha_{z}$)
the state of the $M$-qubit system $\rho_{\omega,\gamma}(t)$, solution
of the master equation, can be written in $TPS^{R}$ representation
as 
\begin{equation}
\rho_{\omega,\gamma}(t)=\sum_{k}p_{k}(t)\tau_{k}(\omega,t)\tilde{\otimes}\ket{k}\bra{k}\label{eq: Acin global state}
\end{equation}
where: $\tau_{k}(\omega,t)$ are single-qubit density matrices ; $p_{k}(t)$
is the probability of finding the system in the $k$-sector and it
does not depend on $\omega$ (for a detailed expression of $\rho_{\omega,\gamma}(t)$
see Appendix \ref{sec:Appendix F QFI GHZ states noisy} or Ref. \cite{Acin}).
Before continuing our analysis we notice that, in principle the SLD
$L_{\omega}(t)$ is a function of time and so are its eigenvectors.
The latter can be numerically evaluated and the relative $TPS^{R}(t)$
determined. The procedure is in general involved. However, having
realized that in $TPS^{R}=TPS^{R}(0)$ (\ref{eq: GHZ TPSR}) the state
(\ref{eq: Acin global state}) has a very simple structure, we proceed
within this picture. Its usefulness can be further justified by noticing
that, since in general $t_{opt}\ll1$, one may argue that the dynamics
taking place in the single qubit factor identified by $TPS^{R}(0)$
well represents the one taking place in the actual single qubit factor
defined by $TPS^{R}(t_{opt})$.

Thanks to the introduced representation,  we can first notice that
the coherent dynamics only affects the single qubit subsystem $\mathcal{H}_{2}$.
The state of the system $\rho_{\omega,\gamma}(t)$ has, for any $\left(t,\omega\right)$,
the generic form introduced in Section \ref{sub: Main Example discordant}.
Therefore 
\begin{equation}
QFI(t,\omega)=\sum_{k}p_{k}(t)QFI_{k}(t,\omega)=\left\langle QFI_{k}\right\rangle (t,\omega)\label{eq: QFI as average over k_M}
\end{equation}
i.e., the QFI is the average of the single qubit $QFI_{k}$ (in the
following we adopt the simplified notation $\sum_{k}p_{k}f_{k}\rightarrow\left\langle f_{k}\right\rangle $).
Equation (\ref{eq: QFI as average over k_M}) is valid for all kind
of noise directions (i.e., all allowed values of $\alpha_{x},\alpha_{y},\alpha_{z}$)
and it neatly shows the qubit-like nature of the estimation process,
allowing us to investigate the connection between the QFI and single
qubit coherences.

We start by focusing on the case of ``parallel'' noise. When $\alpha_{z}=1,\alpha_{y}=\alpha_{x}=0$,
the coherent and decoherent part of the evolution commute and furthermore,
since $\ket{\psi_{0}}=\ket{GHZ_{0}^{+}}=\frac{1}{\sqrt{2}}(\ket{+}+\ket{-})\tilde{\otimes}\ket{0}$,
the evolution takes place entirely in the $k=0$ sector. The master
equation can then be reduced to a master equation for $\tau_{0}$
as (see Appendix \ref{sec:Appendix F QFI GHZ states noisy}) 
\begin{equation}
\partial_{t}\tau_{0}=-\frac{iM\omega}{2}[S_{x},\tau_{0}]-\frac{M\gamma}{2}\left[\tau_{0}-S_{x}\tau_{0}S_{x}\right]\label{eq: master equation for rho_k_0 parallel noise}
\end{equation}
where $S_{x}$ is the Pauli operators acting on the single qubit factor
$\mathcal{H}_{2}$ as $S_{x}\ket{\pm}=\ket{\mp}$. Just as in the
noiseless GHZ case, the estimation process is a single qubit one and
our approach applies. In particular, the eigenvalues of $\tau_{0}$
depend on $\left(t,\gamma\right)$ but not on $\omega$, and $f(\tau_{0}(\omega))=0$;
therefore equation (\ref{eq: Simplified second order variation Coh - general noisy case})
reduces to 
\begin{equation}
-\left(\partial_{\omega_{1}}^{2}Coh\right)_{\omega_{1}=\omega}=QFI^{Q}.
\end{equation}
The Quantum Fisher Information $QFI(t,\omega)=QFI_{k=0}(t,\omega)=-\left(\partial_{\omega_{1}}^{2}Coh\left[\tau_{0}(t,\omega_{1})\right]\right)_{\omega_{1}=\omega}$
can be written for all $\left(t,\omega\right)$ as a second order
variation of the coherence of the measurement basis. The $TPS^{R}$
picture allows to understand why in the ``parallel'' case the estimation
precision cannot beat the SQL. While the phase imprinted, $M\omega$,
is proportional to the number of qubits, the whole decoherence acts
on the single qubit with a strength $M\gamma$ also proportional to
the number of qubits, thus neutralizing the enhancement in the precision
provided by $M\omega$. The same result can be explained by noticing
that in the parallel case the coherent (\ref{eq: Coherent part GHZ multi qubit})
and the decoherent (\ref{eq: Decoherent part GHZ multi qubit}) map
commute, and therefore the estimation process is therefore fully equivalent
to a single qubit multiround one enacted on a highly decohered state.

We now pass to analyze the ``transverse'' noise case. Its relevance
stems from the fact that in \cite{Acin} it was shown that it is possible
to attain a precision in the estimation of $\omega$ that scales with
the number of qubits as $M^{-5/3}$, i.e. to attain a quasi-Heisenberg
scaling. When $\alpha_{x}=1,\alpha_{y}=\alpha_{z}=0$ the coherent
and decoherent part of the state evolution no longer commute. In order
to discuss the role of coherence in the estimation process and to
investigate the connection with single qubit coherences we then have
to use formula (\ref{eq: Simplified second order variation Coh - general noisy case}).
Therefore in general one expects that $QFI_{k}\neq-\partial_{\omega}^{2}Coh\left[\tau_{k}(t=t_{opt},\omega)\right]$.
However, our numerical results show that the behavior of the terms
in (\ref{eq: Simplified second order variation Coh - general noisy case})
is the following. For fixed $\left(t=t_{opt}\right)$ the term $\left\langle QFI_{k}^{Q}\right\rangle $
grows as $M^{5/3}$ and it has the same scaling as $\left\langle QFI_{k}\right\rangle $.
On the other hand , the remaining contributions in Eq. (\ref{eq: second order variation Coh - general noisy case})
do not scale with $M$ and remain small: 
\begin{equation}
\left|\left\langle f_{k}^{\epsilon}(\epsilon_{i}^{\lambda})\right\rangle \right|,\left|\left\langle QFI_{k}^{c}\right\rangle \right|,\left|\left\langle f^{\chi}(\rho_{k})\right\rangle \right|<1\ll\left\langle QFI_{k}^{Q}\right\rangle 
\end{equation}
Therefore, for large $M$, we have: 
\begin{equation}
QFI\approx\left\langle -\partial_{\omega}^{2}Coh\left[\tau_{k}(\omega,t=t_{opt})\right]\right\rangle \approx\left\langle QFI_{k}^{Q}\right\rangle .\label{eq: Average QFI approx Average Coherence}
\end{equation}
In Fig. \ref{fig:QFIscaling}(a) we present, for $\gamma=\omega=1$,
the scaling with $M$ of $\left(\left\langle QFI_{k}(\omega,t=t_{opt})\right\rangle /t_{opt}\right)^{-1}$
and $\left(\left\langle -\partial_{\omega}^{2}Coh\left[\tau_{k}(\omega,t=t_{opt})\right]\right\rangle /t_{opt}\right)^{-1}$,
together with the asymptotic scaling $\left(QFI/t_{opt}\right)^{-1}=(9/8\ \gamma\omega^{2})^{1/3}N^{-5/3}$
that can be predicted with the channel extension method \cite{Acin}.
As shown in the plots, the two quantities have the same asymptotic
behavior leading to the expected super-SQL scaling on the precision
$\delta^{2}\omega\ T$. This result neatly shows that, in the transverse
case, what matters for the achievement of a super-SQL scaling of the
estimation precision is the second order variation of single qubit
coherences. If we now look at the state (\ref{eq: Acin global state})
we see that the single qubit subsystem $\mathcal{H}_{2}$, the only
one relevant for the estimation process, is never entangled with the
subsystem $\mathcal{H}_{2^{M-1}}$ during the whole evolution. 
\begin{figure}
\subfigure{\includegraphics[scale=0.5]{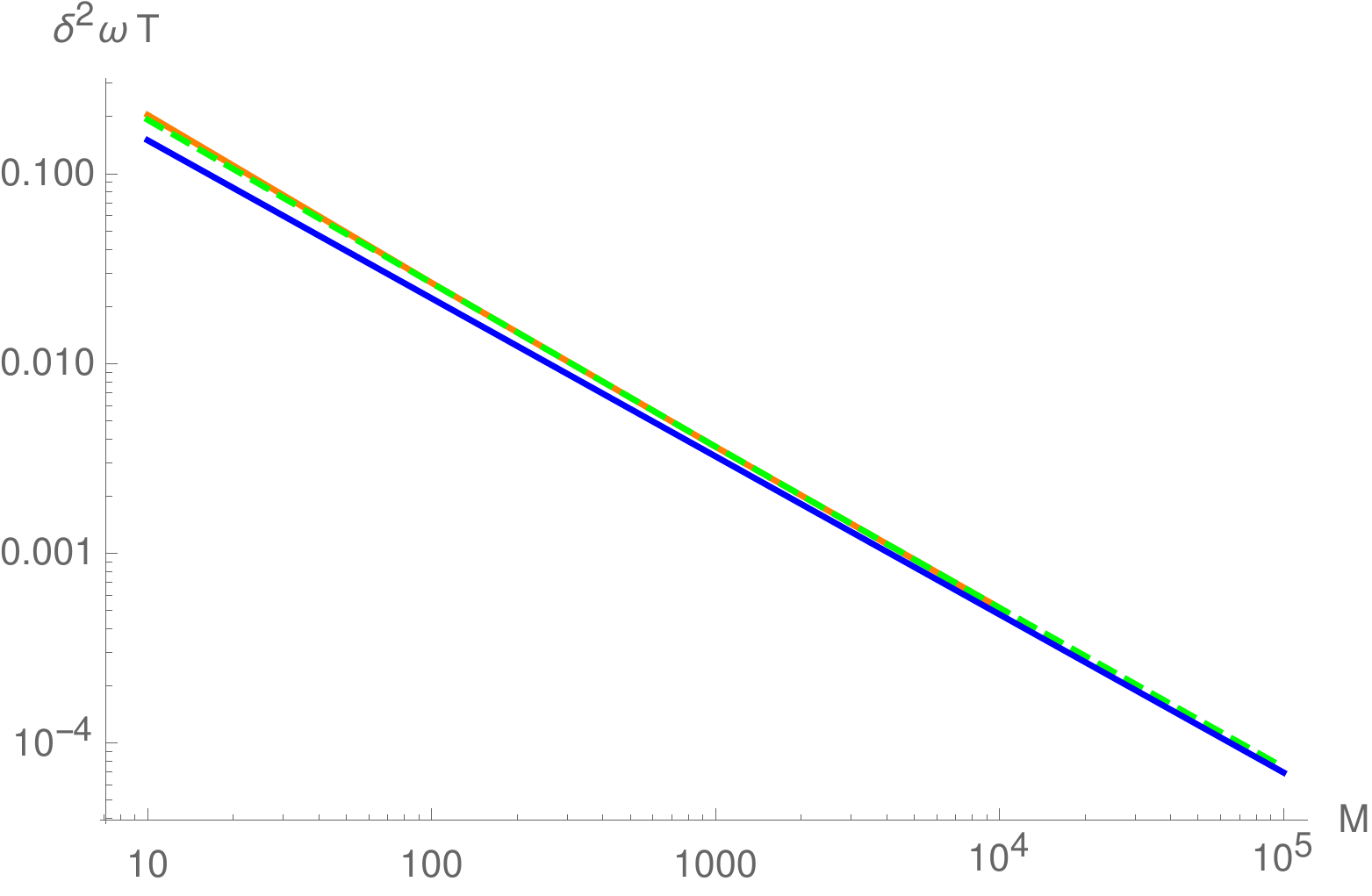}}\subfigure{\includegraphics[scale=0.3]{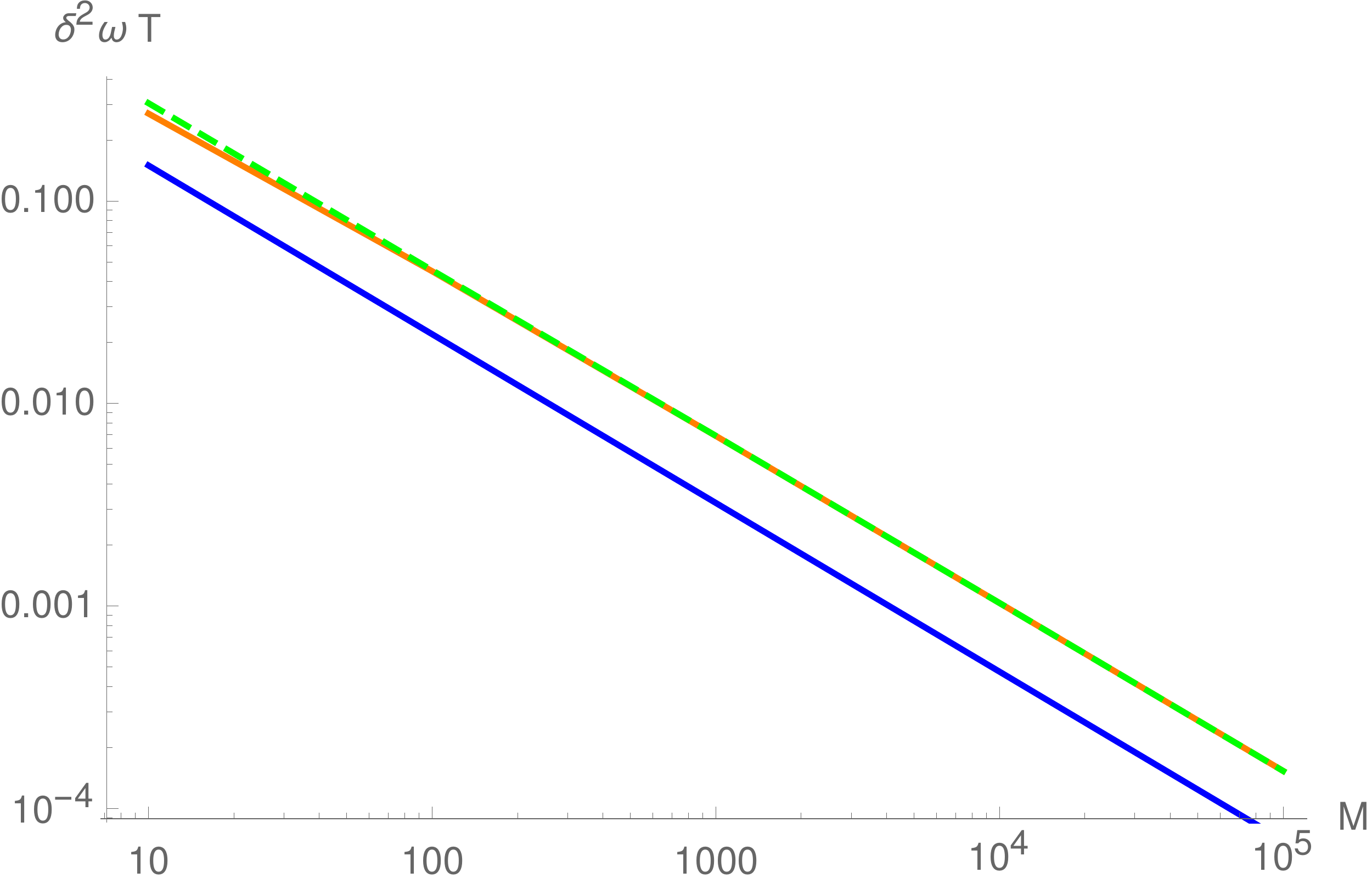}}\caption{\textbf{(a)} $\left(\langle QFI_{k}\rangle/t_{opt}\right)^{-1}$ (orange)
and $-\left(\left\langle \partial_{\omega}^{2}Coh\left[\tilde{\tau}_{k}\right]\right\rangle /t_{opt}\right)^{-1}$(green)
as a function of the number of qubits $M$, for $\omega=\gamma=1$.
The blue curve is the predicted bound $\left(QFI/t_{opt}\right)^{-1}=(9/8\ \gamma\omega^{2})^{1/3}N^{-5/3}$\textbf{(b)}
$\left(QFI_{\xi}/t_{opt}\right)^{-1}$ (orange) and $\left(-\partial_{\omega}^{2}Coh\left[\xi_{\lambda}\right]/t_{opt}\right)^{-1}$(green)
of the reduced qubit state $\xi_{\lambda}$ as a function of the number
of qubits $M$, for $\omega=\gamma=1$. The blue curve is the precision
bound $\left(QFI/t_{opt}\right)^{-1}=(9/8\ \gamma\omega^{2})^{1/3}N^{-5/3}$.
\label{fig:QFIscaling}}
\end{figure}

In order to grasp how the estimation process physically works we give
the following explanation. By using the description of the solution
given by Eq. (\ref{eq: Acin global state}) and by projecting the
master equation onto the $k$ subspaces (see the Appendix \ref{sec:Appendix F QFI GHZ states noisy}
) one obtains a system of coupled differential equations 
\begin{equation}
\partial_{t}\tilde{\tau}_{k}=\mathcal{H}_{k}+\mathcal{L}_{k}\label{eq:MasterEqk}
\end{equation}
that is equivalent to the master equation. For each of the (un-normalized)
single qubit density matrices $\tilde{\tau}_{k}(\omega,t)=p_{k}(t)\tau_{k}(\omega,t)$
one has that the coherent part of the evolution is dictated by 
\begin{eqnarray}
\mathcal{H}_{k} & = & -\frac{i\omega}{2}\left(M-2|k|\right)[S{}_{x},\tilde{\tau}_{k}]\label{eq: Differntial equation for rho_k_M tranverse noise: Coherent part}
\end{eqnarray}
where $|k|$ is the number of ones in the binary representation of
$k$, and the decoherent part by 
\begin{eqnarray}
\mathcal{L}_{k} & =- & \frac{\gamma}{2}\left[M\tilde{\tau}_{k}-S_{z}\tilde{\tau}_{k'(M)}S_{z}-\sum_{h=1}^{M-1}\tilde{\tau}_{k'(h)}\right]\label{eq: Differntial equation for rho_k_M tranverse noise: Decoherent part}
\end{eqnarray}
where $k'(h)$ is an $h$-dependent permutation of the $k$ sectors
(see Appendix \ref{sec:Appendix F QFI GHZ states noisy} for details).
The role of the decoherent part (\ref{eq: Differntial equation for rho_k_M tranverse noise: Decoherent part})
is to couple different sectors: each $k$ is coupled to the sectors
$k'\left(h\right)$ whose binary representation is obtained by flipping
just one of the $M$ bits in $k$. \\
In order to study how the process evolves we focus on the case $\omega=\gamma=1$
. Numerical results show that the system of equations (\ref{eq:MasterEqk})
is approximately solved by the following ansatz. Writing the unnormalized
single qubit states as 
\[
\tilde{\tau}_{k}=p_{k}(t)\left(\frac{\mathbb{I}_{2}+h_{k}\hat{n}_{k}\cdot\vec{S}}{2}\right)
\]
where $h_{k}\hat{n}_{k}$ are the corresponding Bloch vectors (lying
on the $yz$ plane in $TPS^{R}$ representation), an approximate solution
is obtained by replacing each $\hat{n}_{k}$ with the ansatz 
\begin{eqnarray}
\hat{a}_{k} & = & \left(0,\sin\left[(m_{k}-1)\omega t\right],\cos\left[(m_{k}-1)\omega t\right]\right)\label{eq: Transverse noise Approximate Bloch vector}
\end{eqnarray}
where $m_{k}=max(|k|,M-|k|)$. Thus, the overall process can be pictured
as follows. The evolution starts in the $k=0$ sector with the initial
state $\ket{\psi_{0}}=\ket{GHZ_{0}^{+}}=\frac{1}{\sqrt{2}}(\ket{+}+\ket{-})\tilde{\otimes}\ket{0}$.
The result of the evolution is given by the interplay of the coherent
$\mathcal{H}_{k}$ (\ref{eq: Differntial equation for rho_k_M tranverse noise: Coherent part})
and decoherent $\mathcal{L}_{k}$ (\ref{eq: Differntial equation for rho_k_M tranverse noise: Decoherent part})
part. The action of the ``transverse'' decoherence is twofold. On
one hand it progressively populates the sectors $k\in K(t)=\left\{ k:\ m_{k}\lesssim M)\right\} $.
Once these sectors are populated, the corresponding states undergo
the action of the coherent evolution $\mathcal{H}_{k}$. As it can
be inferred from the approximate Bloch vector (\ref{eq: Transverse noise Approximate Bloch vector}),
the latter is similar to the multi-round single qubit evolution described
in Section \ref{sub: Main Examples:-GHZ-states}. The actual phase
(frequency $\omega$) imprinted onto the state is $(m_{k}-1)\omega\lesssim M\omega$.
The process continues up to $t_{opt}$ where one can check that only
the sectors $k\in K(t)$ are substantially populated. During the evolution
each $k\in K(t)$ state undergoes a decoherence process. However the
choice of $t_{opt}\ll1$ guarantees that: $i)$ the effect of decoherence
is not too relevant; $ii)$ the multi round single qubit processes
has sufficient time to impress the phase $(m_{k}-1)\omega\lesssim M\omega$
onto each single qubit state. Should the measurement time be $\ll t_{opt}$
the latter effect would not be seen.

While this qualitative picture is based on an approximate solution
devised for $\omega=\gamma=1$, one can find similar solutions for
other values of $\omega,\gamma$. The relevant point is that the above
discussion shows how the overall estimation process can be seen as
a \textit{parallel multi round estimation one,} enacted onto (a fraction)
of the single qubits $\tilde{\tau}_{k}$ such that $k\in K(t)$. In
terms of the overall $QFI$, each relevant sector contributes with
a $QFI_{k}$ that is fairly well approximated by the second order
derivative of the relative coherence function: 
\[
QFI_{k}\approx QFI_{k}^{Q}\approx-\partial_{\omega}^{2}Coh\left[\tau_{k}(\omega,t=t_{opt})\right]
\]
Since in each relevant sector the process is a multi-round single
qubit one, the super-SQL scaling is the result of a change of the
coherence of the eigenbasis of the SLD ($L_{k}$) pertaining to each
sector.

We conclude this section by analyzing how possible bounds on $QFI$
that can be obtained within the $TPS^{R}$ used. In Section \ref{sub: QFI and Coherence in  N dimensional-case}
we have seen that given a $TPS^{R}$ one has 
\begin{eqnarray*}
FI_{2}\le & QFI(Tr_{\mathcal{H}_{N/2}}\left[\rho_{(t,\omega)}\right]) & \le QFI(\rho_{(t,\omega)}).
\end{eqnarray*}
Now if one aims at computing $FI_{2}$ in principle one needs to determine
the $TPS^{R}(t)$ induced by the time dependent SLD $L_{\omega}(t)$.
  However, as noticed above the $TPS^{R}(0)$ we have used so far
not only is easier to use but it allows to analytically compute $\xi\left(t,\omega\right)=Tr_{\mathcal{H}_{N/2}}\left[\rho_{\lambda}\right]$
and $QFI_{\xi}\left(t,\omega\right)$. The latter turns out to be
a meaningful lower bound to $QFI(\rho_{(t,\omega)})$. Indeed, as
shown in Fig. \ref{fig:QFIscaling}(b) one can verify that also in
this case 
\begin{eqnarray*}
\left(QFI_{\xi}(t=t_{opt},\omega)/t_{opt}\right)^{-1} & \approx & \left(-\partial_{\omega}^{2}Coh\left[\xi(t=t_{opt},\omega_{1})\right]/t_{opt}\right)_{\omega_{1}=\omega}^{-1}
\end{eqnarray*}
and the precision that can be obtained by using $\xi(\omega,t)$ as
probe state allows to attain the usual super-SQL scaling with $M$.
This result can be neatly interpreted as follows. If one traces over
$\mathcal{H}_{2^{M-1}}$ the general master equation, which amounts
to sum over $k$ the equations (\ref{eq:MasterEqk}) , one obtains
the following equation for $\xi\left(t,\omega\right)$ 
\begin{eqnarray*}
\partial_{t}\xi\left(t,\omega\right) & = & \frac{-i\omega}{2}\sum_{k}\left(M-2|k|\right)\left[S_{x},\tilde{\tau}_{k}\right]+\frac{\gamma}{2}\left(\xi\left(t,\omega\right)-S_{z}\xi\left(t,\omega\right)S_{z}\right).
\end{eqnarray*}
If now supposes that $M$ is large enough one can numerically check
that the sectors involved by the evolution on the time scale given
by $t_{opt}$ are only those such that $|k|\ll M$ then the previous
equation can be approximated by 
\begin{eqnarray*}
\partial_{t}\xi\left(t,\omega\right) & \approx & \frac{-i\omega}{2}M\left[S_{x},\xi\left(t,\omega\right)\right]+\frac{\gamma}{2}\left(\xi\left(t,\omega\right)-S_{z}\xi\left(t,\omega\right)S_{z}\right)
\end{eqnarray*}
and one directly sees that, while the coherent part acts with $MS_{x}$
the decoherence part is only proportional to $\gamma/2$. While the
latter is just a rough approximation, it gives an intuitive way to
understand why the estimation shows a super-SQL scaling: compared
to the parallel case, on the time scale given by $t_{opt}$ the single
qubit is only marginally touched by decoherence.

\section{Coherence, Quantum Phase Transitions and Estimation\label{sec:Coherence,-Quantum-Phase}}

We finally analyze the case of criticality-enhanced quantum estimation
processes. The theory is based on the fidelity approach to QPTs \cite{GuFidelityReview}.
The scenario is the following: suppose $H_{\lambda}=H_{0}+\lambda V,\ \lambda\in\mathbb{R}$
is a family of Hamiltonians. The corresponding manifold of ground
states $\left\{ \ket{0^{\lambda}}\right\} _{\lambda}$ of $H_{\lambda}$
can be adiabatically generated by means of the unitary operator \cite{GiordaQPTFidelity}
\[
O_{\lambda}=\sum_{k}\ket{n^{\lambda+\delta\lambda}}\bra{n^{\lambda}}
\]
where $\left\{ \ket{n^{\lambda}},E_{n}^{\lambda}\right\} $ are the
eigenstates and the corresponding eigenvalues of $H_{\lambda}$. The
QFI for a given $\lambda$ reads\cite{estimationcriticality}: 
\begin{equation}
QFI(\lambda)=4\sum_{n>0}\frac{\left|\bra{n^{\lambda}}V\ket{0^{\lambda}}\right|^{2}}{\left(E_{0}^{\lambda}-E_{n}^{\lambda}\right)^{2}}=4g_{\lambda}^{FS}\label{eq: QFI for CEQE}
\end{equation}
and the estimation precision is proportional to the Fubini-Study metric.
Suppose the system described by $H_{\lambda}$ undergoes a QPT when
$\lambda=\lambda_{c}$. If one aims at estimating the parameter $\lambda$
with the highest possible precision, in proximity of the critical
point $\lambda_{c}$ one can exploit the scaling behavior of the $QFI$
with respect the size of the system $L$. The scaling is determined
by the critical exponents that define the given QPT \cite{GuFidelityReview,estimationcriticality}.
Indeed 
\[
g_{\lambda}\sim L^{-\nu\Delta_{Q}+d}
\]
where $\Delta_{Q}=2\Delta_{V}-2\zeta-d$ is a function of the scaling
exponent $\Delta_{V}$ of the operator $V$ that drives the QPT, the
dynamical exponent $\zeta$ and the scaling exponent of the correlation
length $\nu$. If $\Delta_{V}$ is such that $\Delta_{Q}<0$ i.e.,
if $V$ is ``sufficiently'' relevant, $g_{\lambda}$ scales in a
super-extensive way and so does the $QFI$, thus allowing for an enhancement
of the estimation precision. If $V$ is ``insufficiently'' relevant
i.e., $\Delta_{V}$ is such that $\Delta_{Q}>0$, as for example in
Berezinskii-Kosterlitz-Thouless type of QPTs, one cannot take advantage
of the super-extensive behaviour to estimate $\lambda$. \\
It is easy to show that Equation (\ref{eq: QFI for CEQE}) is a particular
case of (\ref{eq: Cohernce vs Fubini Study -pure states}). Indeed,
given $|0^{\lambda+\delta\lambda}\rangle=O_{\lambda}|0^{\lambda}\rangle$
i.e., the ground state for $\lambda+\delta\lambda$, using perturbative
expansion $ $$|0^{\lambda+\delta\lambda}\rangle\approx|0^{\lambda}\rangle+\ket{\vec{v}^{\lambda}}$,
one writes the the first order correction in $\delta\lambda$ as 
\[
\ket{\vec{v}^{\lambda}}=\sum_{n>0}\frac{\bra{n^{\lambda}}V\ket{0^{\lambda}}}{\left(E_{0}^{\lambda}-E_{n}^{\lambda}\right)}\ket{n^{\lambda}}
\]
with $\ket{\hat{v}^{\lambda}}=\ket{\vec{v}^{\lambda}}/|\vec{v}|$.
Consider now the following orthonormal basis 
\[
\mathcal{B}_{0,v}=\left\{ \left(\ket{0^{\lambda}}\pm\ket{\hat{v}^{\lambda}}\right)/\sqrt{2}\right\} \bigcup\left\{ \ket{\alpha_{n}}\right\} _{n=2}^{L^{d}}
\]
where $\left\{ \ket{\alpha_{n}}\right\} _{n=2}^{L^{d}}$ is a generic
set of orthonormal vectors. Then, as proven in Appendix \ref{sec: Appendix Coherence-and-QPTs},
\begin{equation}
4g_{\lambda}^{FS}=\left[-\partial_{\delta\lambda}^{2}Coh_{\mathcal{B}_{0,v}}\left(\ket{0^{\lambda}}\right)\right]_{\delta\lambda=0}
\end{equation}
i.e., the metric coincides with the second order variation of the
coherence of $\mathcal{B}_{0,v}$ with respect to the ground state
$\ket{0^{\lambda}}$. In accordance with (\ref{eq: Cohernce vs Fubini Study -pure states})
the first consequence of the previous result is that the geometry
of the manifold of ground states is determined by the coherence properties
of $\mathcal{B}_{0,v}$. Secondly, the non-analiticities and scaling
properties of $g_{\lambda}^{FS}$, that on one hand signal the presence
of a QPT and on the other hand are at the basis of the CEQE, are those
pertaining to the physical quantity $Coh_{\mathcal{B}_{0,v}}\left(\ket{0^{\lambda}}\right)$,
and in particular its second order variation. The latter is single-qubit
in nature since it pertains the subspace $span\left\{ \ket{0^{\lambda}},\ket{\hat{v}^{\lambda}}\right\} $.

Quantum Phase transitions can in general be signaled by several different
properties of the underlying system. For example, by focusing on subsystems
such as one- or two-site density matrices for spin chains, one can
find several QPTs signature by analyzing the non-analiticities of
correlations and coherence measures \cite{ChenFanQPT,Malvezzi,LiLinQPT,entanglement&QPTs,KarpatCoherenceQPTs}.
Our approach instead focuses on the single qubit subspace $span\left\{ \ket{0^{\lambda}},\ket{\hat{v}^{\lambda}}\right\} $
and on the variation of the relevant coherence impressed by $O_{\lambda}(V)$.
The fidelity approach fails to signal the QPTs and. correspondingly
the criticality does not allow to enhance the estimation precision
in CEQE, whenever the operator $V$ is ``insufficiently'' relevant
i.e., $\Delta_{V}$ is such that $\Delta_{Q}>0$. In our picture this
can be interpreted as the consequence of the fact that the variation
of the relevant coherence impressed by $O_{\lambda}(V)$ is in these
cases too weak and one cannot take advantage of the super-extensive
behavior to estimate $\lambda$ \cite{estimationcriticality}.

\section{Conclusions}

Coherence is one of the fundamental features that distinguish the
quantum from the classical realm. The perspective adopted in this
work allows to link coherence (its second order variation in a specific
basis) to the geometry of quantum states and their statistical distinguishability,
and thus to the Quantum Fisher Information. The connection allows
to establish a framework that encompasses a wide variety of single
parameter estimation processes: noiseless and noisy quantum phase
estimation based on single/multi-qubit probes, and criticality enhanced
quantum estimation. Overall our findings show how to quantify the
notion that coherence is the resource that must be engineered, controlled
and preserved in these quantum estimation processes.

As for quantum phase estimation, the use of specific factorizations
of the underlying quantum system, i.e., specific tensor product structures,
allows to express the Quantum Cramer-Rao bound to the estimation precision
in terms of two contributions: the Fisher Information of a single
qubit; and the second order variation of the classical correlations
between the observables defined by the main object of the theory i.e.,
the Symmetric Logarithmic Derivative. The adopted perspective thus
allows to discuss the role of (quantum) correlations in estimation
processes. In several relevant cases (quantum) correlations in the
state probe  are not intrinsically required and the estimation is
effectively equivalent to a  process based on a single qubit interacting
with a possibly much larger system. In particular we show how various
relevant protocols based on different strategies, such as multi-round
application of phase shifts to a single qubit or protocols based on
pure and mixed GHZ and NOON states, are formally equivalent and are
based on the exploitation of the very same resource: the variation
of a single qubit coherence. In doing so we provide an example of
$M$-qubit based evolution in which the Heisenberg limit in the estimation
precision can be attained with the use of an uncorrelated $M$-qubit
probe state. 

As for noisy estimation processes, we have focused on a prototypical
example based on GHZ states that achieves a quasi-Heisenberg scaling
of the precision. We have discussed the protocol and we have shown
how within our perspective: $i)$ the estimation procedure can be
described as a parallelization of the single qubit multi-round strategy,
where the quasi-Heisenberg scaling is rooted in single-qubit coherences
variations; $ii)$ one can analytically derive, even for complex multi-qubit
noisy evolutions, meaningful lower bounds to the Quantum Fisher Information
that allow to infer its scaling behaviour. The approach is suitable
to be extended and applied to other relevant noisy estimation processes.

We have finally discussed criticality-enhanced quantum estimation
processes. In doing so we have recognized the role of a specific kind
of (global) coherence in quantum phase transitions. The non-analiticities
of such coherence are at the basis of the sensitivity scaling of criticality-enhanced
estimation protocols and they correspond to the global signatures
of zero-temperature phase transitions found within the fidelity approach.

While in laying down our framework we have focused on particular kinds
of single-parameter estimation protocols, our approach is suitable
to be extended to more general estimation processes.

\acknowledgements

We thank Dr. Marco Genoni for his precious feedback. 

We thank Dr. Giorgio Villosio for his illuminating comments as well
as his enduring hospitality at the Institute for Women and Religion,
Turin (``oblivio c{*}e soli a recta via nos avertere possunt'').

\section{Fubini-Study metric and coherence\label{sec: Appendix Fubini-Study-metric-and}}

We show below how the Fubini-Study (FS) metric can be related to variation
of the coherence of a generic eigenbasis $\mathcal{B}_{\boldsymbol{\mathbf{\alpha}}}^{\lambda}$
of $L_{\lambda}$. Suppose one has a one-parameter family of states
$\{|\psi_{\lambda}\rangle\in\mathbb{C}^{N},\ \lambda\in\mathcal{I\subset\mathbb{R}}\}$
and one wants to estimate $\lambda$. The estimation problem was solved
in \cite{BraunsteinCaves} as follows. Assume that the states $|\psi_{\lambda}\rangle$
are normalized, and the curve $|\psi_{\lambda}\rangle$ in $\mathcal{H}_{N}$
is of class $C_{2}$. Then, in a infinitesimal neighborhood $\lambda+d\lambda$
around to the generic $\lambda$, one can expand 
\begin{equation}
|\psi_{\lambda+d\lambda}\rangle=|0\rangle+|v\rangle d\lambda+|w\rangle d\lambda^{2}+\mathcal{O}(d\lambda^{3})
\end{equation}
with $|0\rangle\equiv|\psi_{\lambda}\rangle$, $|v\rangle\equiv\left(\frac{d}{d\lambda}|\psi_{\lambda+d\lambda}\rangle\right)_{d\lambda=0}$
and $|w\rangle\equiv\left(\frac{d^{2}}{d\lambda^{2}}|\psi_{\lambda+d\lambda}\rangle\right)_{d\lambda=0}$.
In Ref. \cite{BraunsteinCaves}, it was shown that the SLD in $d\lambda=0$
can be written as 
\begin{equation}
L_{\lambda}=|0\rangle\langle v^{\perp}|+|v^{\perp}\rangle\langle0|
\end{equation}
where $|v^{\perp}\rangle=\ket{v}-\langle0|v\rangle|0\rangle$, and
one gets the quantum Fisher information (QFI) 
\begin{equation}
QFI=4\langle v^{\perp}|v^{\perp}\rangle=4\big(\langle v|v\rangle-|\langle v|0\rangle|^{2}\big)
\end{equation}
which is seen to coincide with the Fubini-Study metric\cite{Uhlmann}.
The latter provides both the geometric distance and a measure of the
statistical distinguishability of two infinitesimally closed pure
states. Due to the form of $L_{\lambda}$ the estimation problem pertains
a single qubit subspace $\mathcal{H}_{2}=span\left\{ \ket{0},\ket{v^{\perp}}\right\} $.
The optimal measurement basis that allows to attain $QFI$ is uniquely
defined only in $\mathcal{H}_{2}$ where 
\begin{equation}
|\pm\rangle=\frac{1}{\sqrt{2}}(|0\rangle\pm\frac{1}{\langle v{}^{\perp}|v^{\perp}\rangle^{1/2}}|v^{\perp}\rangle)
\end{equation}
give the eigenbasis $\mathcal{B}_{\pm}=\left\{ |\pm\rangle\right\} $
of $L_{\lambda}$ pertaining to the only non zero eigevalues $\pm\left|\langle v^{\perp}|v^{\perp}\rangle\right|$
(in the following we drop for simplicity the eigenvectors' dependence
on $\lambda$). Suppose now we choose a generic basis for the kernel
of $L_{\lambda}$ $\mathcal{B}_{Ker}=\left\{ \ket{n}\right\} _{n=3}^{N}$
such that $\left\langle n|\pm\right\rangle =0,\ \forall n$. Then
for all bases $\mathcal{B}_{\boldsymbol{\mathbf{\alpha}}}^{\lambda}=\mathcal{B}_{\pm}\bigcup\mathcal{B}_{Ker}$
of the whole of $\mathcal{H}_{N}$ we have that the probabilities
$p_{\pm}^{\lambda+d\lambda}=\left|\left\langle \pm|\psi^{\lambda+d\lambda}\right\rangle \right|^{2}$
evaluated up to order $\mathcal{O}(d\lambda^{2})$ read 
\begin{equation}
p_{\pm}^{\lambda+d\lambda}=\frac{1}{2}(1\pm2\langle v^{\perp}|v^{\perp}\rangle^{1/2}d\lambda\pm2Re\frac{\langle w|v^{\perp}\rangle}{\langle v^{\perp}|v^{\perp}\rangle^{1/2}}d\lambda^{2}+\mathcal{O}(d\lambda^{3}))
\end{equation}
where we have used $\langle v|v{}^{\perp}\rangle=\langle v^{\perp}|v{}^{\perp}\rangle$
and the conditions $Re[\langle v|0\rangle]=0$ and $2Re\langle w|0\rangle=-\langle v|v\rangle$
(implied by the normalization condition $\langle\psi_{\lambda}|\psi_{\lambda}\rangle=1$
at first and second order in $d\lambda$). Thus, we obtain 
\begin{equation}
\left(p_{\pm}^{\lambda+d\lambda}\right)_{d\lambda=0}=p_{\pm}^{\lambda}=1/2,\quad\left(\partial_{d\lambda}p_{\pm}^{\lambda+d\lambda}\right)_{d\lambda=0}=\pm\langle v^{\perp}|v^{\perp}\rangle^{1/2},\quad\left(\partial^{2}p_{\pm}^{\lambda+d\lambda}\right)_{d\lambda=0}=\pm Re\frac{\langle w|v^{\perp}\rangle}{\langle v^{\perp}|v^{\perp}\rangle^{1/2}}
\end{equation}
As for $\mathcal{B}_{Ker}$ one has $p_{n}^{\lambda+d\lambda}=|\langle\psi^{\lambda}|n\rangle|^{2}=\mathcal{O}(d\lambda^{4})$
and what matters is that they are $o(d\lambda^{2})$. Consequently,
if one considers the coherence function $Coh_{\mathcal{B}_{\boldsymbol{\mathbf{\alpha}}}}\left(|\psi{}_{\lambda}\rangle\right)$
for a generic $\mathcal{B}_{\boldsymbol{\mathbf{\alpha}}}^{\lambda}$
one has the two relations: 
\begin{equation}
\left[\partial_{d\lambda}Coh_{\mathcal{B}_{\boldsymbol{\mathbf{\alpha}}}^{\lambda}}\left(|\psi{}_{\lambda+d\lambda}\rangle\right)\right]_{d\lambda=0}=0
\end{equation}

\begin{equation}
-\left[\partial_{d\lambda}^{2}Coh_{\mathcal{B}_{\boldsymbol{\mathbf{\alpha}}}^{\lambda}}\left(|\psi{}_{\lambda+\lambda}\rangle\right)\right]_{d\lambda=0}=\sum_{i=\pm}\frac{\left(\partial_{d\lambda}p_{\pm}^{\lambda+d\lambda}\right)_{d\lambda=0}}{p_{i}^{\lambda}}=QFI
\end{equation}
while $f\left(|\psi{}_{\lambda}\rangle\langle\psi_{\lambda}|\right)=\sum_{i=\pm}\left(\partial^{2}p_{i}^{\lambda+d\lambda}\right)_{d\lambda=0}\log_{2}p_{i}^{\lambda}=0$.
Therefore the FS metric can in general be expressed as a curvature
of the coherence $Coh_{\mathcal{B}_{\boldsymbol{\mathbf{\alpha}}}^{\lambda}}$
of a generic eigenbasis $\mathcal{B}_{\boldsymbol{\mathbf{\alpha}}}^{\lambda}$
of $L_{\lambda}$ with respect to $|\psi{}_{\lambda}\rangle$ around
a maximum. Thus, in terms of coherence, what matters for the estimation
process and for the statistical distinguishability between two neighboring
pure states is indeed the variation of the coherence within the single-qubit
subspace $\mathcal{H}_{2}$ spanned by $\mathcal{B}_{\pm}=\left\{ |\pm\rangle\right\} $.

\section{Results for $N=2$\label{sec:Appendix II -Results for N=00003D00003D2}}

In this section, we derive the results presented in the main text
for the qubit case $N=2$. Without loss of generality we choose the
single qubit state 
\begin{equation}
\rho_{0}=(1+\vec{z}\cdot\boldsymbol{\sigma})/2
\end{equation}
where $\vec{z}=z\hat{z}=z(0,0,1)$, $0\le z\le1$ and $\boldsymbol{\sigma}=\left(\sigma_{x},\sigma_{y},\sigma_{z}\right)$
is the vector of Pauli matrices. The phase generator is 
\begin{equation}
G=\gamma\left(\hat{\gamma}\cdot\boldsymbol{\sigma}\right)
\end{equation}
with $\hat{\gamma}=(\sin\delta,0,\cos\delta)$, such that its eigenbasis
lies in the $\hat{x}\hat{z}$ plane, forming an angle $0\leq\delta\leq\frac{\pi}{2}$
with $\hat{z}$. The strength of $G$ is measured by its norm $Tr[G^{2}]=2\gamma^{2}$,
where $\gamma>0$. A generic measurement basis $\mathcal{B}_{\hat{b}}$
is defined by the projectors 
\begin{equation}
\Pi_{\pm}^{\hat{b}}=(1\pm\hat{b}\cdot\boldsymbol{\sigma})/2
\end{equation}
with $\hat{b}=\{\sin\theta\cos\phi,\sin\theta\sin\phi,\cos\theta\}$
a generic Bloch vector. The state 
\begin{equation}
\rho_{\lambda}=e^{-i\lambda\vec{\gamma}\cdot\boldsymbol{\sigma}}\rho_{0}e^{i\lambda\vec{\gamma}\cdot\boldsymbol{\sigma}}
\end{equation}
is given as $\rho_{\lambda}=(1+\vec{z}_{\lambda}\cdot\boldsymbol{\sigma})/2$
with 
\begin{equation}
\vec{z}_{\lambda}=\cos2\gamma\lambda\ \vec{z}+\sin2\gamma\lambda\ (\vec{z}\times\vec{\gamma})+(1-\cos2\gamma\lambda)\vec{\gamma}(\hat{\gamma}\cdot\vec{z})
\end{equation}
Thus, the probabilities $p_{\pm}^{\hat{b}}(\lambda)=Tr[\rho_{\lambda}\Pi_{\pm}^{\hat{b}}]$
are obtained as 
\begin{equation}
p_{\pm}^{\hat{b}}(\lambda)=1/2Tr[(1\pm\hat{b}\cdot\boldsymbol{\sigma})\rho_{\lambda}]=1/2\Big(1\pm\hat{b}\cdot\vec{z}_{\lambda}\Big)
\end{equation}
Their derivatives are computed as 
\begin{equation}
\partial_{\lambda}p_{\pm}^{\hat{b}}(\lambda)=\pm1/2\hat{b}\cdot\partial_{\lambda}\vec{z}_{\lambda},\qquad\partial_{\lambda}^{2}p_{\pm}^{\hat{b}}(\lambda)=\pm1/2\hat{b}\cdot\partial_{\lambda}^{2}\vec{z}_{\lambda}
\end{equation}
with $\partial_{\lambda}\vec{z}_{\lambda}=2\gamma(-\sin2\gamma\lambda\ \vec{z}+\cos2\gamma\lambda\ (\vec{z}\times\vec{\gamma}))$
and $\partial_{\lambda}^{2}\vec{z}_{\lambda}=4\gamma^{2}(-\cos2\gamma\lambda\ \vec{z}-\sin2\gamma\lambda\ (\vec{z}\times\hat{\gamma}))$.
In $\lambda=0$, one gets 
\begin{eqnarray}
p_{\pm}^{\hat{b}}(\lambda=0)= & 1/2\Big(1\pm\hat{b}\cdot\vec{z}\Big)\label{eq: probsqubit}\\
(\partial_{\lambda}p_{\pm}^{\hat{b}}){}_{\lambda=0}= & \pm\gamma(\vec{z}\times\hat{\gamma})\cdot\hat{b} & =\pm(\vec{z}\times\vec{\gamma})\cdot\hat{b}\label{eq:probsqubit2}\\
(\partial_{\lambda}^{2}p_{\pm}^{\hat{b}}){}_{\lambda=0}= & \pm2\gamma^{2}\ \hat{b}\cdot\vec{z}\label{eq:probsqubit3}
\end{eqnarray}
The Fisher Information for $\mathcal{B}_{\theta,\phi}$ in $\lambda=0$
is computed as 
\begin{eqnarray*}
FI(\mathcal{B}_{\hat{b}},\rho_{0},G)= & \frac{((\partial_{\lambda}p_{+}^{\hat{b}})_{\lambda=0})^{2}}{p_{+}^{\hat{b}}(\lambda=0)}+\frac{((\partial_{\lambda}p_{-}^{\hat{b}})_{\lambda=0})^{2}}{p_{-}^{\hat{b}}(\lambda=0)} & =\\
= & 2((\vec{z}\times\vec{\gamma})\cdot\hat{b})^{2}\Big(\frac{1}{\Big(1+\hat{b}\cdot\vec{z}\Big)}+\frac{1}{\Big(1-\hat{b}\cdot\vec{z}\Big)}\Big) & =4\frac{(\vec{\gamma}\times\vec{z}\cdot\hat{b})^{2}}{1-(\vec{z}\cdot\hat{b})^{2}}
\end{eqnarray*}
In terms of $\theta,\phi$ and $\delta$ the latter can be written
as 
\[
FI(\mathcal{B}_{\theta,\phi},\rho_{0},G)=4\gamma^{2}z^{2}\sin^{2}\delta\frac{\sin^{2}\theta\sin^{2}\phi}{1-z^{2}\cos^{2}\theta}
\]
For mixes states $z<1$, the latter is maximized when $\hat{z},\hat{\gamma},\hat{b}$
form an orthogonal triple. The maximization over $\hat{b}$ can be
performed by finding the maximum with respect to $\theta,\phi$ .
The maximum over $\phi$ is obviously $\phi=\pi/2$, while the maximum
over $\theta$ can be easily found by computing the critical points
of $\frac{\sin^{2}\theta}{1-z^{2}\cos^{2}\theta}$, which gives $\theta=\pi/2$.
Therefore, the maximization over $\hat{b}$ results in the choice
$\hat{b}=\{0,1,0\}\propto\hat{\gamma}\times\hat{z}$. In turn, the
symmetric logarithmic derivative in $\lambda=0$ can be shown to be
\begin{equation}
L_{0}=-2(\hat{\gamma}\times\hat{z})\cdot\boldsymbol{\sigma}
\end{equation}
Indeed, one has 
\begin{equation}
\partial_{\lambda}\rho_{\lambda}|_{\lambda=0}=-i[G,\rho_{0}]=\vec{\gamma}\times\vec{z}\cdot\boldsymbol{\sigma}
\end{equation}
and one can immediately verify that $\frac{1}{2}(L_{0}\varrho_{o}+\varrho_{o}L_{0})=\vec{\gamma}\times\vec{z}\cdot\boldsymbol{\sigma}$.
Thus, the eigenbasis of $L_{0}$ corresponds to $\hat{\alpha}=\vec{\gamma}\times\vec{z}/|\vec{\gamma}\times\vec{z}|=\{0,1,0\}$,
which coincides with the optimal measurement. \\

Furthermore, from the above formulas (\ref{eq: probsqubit},\ref{eq:probsqubit2},\ref{eq:probsqubit3}),
when considering the coherence function $Coh_{\mathcal{B}_{\alpha}}(\rho_{\lambda})=-\mathcal{V}(\rho_{\lambda})+\sum_{i=\pm}p_{i}^{\hat{\alpha}}\log p_{i}^{\hat{\alpha}}$,
one obtains with some simple algebra 
\begin{equation}
\left[\partial_{\lambda}Coh_{\mathcal{B}_{\alpha}}(\rho_{\lambda})\right]_{\lambda=0}=0,\qquad QFI=-\left(\partial_{\lambda}^{2}Coh_{\mathcal{B}_{\alpha}}(\rho_{\lambda})\right)_{\lambda=0}\label{eq: single mixed qubit QFI as a maximum}
\end{equation}
The above results hold in particular for the limiting case of pure
states ($z=1$), if one measures on the eigenbasis $\mathcal{B}_{\alpha}$
of the SLD. The latter is not the only basis that allows to attain
the $QFI.$ Indeed, for pure states the Fisher information for a generic
measurement basis is independent of the angle $\theta$ and reads
\[
FI(\mathcal{B}_{\theta,\phi},\rho_{0},G)=4\gamma^{2}z^{2}\sin^{2}\delta\sin^{2}\phi
\]
such that the bound $QFI$ can in principle be achieved by any basis
such that $\phi=\pi/2$. However, such bases in fact not all equivalent.
Indeed the logic of the estimation process as described within the
Crámer-Rao formalism is the following. In general one needs to know
in advance with some precision the value of $\lambda$. This can be
achieved with a (non-optimal) pre-estimation process on a subset of
the probes, which leads to a value $\lambda_{est}$. Then one applies
to the initial state the shift $U_{\delta\lambda}=\exp-i\delta\lambda G$
with $\delta\lambda=\lambda-\lambda_{est}\ll1$. Only then the choice
of the optimal measurement basis becomes meaningful; in particular,
if one supposes that $\delta\lambda=0$ the actual precision for a
generic basis $\mathcal{B}_{\hat{b}}$ is given by 
\begin{equation}
F(\mathcal{B}_{\hat{b}},\ket{\psi_{\delta\lambda}}\bra{\psi_{\delta\lambda}})\approx4\gamma^{2}\sin^{2}\phi-\gamma^{3}\left(16\cos^{2}\phi\sin\phi\cot\theta\right)\delta\lambda.\label{eq: Fisher pure states in deltalambda}
\end{equation}
The latter now \emph{does depend on} $\theta$, and if $\phi$ is
only approximately equal to $\pi/2$, for example due to imprecision
in the measurement apparatus, the bases corresponding to different
values of $\theta$ are no longer equivalent. For example, if $\theta\approx0,\pi$
it can happen that $F(\mathcal{B}_{\theta,\phi},\ket{\psi_{\delta\lambda}}\bra{\psi_{\delta\lambda}})\ll QFI$.
Instead, for $\theta\approx\pi/2$ i.e., $\mathcal{B}_{\hat{b}}\approx\mathcal{B}_{\hat{\alpha}}$,
this problem can be avoided. The choice of $\hat{b}=\hat{\alpha}$
becomes of fundamental importance when $Tr[G^{2}]$ is very large
(e.g. for estimation protocols that are based on a multi-round procedures,
where $\gamma\gg1$) since as shown in (\ref{eq: Fisher pure states in deltalambda})
the first order correction in $\delta\lambda$ would be amplified
by a factor $\gamma^{3}$. Furthermore, if the initial state is even
slightly impure ($z=1-\epsilon,\ \epsilon\ll1)$ the choice $\hat{b}=\hat{\alpha}$
becomes the only for which the bound can be fully attained. The above
reasoning can be summarized as follows. On one hand the condition
$\hat{b}=\hat{\alpha}$ guarantees that $F(\mathcal{B}_{\hat{b}},\ket{\psi_{\delta\lambda}}\bra{\psi_{\delta\lambda}})$
has a maximum equal to the QFI i.e., it guarantees the highest sensitivity
in the variation of $\lambda$. On the other hand, the condition $\hat{b}=\hat{\alpha}$
allows to have the lowest sensitivity with respect to small variations
of the measurement angles $\delta\theta,\delta\phi$ and the purity
$z\lesssim1$.

\section{SLD and coherence for N-dimensional states\label{sec Appendix :SLD-and-coherence}}

In the following we give the demonstration of result $3.1$ in Proposition
3. We will use the following notations: given $\ket{\alpha_{\pm,k}}$,
the eigenstates of $L_{0}$, we define the probabilities 
\begin{equation}
p_{\pm,k}^{\lambda}=\left\langle \alpha_{\pm,k}\left|\rho_{\lambda}\right|\alpha_{\pm,k}\right\rangle ,\quad p_{\pm}^{\lambda}=\sum_{k}p_{\pm,k}^{\lambda},\quad p_{k}^{\lambda}=\sum_{i=\pm}p_{i,k}^{\lambda}
\end{equation}
Under the following hypotheses: 
\begin{itemize}
\item $N$ is even; 
\item the initial diagonal state $\rho_{0}=\sum_{n}p_{n}\ket{n}\bra{n}$
is full rank 
\item $\langle n|G|m\rangle\in\mathbb{R}\ \forall n,m$ i.e., $G$ has purely
real matrix elements when expressed in the eigenbasis of $\rho_{0}$ 
\item $L_{\lambda=0}$ is full rank. 
\end{itemize}
Under the above hypotheses, it holds that: 
\begin{enumerate}
\item the eigenvalues of $L_{0}$ are opposite in pairs, 
\begin{equation}
L_{0}\ket{\alpha_{\pm,k}}=\pm\alpha_{k}\ket{\alpha_{\pm,k}}\label{eq:eigenvectorsL}
\end{equation}
and the Quantum Fisher Information reads 
\begin{eqnarray}
QFI & = & 2\sum_{k=1}^{N/2}\left(\alpha_{+,k}\right)^{2}p_{+,k}^{0}\label{eq: QFI - N-dim mixed states-1}
\end{eqnarray}

\item The coherence function of the eigenbasis $\mathcal{B}_{\pm,k}=\left\{ \ket{\alpha_{\pm,k}}\right\} _{k=1}^{N/2}$
with respect to the state $\rho_{\lambda}$ reads 
\begin{equation}
Coh_{\mathcal{B}_{\pm,k}}(\rho_{\lambda})=-V(\rho_{\lambda})-2\sum_{k}p_{+,k}^{\lambda}\log_{2}p_{+,k}^{\lambda}\label{eq: Coherence N-dim mixed-1}
\end{equation}
The Quantum Fisher Information is attained in correspondence of a
critical point of $Coh_{\mathcal{B}_{\pm,k}}(\rho_{\lambda})$ and
\begin{eqnarray*}
-\left[\partial^{2}Coh_{\mathcal{B}_{\pm,k}}(\rho_{\lambda})\right]_{\lambda=0} & = & QFI+f(\rho_{\lambda=0})
\end{eqnarray*}
with 
\begin{eqnarray*}
f(\rho_{\lambda=0}) & = & \sum_{k}\left(\partial_{\lambda}^{2}p_{k}^{\lambda}\right){}_{\lambda=0}\log p_{k}^{0}
\end{eqnarray*}

\end{enumerate}
\textbf{\emph{Proof}}\emph{. }We start by analyzing the eigendecomposition
of $L_{0}$, the symmetric logarithmic derivative in $\lambda=0$,
and proving 1). $L_{0}$ is in the form\cite{ParisEstimationBook}\emph{
\[
\langle n|L_{0}|m\rangle=2i\langle n|G|m\rangle(p_{m}-p_{n})/(p_{m}+p_{n})
\]
} \begin{proof} When $G$ is real in the eigenbasis of $\rho_{0}$,
$\langle n|G|m\rangle\in\mathbb{R}$, then $L_{0}$ is purely imaginary
($L_{0}=-L_{0}^{*}$), and thus it can be written as $L_{0}=i\tilde{L}_{0}$,
with $\tilde{L}_{0}$ real ($\langle n|\tilde{L}_{0}|m\rangle\in\mathbb{R}$)
and antisymmetric ($\langle n|\tilde{L}_{0}|m\rangle=-\langle m|\tilde{L}_{0}|n\rangle$).
Therefore, there exists a real orthogonal matrix $O,\ \langle n|O|m\rangle\in\mathbb{R}$
implementing a change of basis $\ket{l_{n}}=O\ket{n}$ such that $O^{T}\tilde{L}_{0}O$
is in a standard form (see e.g. \cite{skewsymmetric}), i.e., it is
block diagonal and composed by $N/2$ blocks of dimension $2\times2$
of the form 
\[
\left(\begin{array}{cc}
0 & -\alpha_{k}\\
\alpha_{k} & 0
\end{array}\right)=-\alpha_{k}\ket{l_{2k-1}}\bra{l_{2k}}+\alpha_{k}\ket{l_{2k-1}}\bra{l_{k}}
\]
Correspondingly, the SLD is in a block diagonal form $L_{0}=\oplus_{k}\left(\alpha_{k}\sigma_{y}^{k}\right)$
in the basis of the $\ket{l_{n}}$. Each block can be diagonalized
by the same kind of unitary transformation 
\[
\left(\mathbb{I}+i\sigma_{x}^{k}\right)/\sqrt{2}=\left[\ket{l_{2k-1}}\bra{l_{2k-1}}+\ket{l_{2k}}\bra{l_{2k}}+i\left(\ket{l_{2k-1}}\bra{l_{2k}}+\ket{l_{2k-1}}\bra{l_{2k}}\right)\right]/\sqrt{2}
\]
i.e., the whole matrix can be diagonalized by means of the block-diagonal
unitary operator $U=\frac{1}{\sqrt{2}}\oplus_{k}\left(\mathbb{I}+i\sigma_{x}^{k}\right)$.
The eigenvectors of $L_{0}$ can be expressed as 
\begin{equation}
\ket{\alpha_{\pm,k}}=(\ket{l_{2k-1}}\pm i\ket{l_{2k}})/\sqrt{2}
\end{equation}
with $k=1,..,N/2$ and we obtain the result in Eq. (\ref{eq:eigenvectorsL}),
\[
L_{0}\ket{\alpha_{\pm,k}}=\pm\alpha_{k}\ket{\alpha_{\pm,k}}
\]
The eigenvalues of $L_{0}$ are opposite in pairs, $\alpha_{\pm,k}=\pm\alpha_{k},\ \alpha_{k}>0$.
We now show that 
\begin{equation}
p_{+,k}^{0}=p_{-,k}^{0}\label{eq:ppluspminusk}
\end{equation}
The reality of $G$ implies that $\left[\rho_{0},G\right]^{T}=-\left[\rho_{0},G\right]$,
thus the commutator is itself anti-symmetric. Since the change of
basis $O$ is real and it preserves the anti-symmetry of $\left[\rho_{0},G\right]$,
the diagonal elements of $\left[\rho_{0},G\right]$ in the $\{\ket{l_{n}}\}$
basis are zero: 
\[
\bra{l_{2k-1}}\left[\rho_{0},G\right]\ket{l_{2k-1}}=0
\]
Therefore, taking into account that $\left[\rho_{0},G\right]=-i(L_{0}\rho+\rho L_{0})$,
we also have 
\[
\bra{l_{2k-1}}(L_{0}\rho+\rho L_{0})\ket{l_{2k-1}}=2Re\left\{ \bra{l_{2k-1}}L_{0}\rho_{0}\ket{l_{2k-1}}\right\} =0
\]
If we now express $\ket{l_{2k-1}}$ in terms of the respective $\ket{\alpha_{\pm,k}}$
we have that 
\begin{eqnarray*}
2Re\left\{ \bra{l_{2k-1}}L_{0}\rho_{0}\ket{l_{2k-1}}\right\}  & = & \alpha_{+}^{k}p_{+,k}^{0}-\alpha_{+}^{k}p_{-,k}^{0}=0
\end{eqnarray*}
since the ``cross term'' 
\[
\alpha_{+,k}\bra{\alpha_{+,k}}\rho_{0}\ket{\alpha_{-,k}}/2-\alpha_{+,k}\bra{\alpha_{-,k}}\rho_{0}\ket{\alpha_{+,k}}/2=i\alpha_{+,k}Im\left\{ \bra{\alpha_{+,k}}\rho_{0}\ket{\alpha_{-,k}}\right\} 
\]
is purely imaginary. Thus, we finally obtain the result in Eq. (\ref{eq:ppluspminusk})
\[
p_{+,k}^{0}=p_{-,k}^{0}
\]
From this result one can easily derive some relations for the marginal
probabilities $p_{\pm}=\sum_{k}p{}_{\pm,k}$ and $p_{k}=\sum_{i=\pm}p_{\pm,i}$.
Since $\sum_{k,i=\pm}p_{i,k}^{0}=2\sum_{k,i=\pm}p_{i,k}^{0}=2p_{+}^{0}=1$,
we get 
\begin{equation}
p_{+}^{0}=p_{-}^{0}=1/2\label{eq:ppluspminus}
\end{equation}

Moreover, 
\begin{equation}
p_{k}=\sum_{i=\pm}p_{i,k}^{0}=2p_{+,k}^{0}\label{eq: p0k}
\end{equation}

From Eqs. (\ref{eq:ppluspminus}) and (\ref{eq: p0k}) one also obtains
that the probability distribution is factorized in $\lambda=0$, 
\begin{equation}
p_{+,k}^{0}=p_{+}^{0}p_{k}^{0}\label{eq:pproduct}
\end{equation}
We are now ready to derive Eq. (\ref{eq: QFI - N-dim mixed states-1}).
Given $\Pi_{\pm,k}=\ket{\alpha_{\pm,k}}\bra{\alpha_{\pm,k}}$, the
derivatives of the $p_{\pm,k}^{\lambda}$ are 
\begin{eqnarray*}
\left(\partial_{\lambda}p_{\pm,k}^{\lambda}\right){}_{\lambda=0} & = & Tr\left\{ \Pi_{\pm,k}\left(\partial_{\lambda}\rho^{\lambda}\right){}_{\lambda=0}\right\} =iTr\left\{ \Pi_{\pm,k}\left[\rho_{0},G\right]\right\} =\\
=Re\left\{ Tr[\rho_{0}\Pi_{\pm,k}L_{0}]\right\}  & = & Re\left\{ \pm\alpha_{k}^{+}Tr[\rho_{0}\Pi_{\pm,k}]\right\} =\pm\alpha_{+,k}p_{\pm,k}^{0}
\end{eqnarray*}

so that 
\begin{equation}
\left(\partial_{\lambda}p_{\pm,k}^{\lambda}\right){}_{\lambda=0}=\pm\alpha_{+,k}p_{\pm,k}^{0}\label{eq:dlambdappluspminus}
\end{equation}
and the $QFI$ reads 
\begin{eqnarray*}
QFI & = & \sum_{i=\pm,k=1}^{N/2}\left(\partial_{\lambda}p_{\pm,k}^{\lambda}\right){}_{\lambda=0}^{2}/p_{i,k}^{0}=\\
 & = & 2\sum_{k=1}^{N/2}\left(\partial_{\lambda}p_{+,k}^{\lambda}\right){}_{\lambda=0}^{2}/p_{+,k}^{0}=\\
 & = & 2\sum_{k=1}^{N/2}\left(\alpha_{+,k}\right){}^{2}p_{+,k}^{0}
\end{eqnarray*}
where between the first and the second line we have used Eq. (\ref{eq: p0k})
and from the second to the third line Eq. (\ref{eq:dlambdappluspminus}).
\\

We now prove 2). The coherence function $Coh_{\mathcal{B}_{\pm,k}}(\rho_{\lambda})$
reads by definition 
\[
Coh_{\mathcal{B}_{\pm,k}}(\rho_{\lambda})=-V(\rho_{\lambda})-\sum_{k,i=\pm}p_{i,k}^{\lambda}\log_{2}p_{i,k}^{\lambda}
\]

By considering Eqs (\ref{eq:ppluspminusk}) and (\ref{eq:dlambdappluspminus}),
one obtains 
\begin{eqnarray*}
\left(\partial_{\lambda}Coh_{\mathcal{B}_{\pm,k}}(\rho_{\lambda})\right){}_{\lambda=0} & = & -\sum_{\pm,k}\left(\partial_{\lambda}p_{\pm,k}^{\lambda}\right){}_{\lambda=0}\log\left(p_{\pm,k}^{0}\right)\\
 & = & -\sum_{k}\alpha_{+,k}p_{+,k}^{0}\big(\log\left(p_{+,k}^{0}\right)-\log\left(p_{-,k}^{0}\right)\big)=0
\end{eqnarray*}

i.e., the coherence function for the basis $\mathcal{B}_{\pm,k}$
has a critical point in $\lambda=0$. 
\begin{eqnarray*}
\left[\partial_{\lambda}^{2}Coh_{\mathcal{B}_{\pm,k}}(\rho_{\lambda})\right]{}_{\lambda=0} & = & -\sum_{i=\pm,,k}\frac{\left(\partial_{\lambda}p_{i,k}^{\lambda}\right)_{\lambda=0}^{2}}{p_{i,k}^{0}}-\sum_{i=\pm,k}\left(\partial_{\lambda}^{2}p_{i,k}^{\lambda}\right)_{\lambda=0}\log\left(p_{i,k}^{0}\right)\\
 & = & -QFI-\sum_{,k}\left(\left(\partial_{\lambda}^{2}p_{+,k}^{\lambda}\right)_{\lambda=0}+\left(\partial_{\lambda}^{2}p_{-,k}^{\lambda}\right)_{\lambda=0}\right)\log\left(p_{+,k}^{0}\right)\\
 & = & -QFI-\sum_{k}\left(\partial_{\lambda}^{2}p_{k}^{\lambda}\right){}_{\lambda=0}\log p_{+,k}^{0}\\
 & = & -QFI-\sum_{k}\left(\partial_{\lambda}^{2}p_{k}^{\lambda}\right){}_{\lambda=0}\log p_{k}^{0}
\end{eqnarray*}
where: in the second last line we have used $p_{+,k}^{0}=p_{-,k}^{0}$;
in the third line $\left(\partial_{\lambda}^{2}p_{+,k}^{\lambda}\right)_{\lambda=0}+\left(\partial_{\lambda}^{2}p_{-,k}^{\lambda}\right)_{\lambda=0}=\left(\partial_{\lambda}^{2}\left(p_{+,k}^{\lambda}+p_{-,k}^{\lambda}\right)\right)=\left(\partial_{\lambda}^{2}p_{k}^{\lambda}\right)_{\lambda=0}$;
in the last line we have used $p_{+,k}^{0}=p_{+}^{0}p_{,k}^{0}=p_{k}^{0}/2$
and $\sum_{k}\left(\partial_{\lambda}^{2}p_{k}^{\lambda}\right){}_{\lambda=0}\log2=0$.
\end{proof}

\section{SLD-induced TPS for N-dimensional states\label{sec: Appendix SLD-induced-TPS-for N dimensional states}}

When is even $N$, $G$ has real matrix elements in the eigenbasis
of $\rho_{0}$ and $L_{0}$ is full rank the eigenstates of the SLD
provide a natural way to introduce a proper tensor product structure
$TPS^{R}$ that allows to relate the $QFI$ to a specific kind of
classical correlations.

We first show how the result $3.2)$ in Proposition 3 can be derived.

Given two subalgebras $\mathcal{A}_{A},\mathcal{A}_{B}\subset End(\mathcal{H})$
they induce a tensor product structure \cite{ObservableInduced} if
the following conditions are satisfied: i) independence, $[\mathcal{A}_{A},\mathcal{A}_{B}]=0$
ii) completeness, $\mathcal{A}_{A}\lor\mathcal{A}_{B}=End(\mathcal{H})$.
In our case, given the $N/2$ pairs of eigenstates of the SLD 
\[
\ket{\alpha_{\pm,k}}=\left(\ket{l_{2k-1}}\pm i\ket{l_{2k}}\right)/\sqrt{2},\ k=1,2,..,N/2
\]
one can identify a new TPS that will split the Hilbert in a ``qubit''
and an $N/2$-dimensional space, $\mathcal{H}_{N}\sim\mathcal{H}_{2}\tilde{\otimes}\mathcal{H}_{N/2}$
and will allow writing $\ket{\alpha_{\pm,k}}=\ket{\pm}\tilde{\otimes}\ket{k}$.
The subalgebras of Hermitian operators $\mathcal{A}_{2},\mathcal{A}_{N/2}$
acting locally on $\mathcal{H}_{2}$ and $\mathcal{H}_{N/2}$ are
identified as follows. We choose $\mathcal{A}_{2}=span\left\{ \sigma_{0},\sigma{}_{x},\sigma_{y},\sigma_{z}\right\} \cong u(2)$
where 
\begin{eqnarray}
\sigma_{x}\equiv\sum_{k=1}^{N/2}\left(\ket{l_{2k-1}}\bra{l_{2k}}+\ket{l_{2k}}\bra{l_{2k-1}}\right) & = & \sum_{k=1}^{N/2}\left(\ket{\alpha_{+,k}}\bra{\alpha_{-,k}}+\ket{\alpha_{-,k}}\bra{\alpha_{+,k}}\right)\nonumber \\
\sigma_{y}\equiv-i\sum_{k=1}^{N/2}\left(\ket{l_{2k-1}}\bra{l_{2k}}-\ket{l_{2k}}\bra{l_{2k-1}}\right) & = & \sum_{k=1}^{N/2}\left(\ket{\alpha_{+,k}}\bra{\alpha_{+,k}}-\ket{\alpha_{-,k}}\bra{\alpha_{-,k}}\right)\nonumber \\
\sigma_{z}\equiv\sum_{k=1}^{N/2}\left(\ket{l_{2k-1}}\bra{l_{2k-1}}-\ket{l_{2k}}\bra{l_{2k}}\right) & = & i\sum_{k=1}^{N/2}\left(\ket{\alpha_{+,k}}\bra{\alpha_{-,k}}-\ket{\alpha_{-,k}}\bra{\alpha_{+,k}}\right)\nonumber \\
\sigma_{0}\equiv\sum_{k=1}^{N/2}\left(\ket{\alpha_{+,k}}\bra{\alpha_{+,k}}+\ket{\alpha_{-,k}}\bra{\alpha_{-,k}}\right) & = & \sum_{k=1}^{N/2}\left(\ket{\alpha_{+,k}}\bra{\alpha_{+,k}}+\ket{\alpha_{-,k}}\bra{\alpha_{-,k}}\right)\label{eq: definition of sigmas in new TPS}
\end{eqnarray}
The other subalgebra $\mathcal{A}_{N/2}\cong u(N/2)$ can be constructed
in an analogous way by starting from the following general definition
of the operators that form a basis of $u(N/2)$ 
\[
\mathcal{A}_{N/2}\equiv span\left\{ \ket{k}\bra{h}+\ket{h}\bra{k},\quad-i\ket{k}\bra{h}+i\ket{h}\bra{k},\quad\ket{h}\bra{h},\ h\neq k=1,\dots,N/2\right\} 
\]
where, in order to adapt the result to our specific case one has to
use 
\[
\ket{k}\bra{h}\equiv\ket{l_{2k-1}}\bra{l_{2h-1}}+\ket{l_{2k}}\bra{l_{2h}}=\ket{\alpha_{+,k}}\bra{\alpha_{+,h}}+\ket{\alpha_{-,k}}\bra{\alpha_{-,h}},\qquad k\neq h\in{1,2,..,N/2}
\]

\begin{equation}
\ket{h}\bra{h}\equiv\ket{l_{2h-1}}\bra{l_{2h-1}}+\ket{l_{2h}}\bra{l_{2h}}=\ket{\alpha_{+,h}}\bra{\alpha_{+,h}}+\ket{\alpha_{-,h}}\bra{\alpha_{-,h}}\qquad h\in{1,2,..,N/2}\label{eq: operators kh in TPSR}
\end{equation}
One has that $\left[\mathcal{A}_{2},\mathcal{A}_{N/2}\right]=0$,
$\mathcal{A}_{1}\lor\mathcal{A}_{2}=u(N)$ and therefore these subalgebras
identify a well-defined TPS $\mathcal{H}_{N}\sim\mathcal{H}_{2}\tilde{\otimes}\mathcal{H}_{N/2}$,
correspondingly the SLD eigenvectors can be written as 
\begin{equation}
\ket{\alpha_{\pm,k}}=\ket{\pm}\tilde{\otimes}\ket{k}
\end{equation}
In the new TPS, we can write the operators in (\ref{eq: definition of sigmas in new TPS})
as 
\begin{eqnarray}
\sigma_{x} & = & S_{x}\tilde{\otimes}\mathbb{I}_{N/2}\qquad\sigma_{y}=S_{y}\tilde{\otimes}\mathbb{I}_{N/2}\qquad\sigma_{z}=S_{z}\tilde{\otimes}\mathbb{I}_{N/2}\\
\sigma_{0} & = & \mathbb{I}_{2}\tilde{\otimes}\mathbb{I}_{N/2}
\end{eqnarray}
where $S_{x},S_{y},S_{z}$ are Pauli operators acting on the single
qubit factor $\mathcal{H}_{2}$. The operators in (\ref{eq: operators kh in TPSR})
can be written as 
\begin{equation}
\ket{k}\bra{h}\rightarrow\mathbb{I}_{2}\tilde{\otimes}|k\rangle\langle h|,\qquad\ket{h}\bra{h}\rightarrow\mathbb{I}_{2}\tilde{\otimes}|h\rangle\langle h|
\end{equation}
and they for a basis for the Hermitian operators acting on $\mathcal{H}_{N/2}$.
For all $O_{2}\in\mathcal{A}_{2},O_{N/2}\in\mathcal{A}_{N/2}$ the
composition of the operators in $\mathcal{H}_{N}$ is given by $O_{2}O_{N/2}$
that now can be written as $O_{2}O_{N/2}\simeq O_{2}\tilde{\otimes}O_{N/2}$;
onto the basis states one has $O_{2}O_{N/2}\ket{\alpha_{\pm,k}}=O_{2}\ket{\pm}\otimes O_{N/2}\ket{k}$.

Before passing to the rest of the proof, we notice that even if the
full controllability of the single $End\left(\mathcal{H}_{2}\right),End\left(\mathcal{H}_{N/2}\right)$
is not practically at hand, in order to carry over the estimation
process one needs only to be able to implement the measurement process
identified by $\mathcal{B}_{\pm,k}$, which amounts to experimentally
observing the probabilities $p_{\pm,k}=Tr\left[\ket{\alpha_{\pm,k}}\bra{\alpha_{\pm,k}}\rho\right]=Tr\left[\Pi_{\pm}\otimes\Pi_{k}\rho\right]$
i.e., the joint probabilities of an experiment carried over onto the
entire $\mathcal{H}_{N}$. And thus the probabilities pertaining to
the local observables $\left(\Pi_{\pm}\otimes\mathbb{I}_{N/2},\mathbb{I}_{2}\otimes\Pi_{k}\right)$
i.e., the marginals $p_{\pm},p_{k}$, can be easily derived.

The SLD in the $TPS^{R}$ can be written in terms of the new product
basis $\mathcal{B}_{\mathbf{\alpha=\pm,k\textrm{}}}=\left\{ \ket{\pm}\tilde{\otimes}\ket{k}\right\} $
as 
\[
L_{0}=\sum_{k=1,..,N/2}\alpha_{+,k}\left(\Pi_{+}\tilde{\otimes}\Pi_{k}-\Pi_{-}\tilde{\otimes}\Pi_{k}\right)
\]
where 
\begin{eqnarray}
\Pi_{+}\tilde{\otimes}\mathbb{I}_{N/2} & \equiv & \sum_{k=1}^{N/2}\ket{\alpha_{+,k}}\bra{\alpha_{+,k}}=\label{eq: newtps}\\
 & = & \left[\mathbb{I}_{N}+\sum_{k=1}^{N/2}\left(\ket{\alpha_{+,k}}\bra{\alpha_{+,k}}-\ket{\alpha_{-,k}}\bra{\alpha_{-,k}}\right)\right]/2\\
 & = & \frac{\left(\mathbb{I}_{2}+S_{y}\right)}{2}\tilde{\otimes}\mathbb{I}_{N/2}
\end{eqnarray}
where in the second line we have used $\sum_{k=1}^{N/2}\ket{\alpha_{+,k}}\bra{\alpha_{+,k}}=\mathbb{I}_{N}-\sum_{k=1}^{N/2}\ket{\alpha_{-,k}}\bra{\alpha_{-,k}}$.
Analogously 
\begin{eqnarray*}
\Pi_{-}\tilde{\otimes}\mathbb{I}_{N/2} & = & \frac{\left(\mathbb{I}_{2}-S_{y}\right)}{2}\tilde{\otimes}\mathbb{I}_{N/2}.
\end{eqnarray*}
On the other hand

\begin{equation}
\mathbb{I}_{2}\otimes\Pi_{k}\equiv\sum_{i=\pm}\ket{\alpha_{i,k}}\bra{\alpha_{i,k}}=\mathbb{I}_{2}\tilde{\otimes}|k\rangle\langle k|
\end{equation}
Given the previous definitions, the probabilities defined in the previous
section read 
\begin{eqnarray*}
p_{\pm,k}^{\lambda} & = & \bra{\alpha_{\pm,k}}\rho_{\lambda}\ket{\alpha_{\pm,k}}=Tr[\Pi_{\pm}\otimes\Pi_{k}\rho_{\lambda}]\\
p_{\pm}^{\lambda} & = & \sum_{k=1}^{N/2}\bra{\alpha_{\pm,k}}\rho_{\lambda}\ket{\alpha_{\pm,k}}=Tr[\Pi_{\pm}\otimes\mathbb{I}_{N/2}\rho_{\lambda}]\\
p_{k}^{\lambda} & = & \sum_{i=\pm}\bra{\alpha_{i,k}}\rho_{\lambda}\ket{\alpha_{i,k}}=Tr[\mathbb{I}_{2}\otimes\Pi_{k}\rho_{\lambda}]
\end{eqnarray*}
and they correspond to an experiment with joint $\left(\Pi_{\pm}\tilde{\otimes}\Pi_{k}\right)$
vs local $\left(\Pi_{\pm}\tilde{\otimes}\mathbb{I}_{N/2}\right),$$\left(\mathbb{I}_{2}\tilde{\otimes}\Pi_{k}\right)$
measurements onto $\rho_{\lambda}$. In general, the set of probabilities
$p_{\pm,k}^{\lambda},p_{\pm}^{\lambda},p_{,k}^{\lambda}$ are \textit{those
generated by the measurement of any observable $O$ commuting with
}$L_{0}$ onto $\rho_{\lambda}$. And the correlations relative to
those observables can be expressed by the mutual information.

We are now ready to derive result $4.3$ in Proposition 4. Given the
definition of mutual information 
\begin{equation}
\mathcal{M}_{L_{0}}^{\lambda}\equiv\mathcal{H}(p_{\pm}^{\lambda})+\mathcal{H}(p_{k}^{\lambda})-\mathcal{H}(p_{\pm,k}^{\lambda})
\end{equation}
one has from Eq. (\ref{eq:pproduct}) that $p_{\pm,k}^{0}=p_{\pm}^{0}p_{k}^{0}$
and therefore 
\[
\mathcal{M}_{L_{0}}^{0}=\mathcal{H}(p_{\pm}^{0})+\mathcal{H}(p_{k}^{0})-\mathcal{H}(p_{\pm,k}^{0})=0
\]
i.e., the observables that commute with $L_{0}$ are uncorrelated.
In terms of the probabilities the coherence function can be written
as 
\begin{equation}
Coh_{\mathcal{B}_{\pm,k}}(\rho_{\lambda})=-\mathcal{V}(\rho_{\lambda})+\mathcal{H}(p_{\pm,k}^{\lambda})\label{eq: Coherence as a function of MI 1}
\end{equation}
or alternatively 
\begin{equation}
Coh_{\mathcal{B}_{\pm,k}}(\rho_{\lambda})=-\mathcal{V}(\rho_{\lambda})+\mathcal{H}(p_{\pm}^{\lambda})+\mathcal{H}(p_{k}^{\lambda})-\mathcal{M}_{L_{0}}^{\lambda}.\label{eq: Coherence as function of MI 2}
\end{equation}
If now one computes the $\left[\partial_{\lambda}^{2}Coh_{\mathcal{B}_{\pm,k}}(\rho_{\lambda})\right]_{\lambda=0}$,
from (\ref{eq: Coherence as a function of MI 1}) one has that 
\begin{eqnarray*}
\left[\partial^{2}Coh_{\mathcal{B}_{\pm,k}}(\rho_{\lambda})\right]_{\lambda=0} & = & -\sum_{k}\left(\partial^{2}p_{k}^{\lambda}\right)_{\lambda=0}\log p_{k}^{0}-QFI
\end{eqnarray*}
while from (\ref{eq: Coherence as function of MI 2}) 
\begin{eqnarray*}
\left[\partial^{2}Coh_{\mathcal{B}_{\pm,k}}(\rho_{\lambda})\right]_{\lambda=0} & = & -\sum_{i}\frac{\left(\partial p_{i}^{\lambda}\right)_{\lambda=0}^{2}}{p_{i}^{0}}-\sum_{k}\left(\partial^{2}p_{k}^{\lambda}\right)_{\lambda=0}\log p_{k}^{0}+\\
 & - & \left(\partial^{2}\mathcal{M}^{\lambda}\right)_{\lambda=0}
\end{eqnarray*}
The results can be obtained using the above found relations (\ref{eq:ppluspminusk}),
(\ref{eq:pproduct}), and (\ref{eq:dlambdappluspminus}), for the
probabilities and their derivatives in $\lambda=0$. Equating the
previous two expression for the second order derivative of the coherence
one obtains 
\begin{eqnarray*}
QFI & = & FI_{2}+\left(\partial_{\lambda}^{2}\mathcal{M}_{L_{0}}^{\lambda}\right)_{\lambda=0}\\
 & = & \sum_{i}\frac{\left(\partial_{\lambda}p_{i}^{\lambda}\right)_{\lambda=0}^{2}}{p_{i}^{0}}+\left(\partial_{\lambda}^{2}\mathcal{M}_{L_{0}}^{\lambda}\right)_{\lambda=0}
\end{eqnarray*}
i.e., the result (\ref{eq: QFI in terms of partial2 MI}) in the main
text. According to the latter, the $QFI$ is composed by two contributions.
The first term is the single qubit Fisher Information $FI_{2}$ that
one would obtain by measuring $\Pi_{\pm}$ onto the single qubit reduced
density matrix $\xi_{\lambda}=Tr_{N/2}\left[\rho_{\lambda}\right]$.
The second term is given by the second order variation of the mutual
information $\mathcal{M}_{L_{0}}^{\lambda}$ between the relevant
observables $O$ that commute with the SLD $L_{0}$. Since we now
that $\mathcal{M}_{L_{0}}^{\lambda=0}=0$, the point $\lambda=0$
is a minimum for $\mathcal{M}_{L_{0}}^{\lambda}$ and therefore $\left(\partial_{\lambda}^{2}\mathcal{M}_{L_{0}}^{\lambda}\right)_{\lambda=0}>0$.

\section{$N$-dimensional mixed states maximal $QFI$\label{sec: Appendix Examples N dimensional-mixed states Maximal QFI}}

Suppose $\rho_{0}=\sum_{n}p_{n}\ket{n}\bra{n}$ is diagonal and the
$p_{n}$ are in \textit{decreasing order}. The QFI for general $N$-dimensional
mixed states reads 
\[
QFI=2\sum_{i\neq j}\frac{(p_{i}-p_{j})^{2}}{(p_{i}+p_{j})}|G_{ij}|^{2}=4\sum_{i<j}\frac{(p_{i}-p_{j})^{2}}{(p_{i}+p_{j})}|G_{ij}|^{2}
\]

\begin{prop} The problem of optimizing the QFI over all $G$ such
that $Tr[G^{2}]=\sum_{ij}|G_{ij}|^{2}\leq2\gamma^{2}$ has the following
solution: 
\[
\max_{Tr[G^{2}]\leq g^{2}}QFI=4\gamma^{2}\frac{(p_{1}-p_{N})^{2}}{(p_{1}+p_{N})}
\]
where the optimal $G$ has $|G_{1N}|=|G{}_{N1}|=\gamma$ and all the
remaining $|G|_{ij}=0$ (including $G_{ii}$). \end{prop} \begin{proof}
Since the QFI only depends on the off-diagonal terms of $G$ when
represented in the $\rho_{0}$ eigenbasis, the optimization can be
done by considering operators $G$ such that, in the same basis, $G_{ii}=0,\ i=1,..N$.
Then the optimization problem can be written as follows: 
\[
\max\sum_{k=1}^{M}a_{k}x_{k}\quad\mbox{over}\sum_{k=1}^{M}x_{k}\leq\gamma^{2},x_{k}>0
\]
where $M=N(N-1)/2$; $\{a_{k}\}=\left\{ \frac{(p_{i}-p_{j})^{2}}{(p_{i}+p_{j})},i<j\right\} $
for $1\leq k\leq M$; and $\{x_{k}\}=\{|G_{ij}|^{2},i<j\}$ for $1\leq k\leq M$.
This is a simple \emph{linear program}\cite{Boyd}.\emph{ }The optimal
solution is found on a vertex of the feasible region defined by $\sum_{k=1}^{M}x_{k}\leq\gamma^{2},x_{k}>0$.
The vertices are the $M$ points $v_{1}=\{x_{1}=\gamma^{2},0,\dots,0\},\dots,v_{M}=\{0,\dots,0,x_{M}=\gamma^{2}\}$.
The maximum is the found at $v_{\ell}$ where $a_{\ell}=\max a_{k}$
and it is unique if $\max a_{k}$ is unique. We then have 
\[
\max_{Tr[G^{2}]\leq g^{2}/2}QFI=4\gamma^{2}\max_{ij}\left(\frac{(p_{i}-p_{j})^{2}}{(p_{i}+p_{j})}\right)
\]
It can be easily seen that $\max_{ij}\Big(\frac{(p_{i}-p_{j})^{2}}{(p_{i}+p_{j})}\Big)=\frac{(p_{1}-p_{N})^{2}}{(p_{1}+p_{N})}$
. Indeed for each pair $i,j$ with one has $\frac{(p_{i}-p_{j})^{2}}{(p_{i}+p_{j})}=p_{i}\frac{(1-x)^{2}}{(1+x)}$
where $x=p_{j}/p_{i}$ and we assume (without restriction of generality)
that $p_{i}>p_{j}$. Now, $\frac{(1-x)^{2}}{(1+x)}$ is a monotonically
decreasing function of $x$, so it attains it maximum for the minimum
$x$, given by $p_{N}/p_{1}$. Then, since $p_{i}\leq p_{1}$, we
have $\frac{(p_{i}-p_{j})^{2}}{(p_{i}+p_{j})}\leq\frac{(p_{1}-p_{N})^{2}}{(p_{1}+p_{N})}$.
$\square$ \end{proof} Let us now assume that the dimension $N$
is even. The optimal $G$ is $G_{1N}=\gamma$, which corresponds to
$G=\gamma\sigma_{x}$ in the $|1\rangle,|N\rangle$ subspace. We have
\[
L_{0}=i\gamma\frac{(p_{1}-p_{N})}{(p_{1}+p_{N})}(|1\rangle\langle N|-|N\rangle\langle1|)
\]
The eigenvalues of $L_{0}$ are 
\[
\alpha_{\pm,1}=\pm\gamma\frac{(p_{1}-p_{N})}{(p_{1}+p_{N})},\alpha_{\pm,k}=0\ \forall k=2,\dots,N/2
\]
As for the optimal measurement basis $\mathcal{B}_{\boldsymbol{\mathbf{\alpha}}}=\left\{ |\alpha_{\pm,1}\rangle\right\} \bigcup\left\{ |\alpha_{\pm,k}\rangle\right\} _{k=2}^{N/2}$
one has that 
\[
|\alpha_{\pm,1}\rangle=\frac{1}{\sqrt{2}}(|1\rangle\mp i|N\rangle)
\]
while since the kernel of $L_{0}$ has dimension $N-2$, one has a
lot of freedom in the choice of remaining part of the basis $\left\{ |\alpha_{\pm,k}\rangle\right\} _{k=2}^{N/2}$.
Whatever the choice of $\mathcal{B}_{\boldsymbol{\mathbf{\alpha}}}$
one has 
\[
p_{\pm,1}^{\lambda}=\frac{p_{1}+p_{N}}{2}\pm\frac{p_{1}-p_{N}}{2}\sin(2\lambda\gamma)
\]
and since $p_{\pm,k}^{\lambda}$ are independent of $\lambda$ for
any $k\ge2$, it follows that 
\begin{eqnarray*}
\left[\partial_{\lambda}Coh_{\mathcal{B}_{\boldsymbol{\mathbf{\alpha}}}}(\rho_{\lambda})\right]_{\lambda=0} & = & \sum_{i=\pm}\left(\partial_{\lambda}p_{i,1}^{\lambda}\right)_{\lambda=0}\log p_{i,1}^{0}=0
\end{eqnarray*}
i.e., the coherence has a critical point in $\lambda=0$. Moreover,
since $\left(\partial_{\lambda}^{2}p_{\pm,k}^{\lambda}\right)_{\lambda=0}=0\ \forall k$,
we get 
\begin{equation}
-\left[\partial_{\lambda}^{2}Coh_{\mathcal{B}_{\boldsymbol{\mathbf{\alpha}}}}(\rho_{\lambda})\right]_{\lambda=0}=QFI=4g_{\lambda}^{Bures}
\end{equation}
Therefore the $QFI$ is identically equal to the second order variation
of $Coh_{\mathcal{B}_{\boldsymbol{\mathbf{\alpha}}}}(\rho_{\lambda})$.
As for the decomposition of $QFI$ (\ref{eq: QFI in terms of partial2 MI})
will vary depending on the choice of the kernel's basis, since the
value of $p_{\pm}^{0}=\sum_{k}p_{\pm,k}^{0}$depends on the actual
choice and 
\begin{eqnarray*}
FI_{2} & = & \left(\partial_{\lambda}p_{+,1}^{\lambda}\right)_{\lambda=0}^{2}\left(\frac{1}{p_{+}^{0}\left(1-p_{+}^{0}\right)}\right)\\
\left(\partial_{\lambda}^{2}\mathcal{M}_{L_{0}}^{\lambda}\right)_{\lambda=0} & =QFI-FI_{2}= & \left(\partial_{\lambda}p_{+,1}^{\lambda}\right)_{\lambda=0}^{2}\left(\frac{2}{p_{+,1}^{0}}-\frac{1}{p_{+}^{0}\left(1-p_{+}^{0}\right)}\right)
\end{eqnarray*}
Since $p_{+}^{0}\geq p_{+,1}^{0}$, and $p_{+}^{0}\left(1-p_{+}^{0}\right)\leq1/4$
one has 
\[
4\left(\partial_{\lambda}p_{+,1}^{\lambda}\right)_{\lambda=0}^{2}\leq FI_{2}\leq\left(\partial_{\lambda}p_{+,1}^{\lambda}\right)_{\lambda=0}^{2}\left(\frac{1}{p_{+,1}^{0}\left(1-p_{+,1}^{0}\right)}\right)=QFI
\]
\[
\left(\partial_{\lambda}p_{+,1}^{\lambda}\right)_{\lambda=0}^{2}\left(\frac{2}{p_{+,1}^{0}}-4\right)\geq\left(\partial_{\lambda}^{2}\mathcal{M}_{L_{0}}^{\lambda}\right)_{\lambda=0}\geq0
\]
For all bases $\mathcal{B}_{\boldsymbol{\mathbf{\alpha}}}$ one always
has $\left(\partial_{\lambda}^{2}\mathcal{M}_{L_{0}}^{\lambda}\right)_{\lambda=0}>0$.
In particular, if the choice is such that $p_{\pm}^{0}=1/2$ one has
that the single qubit contribution $FI_{2}$ is minimal while $\left(\partial_{\lambda}^{2}\mathcal{M}_{L_{0}}^{\lambda}\right)_{\lambda=0}$
is maximal. For example one can choose for the kernel of $L_{0}$
the basis

\begin{eqnarray*}
|\alpha_{\pm,k}\rangle & = & \frac{1}{\sqrt{2}}(|2k-2\rangle\pm|2k-1\rangle)\ k=2,\dots,N/2
\end{eqnarray*}
where $|2k-2\rangle,|2k-1\rangle\ k=2,\dots,N/2$ are eigenstates
of $\rho_{0}$. Accordingly one can define the $TPS^{R}$ (\ref{eq: newtps})
$\mathcal{H}\sim\mathcal{H}_{2}\tilde{\otimes}\mathcal{H}_{N/2}$.
With respect with such representation the state reads 
\[
\rho=\sum_{k}p_{k}(\mathbb{I}_{2}+h_{k}S_{z})\tilde{\otimes}|k\rangle\langle k|
\]
with $h_{k}=\frac{p_{2k-2}-p_{2k-1}}{p_{2k-2}+p_{2k-1}},k>1$ and
$h_{1}=\frac{p_{1}-p_{N}}{p_{1}+p_{N}}$. On the other hand $G=S_{x}\tilde{\otimes}|1\rangle\langle1|$
i.e., $G$ acts as a conditional rotation on the single qubit. The
probabilities for measurement in the defined $|\alpha_{\pm,k}\rangle$
basis are

\[
p_{\pm,1}^{\lambda}=\frac{p_{1}+p_{N}}{2}\pm\frac{p_{1}-p_{N}}{2}\sin2\lambda\gamma,\qquad p_{\pm,k}^{\lambda}=\frac{p_{2k-2}+p_{2k-1}}{2},\quad k=2,\dots,N/2
\]
Thus, one obtains 
\[
p_{\pm}^{\lambda}=\frac{1}{2}\pm\frac{p_{1}-p_{N}}{2}\sin2\gamma\lambda,\qquad p_{k=1}^{\lambda}=p_{1}+p_{N},\qquad p_{k>1}^{\lambda}=p_{2k-2}+p_{2k-1}
\]
from which we have $\partial_{\lambda}p_{\pm}^{\lambda}=\pm2\gamma\frac{p_{1}-p_{N}}{2}\cos\lambda\gamma$
and finally 
\[
FI_{2}=\sum_{i=\pm}\frac{\left(\partial_{\lambda}p_{i}^{\lambda}\right)_{\lambda=0}^{2}}{p_{i}^{0}}=4\gamma^{2}(p_{1}-p_{N})^{2}=(p_{1}+p_{N})QFI
\]

\[
\left(\partial_{\lambda}^{2}\mathcal{M}_{L_{0}}^{\lambda}\right)_{\lambda=0}=4\gamma^{2}(p_{1}-p_{N})^{2}(\frac{1-p_{1}-p_{N}}{p_{1}+p_{N}})=QFI(1-p_{1}-p_{N})
\]
i.e., the result reported in the main text. The (maximal) value $\left(\partial_{\lambda}^{2}\mathcal{M}_{L_{0}}^{\lambda}\right)_{\lambda=0}$
vanishes in the limit of $p_{1}\rightarrow1$ and $p_{n}\rightarrow0,\ \forall n>1$
i.e., in the limiting case of a pure state.

\subsection*{Class of separable states}

Consider the (separable but generally discordant) states 
\begin{equation}
\rho_{0}=\sum_{k=1}^{N}p_{k}\tau_{k}\otimes\ket{k}\bra{k}\label{eq: disc state 1}
\end{equation}
where $\tau_{k}=\left(\mathbb{I}+\vec{n}_{k}\cdot\vec{\sigma}\right)/2$
are pure states in the $xy$ plane, $\vec{n}_{k}=(\cos\delta_{k},\sin\delta_{k},0)$,
and $G=\sigma_{z}\otimes\mathbb{I}_{N/2}$ . The SLD reads $L_{0}=\oplus_{k}L_{k}$
with $L_{k}=2\hat{\alpha}_{k}\cdot\vec{\sigma}$ and $\hat{\alpha}_{k}=\hat{n}_{k}\times\hat{z}$.
The eigenvectors of $L_{0}$ are 
\[
|\alpha_{\pm,k}\rangle=|\pm\hat{\alpha}_{k}\rangle\otimes|k\rangle
\]
where $|\pm\hat{\alpha}_{k}\rangle$ are the states corresponding
to the Bloch vectors $\pm\hat{\alpha}_{k}\cdot\vec{\sigma}$. The
$TPS^{R}$ construction, which allows writing $|\pm\hat{\alpha}_{k}\rangle\otimes|k\rangle=|\pm\rangle\tilde{\otimes}|k\rangle$
is nontrivial. However, as for applying Proposition 4, one needs only
to compute the joint marginal and probabilities for an experiment
in the SLD eigenbasis. We find 
\begin{eqnarray*}
p_{\pm,k}^{\lambda} & = & \frac{1}{2}(1\pm\sin2\lambda)p_{k}\\
p_{\pm}^{\lambda} & = & \frac{1}{2}(1\pm\sin2\lambda)
\end{eqnarray*}
from which we obtain $p_{\pm,k}^{\lambda}=p_{\pm}^{\lambda}p_{k}^{0}$
such that $\mathcal{M}_{L_{0}}^{\lambda}=0$ for all $\lambda$ and
\begin{eqnarray*}
QFI & = & \sum_{i=\pm,k}\left(\partial_{\lambda}p_{i,k}^{\lambda}\right)_{\lambda=0}^{2}/p_{i,k}^{0}=\sum_{k}p_{k}\sum_{i=\pm}\left(\partial_{\lambda}p_{i}^{\lambda}\right)_{\lambda=0}^{2}/p_{i}^{0}=\sum_{k}p_{k}QFI_{k}=\sum_{i=\pm}\left(\partial_{\lambda}p_{i}^{\lambda}\right)_{\lambda=0}^{2}/p_{i}^{0}=4
\end{eqnarray*}
The overall estimation precision is a weighted (in terms of the $p_{k}$)
sum of single qubit estimation precisions $QFI_{k}=\sum_{i=\pm}\left(\partial_{\lambda}p_{i}^{\lambda}\right)_{\lambda=0}^{2}/p_{i}^{0},\ \forall k$.
What matters is the variation of the coherence 
\[
-\left[\partial_{\lambda}^{2}Coh_{|\alpha_{\pm,k}\rangle}(\tau_{k})\right]_{\lambda=0}=QFI_{k}
\]
enacted by $G$ for each single qubit state $\tau_{k}$. The above
analysis holds for any generic state of the class (\ref{eq: disc state 1});
therefore it holds for any dimension $N$, for whatever probability
distribution $\left\{ p_{k}\right\} $ and thus for whatever value
of the discord between the subsystem $\mathcal{H}_{2}$ and $\mathcal{H}_{N/2}$.\\

If now for example one of the $\vec{n}_{h}=(\cos\delta_{h},0,\sin\delta_{h})$
does not lie in the $xy$ plane, the state has still discord, and
the above analysis still holds except that now for the specific $h$
the $QFI_{h}=4\sin^{2}\delta_{h}<4$ and thus the overall $QFI$ decreases
with respect to the previous case.\\
 As specific illustrative example of the above reasoning we choose
the following two-qubit states: 
\[
\rho_{1}=\left(\ket{0_{x}}\bra{0_{x}}\otimes\ket{0_{x}}\bra{0_{x}}+\ket{1_{x}}\bra{1_{x}}\otimes\ket{1_{x}}\bra{1_{x}}\right)/2
\]
\[
\rho_{2}=\left(\ket{0_{x}}\bra{0_{x}}\otimes\ket{0_{x}}\bra{0_{x}}+\ket{1_{z}}\bra{1_{z}}\otimes\ket{1_{x}}\bra{1_{x}}\right)/2
\]
where $\ket{0_{x,z}}$, $\ket{1_{x,z}}$ are an eigenstates of $\sigma_{x,z}$.
While $\rho_{1}$ has discord zero, $\rho_{2}$ has discord different
from zero. As generator of the phase shift we choose $G=\sigma_{z}\otimes\mathbb{I}_{2}$.
The estimation is a single qubit one and the overall $QFI$ is equal
to $4$ for $\rho_{1}$, while it is equal to $2$ for $\rho_{2}$.
Notice that presence of discord is not detrimental per se; it is detrimental
for the estimation procedure because, such kind of quantum correlations
are due to the presence of $\ket{1}_{zz}\bra{1}$ in $\rho_{2}$ which
however does not contribute to the estimation process.

\subsection*{GHZ state\label{sub:GHZ-state}}

The definition of the $TPS^{R}$ has been explicitly given in the
main text. The eigenvectors of $L_{0}$ are 
\begin{eqnarray*}
\ket{\pm}\tilde{\otimes}\ket{k} & = & \left(\ket{GHZ_{k}^{+}}\pm i\ket{GHZ_{k}^{-}}\right)/\sqrt{2}
\end{eqnarray*}
We first write the operator $G=\sum_{h}\sigma_{z}^{h}$ in $TPS^{R}$.
Each $\sigma_{z}^{h}$ acts on $M$-qubits states o\textit{f the computational
basis }$\{\ket{k}_{M}=\ket{k_{M},..,k_{1}}\}$ as: 
\[
\sigma_{z}^{h}\ket{k}_{M}=\left(-1\right)^{k_{h}}\ket{k}_{M}
\]
where $k_{h}$ is the $h$-th digit of the binary representation of
$k$. One has 
\begin{eqnarray*}
\sigma_{z}^{h}\ket{\pm}\tilde{\otimes}\ket{k} & = & \left(-1\right)^{k_{h}}\left(\pm i\right)\ket{\mp}
\end{eqnarray*}
and therefore $\sigma_{z}^{h}$ it can be represented within the $k$-th
sector as $\left(-1\right)^{k_{h}}S_{x}\otimes\Pi_{k}$ and on the
overall state space as 
\begin{equation}
\sigma_{z}^{h}=S_{x}\tilde{\otimes}\sum_{k}\left(-1\right)^{k_{h}}\Pi_{k}\label{eq: sigma_z in TPS}
\end{equation}
Consequently the whole Hamiltonian acts as 
\begin{equation}
\sum_{h}\sigma_{z}^{h}=S_{x}\tilde{\otimes}\sum_{k}\left[\left(\sum_{h}\left(-1\right)^{k_{h}}\right)\Pi_{k}\right]\label{eq: Hamiltonian  in TPS}
\end{equation}
where $\sum_{h}\left(-1\right)^{k_{h}}=M-2|k|$ is the difference
between the number of zeros $M-|k|$ and the number of ones $|k|$
present in the $M$ digits binary representation of $k$. Therefore
over the whole state 
\begin{eqnarray}
G=\sum_{h}\sigma_{z}^{h} & = & S_{x}\tilde{\otimes}\sum_{k}\left(M-2|k|\right)\Pi_{k}\label{eq: G in TPSR for GHZ}
\end{eqnarray}
The action of $U_{\lambda}=\exp-i\lambda G$ onto the initial state
$\rho_{0}=\sum_{k}p_{k}\ket{GHZ_{k}^{+}}\bra{GHZ_{k}^{+}}=\ket{0}_{zz}\bra{0}\tilde{\otimes}\sum_{k}p_{k}\Pi_{k}$
gives 
\begin{eqnarray*}
\rho_{\lambda} & = & \sum_{k}p_{k}\tau_{k}^{\lambda}\tilde{\otimes}\Pi_{k}.
\end{eqnarray*}
with $\tau_{k}^{\lambda}=e^{-i\lambda(M-2|k|)S_{x}}\ket{0}_{zz}\bra{0}e^{i\lambda(M-2|k|)S_{x}}$.
In each sector $k$ the state $\tau_{k}^{\lambda}$ is pure and its
Bloch vector is given by $(0,\sin(2\lambda(M-2|k|)),\cos(2\lambda(M-2|k|)))$.
Therefore the measurement onto the eigenstates of $S_{y}\otimes\Pi_{k}$
in each sector $k$

\[
p_{\pm,k}^{\lambda}=\frac{1}{2}\left[1\pm\sin2\lambda(M-2|k|)\right]p_{k}
\]
such that $p_{\pm,k}^{0}=p_{k}/2$ and since 
\[
p_{\pm,}^{\lambda}=\frac{1}{2}\sum_{k}\left[1\pm\sin2\lambda(M-2|k|)\right]p_{k}
\]
one has $p_{\pm}^{0}=1/2$. Furthermore

\begin{eqnarray*}
\left(\partial_{\lambda}p_{\pm,k}^{\lambda}\right)_{\lambda=0} & = & \pm(M-2|k|)p_{k}.
\end{eqnarray*}
and $QFI$ therefore is given by 
\begin{eqnarray*}
QFI & = & \sum_{k}(M-2|k|)^{2}p_{k}
\end{eqnarray*}
Furthermore from

\[
\left(\partial_{\lambda}p_{\pm,}^{\lambda}\right)_{\lambda=0}=\pm\sum_{k}(M-2|k|)p_{k}.
\]
one gets 
\begin{eqnarray*}
FI_{2} & = & 4\left(\sum_{k}(M-2|k|)p_{k}\right)^{2}.
\end{eqnarray*}

\section{QFI and Coherence for the GHZ state under noise\label{sec:Appendix F QFI GHZ states noisy}}

\subsection*{Noise map and its action on the GHZ state.}

The solution of the master equation (\ref{eq: MasterEquationLocal})
was given in Ref.\cite{Acin} and we report it here for the sake of
completeness. The single-qubit map $\Lambda_{\gamma,\omega}$ can
be written in Kraus form as $\Lambda_{\gamma,\omega}(\rho)=\sum_{i,j=\{0,x,y,z\}}S_{ij}\sigma_{i}\rho\sigma_{j}$
with $S_{00}=a+b$, $S_{xx}=d+f$, $S_{yy}=d-f$, $S_{zz}=a-b$, $S_{0z}=S_{z0}^{\ast}=ic$
with

\begin{eqnarray*}
a & = & e^{-\gamma/2t}\cosh\gamma t\\
b & = & e^{-\gamma/2t}\cos(\zeta_{\omega,\gamma}t)\\
c & = & 2\omega/\zeta_{\omega,\gamma}e^{-\gamma/2t}\sin(\zeta_{\omega,\gamma}t)\\
d & = & e^{-\gamma/2t}\sinh\gamma t\\
f & = & \gamma/\zeta_{\omega,\gamma}e^{-\gamma/2t}\sin(\zeta_{\omega,\gamma}t)
\end{eqnarray*}
with $\zeta_{\omega,\gamma}=\sqrt{4\omega^{2}-\gamma^{2}}$.

As shown in Ref.\cite{Acin}, acting on each qubit of the GHZ state
$\rho_{0}=\ket{GHZ_{0}^{+}}\bra{GHZ_{0}^{+}}$ with 
\[
\ket{GHZ_{0}^{\pm}}=\left(\ket{00\dots0}\pm\ket{11\dots1}\right)/\sqrt{2}
\]
the map yields a state $\rho_{\omega,\gamma}(t)$ that is block-diagonal
with 2-dimensional blocks. Indeed, the only nonzero off-diagonal elements
are 
\[
_{M}\bra{k}\rho_{\omega,\gamma}(t)\ket{\bar{k}}_{M}=\left(_{M}\bra{\bar{k}}\rho_{\omega,\gamma}(t)\ket{k}_{M}\right)^{\ast}
\]
where $\ket{k}_{M}\equiv\ket{k_{M},..,k_{1}},\ \ket{\bar{k}}_{M}\equiv\ket{\bar{k}_{M},..,\bar{k}_{1}},\ k=0,..,2^{M-1}-1$
is the computational basis of the global Hilbert space. One has 
\[
_{M}\bra{k}\rho_{\omega,\gamma}(t)\ket{\bar{k}}_{M}=\frac{1}{2}\left[f^{|k|}(b-ic)^{M-|k|}+f^{M-|k|}(b+ic)^{|k|}\right]
\]
where $|k|$ is the number of ones in the string $k_{1}\dots k_{M}$,
while the diagonal elements are 
\[
_{M}\bra{k}\rho_{\omega,\gamma}(t)\ket{k}_{M}=\frac{1}{2}\left[d^{|k|}a{}^{M-|k|}+a^{M-|k|}d{}^{|k|}\right]=_{M}\bra{\bar{k}}\rho_{\omega,\gamma}(t)\ket{\bar{k}}_{M}
\]
As a result, the state can be written as 
\begin{equation}
\rho_{\omega,\gamma}(t)=\sum_{k}r_{k}\left(\ket{k}_{MM}\bra{k}+\ket{\bar{k}}_{MM}\bra{\bar{k}}\right)+\left(s_{k}\ket{k}_{MM}\bra{\bar{k}}+h.c.\right)\label{eq:rhotlong}
\end{equation}

with $r_{k}=_{M}\bra{k}\rho_{\omega,\gamma}(t)\ket{k}_{M}$ and $s{}_{k}=_{M}\bra{k}\rho_{\omega,\gamma}(t)\ket{\bar{k}}_{M}$.

\subsection*{$TPS^{R}$ notation}

In the following we explicitly develop the calculations that allow
to write Eqs. (\ref{eq: master equation for rho_k_0 parallel noise}),
(\ref{eq: Differntial equation for rho_k_M tranverse noise: Coherent part}),
(\ref{eq: Differntial equation for rho_k_M tranverse noise: Decoherent part}).
We first start by writing the state $\rho_{\omega,\gamma}(t)$ in
the $TPS^{R}$ corresponding to the noisless case, see (\ref{eq: GHZ TPSR}).
The basis states are 
\begin{eqnarray}
\left(\ket{GHZ_{k}^{+}}\pm i\ket{GHZ_{k}^{-}}\right)/\sqrt{2} & = & \left((1\pm i)\ket{k}_{M}+(1\mp i)\ket{\bar{k}}_{M}\right)/2=\\
 & = & \ket{\pm}\tilde{\otimes}\ket{k}
\end{eqnarray}
where now $\mathcal{H}_{2^{M-1}}=span\left\{ \ket{k}\right\} $. The
initial state of the evolution is $\ket{GHZ_{0}^{+}}=\frac{\left(\ket{+}+\ket{-}\right)}{\sqrt{2}}\ket{0}$
while the state (\ref{eq:rhotlong}) can be written as 
\[
\rho_{\omega,\gamma}(t)=\sum_{k}p_{k}(t)\tau_{k}(\omega,t)\tilde{\otimes}|k\rangle\langle k|
\]
with 
\[
p_{k}(t)\tau_{k}(\omega,t)=\left(\begin{array}{cc}
r_{k}+Re(s_{k}) & -i\ Im(s_{k})\\
i\ Im(s_{k}) & r_{k}-Re(s_{k})
\end{array}\right)
\]

\subsection*{Parallel noise}

We now exploit the description of $\rho\left(\omega,t\right)$ and
$G$ (\ref{eq: G in TPSR for GHZ})in $TPS^{R}$ in order to write
the coherent part of the evolution (\ref{eq: Coherent part GHZ multi qubit})
as 
\begin{eqnarray}
-i\frac{\omega}{2}\left[\sum_{h}\sigma_{z}^{h},\rho\left(\omega,t\right)\right] & =\nonumber \\
-i\frac{\omega}{2}\sum_{k}\left[S_{x},\tau_{k}(\omega,t)\right]\tilde{\otimes}\left(N-2|k|\right)\Pi_{k}.\label{eq: Coherent part in TPS}
\end{eqnarray}
where $\tilde{\tau}_{k}(\omega,t)$ is the un-normalized single qubit
state pertaining to the sector $k$, each of which enjoys a coherent
dynamics described by 
\begin{equation}
-i\frac{\omega}{2}\left(N-2|k|\right)\left[S_{x},\tau_{k}(\omega,t)\right].
\end{equation}
We now focus on the decoherent part of the master equation (\ref{eq: Decoherent part GHZ multi qubit})
for the case of parallel noise i.e., $\alpha_{z}=1,\alpha_{x}=\alpha_{y}=0$
and $\mathcal{L}(\rho)=-\frac{\gamma}{2}\sum_{h}\left[\rho-\sum_{h}\sigma_{z}^{h}\rho\sigma_{z}^{h}\right]$.
Given the representation of $\sigma_{z}^{h}$ operators (\ref{eq: sigma_z in TPS}),
one finds that 
\begin{eqnarray}
\sum_{h}\sigma_{z}^{h}\rho\sigma_{z}^{h} & = & \sum_{k}S_{x}\left[p_{k}(t)\tau_{k}(\omega,t)\right]S_{x}\tilde{\otimes}\sum_{h}\left(-1\right)^{2k^{h}}\Pi_{k_{M}}=\nonumber \\
 & = & M\sum_{k}S_{x}\left[p_{k}(t)\tau_{k}(\omega,t)\right]S_{x}\tilde{\otimes}\Pi_{k}\label{eq: decoherent part in TPS parallel noise}
\end{eqnarray}
since $\sum_{h}\left(-1\right)^{2k^{h}}=M$ for all $k$'s. Therefore,
in the parallel noise case the decoherent part does not couple the
various sectors $k$. This together with the fact that the initial
state is $\frac{\left(\ket{+}+\ket{-}\right)}{\sqrt{2}}\ket{0}$ shows
that the noisy evolution takes place in the $k=0$ sector only. By
using (\ref{eq: Coherent part in TPS}) and (\ref{eq: decoherent part in TPS parallel noise})
the master equation reduces to the single differential equation for
the single qubit state $\tau_{0}$ reported in the main text i.e.,
\begin{eqnarray*}
\partial_{t}\tau_{0} & = & -\frac{iM\omega}{2}[S_{x},\tau_{0}]+\\
 & - & \frac{M\gamma}{2}\left[\tau_{0}-S_{x}\tau_{0}S_{x}\right]
\end{eqnarray*}

\subsection*{Transverse noise.}

In order to describe the representation of the master equation in
the case of transverse noise in the above introduced TPS, we first
give the representation of $\sigma_{x}^{h}$. For $h<M$, the latter
acts onto the computational basis states as: 
\[
\sigma_{x}^{h}\ket{k}_{M}=\ket{k_{M},..,k_{h+1},\bar{k}_{h},k_{h-1},..k_{1}}\equiv\ket{k'(h)}_{M}
\]
with $k'(h)\in[0,\dots,2^{M-1}-1]$ and analogously 
\[
\sigma_{x}^{h}\ket{\bar{k}}_{M}=\ket{\bar{k}_{M},..,\bar{k}_{h+1},k_{h},\bar{k}_{h-1},..\bar{k}_{1}}\equiv\ket{\overline{k'(h)}}_{M}
\]
where $k_{h}$ is the $h$-th digit of the binary representation of
$k$ and $\bar{k}_{h}$ its negated value. We have $k'(h)=k+(-1)^{k_{h}}2^{h-1}$
and the number of ones in the binary representation of $k'(h)$ is
given by $|k'(h)|=|k|+(-1)^{k_{h}}=|k|\pm1$. Therefore $\sigma_{x}^{h}$,
$h<M$ has the effect of a permutation of the $k$ sectors. 
\begin{eqnarray*}
\sigma_{x}^{h}\ket{\pm}\tilde{\otimes}\ket{k} & = & \ket{\pm}\tilde{\otimes}\ket{k(h)}
\end{eqnarray*}
For $h=M$, one gets 
\[
\sigma_{x}^{M}\ket{k}_{M}=\ket{\bar{k}_{M},k_{M-1}..,k_{1}}\equiv\ket{\bar{k}'(M)}_{M}
\]

\[
\sigma_{x}^{M}\ket{\bar{k}}_{M}=\ket{k_{M},\bar{k}_{M-1},..\bar{k}_{1}}\equiv\ket{k'(M)}_{M}
\]
such that 
\begin{eqnarray*}
\sigma_{x}^{M}\ket{\pm}\tilde{\otimes}\ket{k} & = & \ket{\mp}\tilde{\otimes}\ket{k'(M)}
\end{eqnarray*}
Here, $k'(M)=2^{M-1}-k-1$ and the number of ones in the binary representation
of $k'(M)$ is given by $|k'(M)|=M-1-|k|$. Each $k$ sector is coupled,
by means of $\sigma_{x}^{h}$'s, to the sectors $k'\left(h\right)$.
The representation of $\sigma_{x}^{h}$ in the $TPS^{R}$ is for $h<M$
\begin{eqnarray}
\sigma_{x}^{h} & = & \mathbb{I}_{2}\tilde{\otimes}O_{h}\label{eq: sigma_x in TPS h neq 1-1}
\end{eqnarray}
where the traceless unitary operator $O_{h}=\left(\sum_{k=0}\ket{k'(h)}\bra{k}+\ket{k}\bra{k'(h)}\right)=O_{h}^{\dagger}$
enacts a permutation on the basis states $\ket{k}$, while 
\begin{eqnarray}
\sigma_{x}^{M} & = & S_{z}\tilde{\otimes}O_{M}\label{eq: sigma_x in TPS h =00003D00003D 1}
\end{eqnarray}
with $O_{M}=\sum_{k=0}\ket{k'(M)}\bra{k}+\ket{k}\bra{k'(M)}$.

Let us check how decoherence works for a single $\rho_{k}$. As we
have already seen the coherent part of the evolution can be written
as 
\begin{eqnarray*}
Tr\left[\mathbb{I}_{2}\tilde{\otimes}\Pi_{k}\left(\frac{-i\omega}{2}\left[H,\rho(\omega,t)\right]\right)\right] & =\\
=\frac{-i\omega\left(M-2|k|\right)}{2}\left[S_{x},\tilde{\tau}_{k}\right].
\end{eqnarray*}
If one instead takes the whole trace $Tr_{\mathcal{H}_{2^{M-1}}}$,
one has to sum up the last relation for all $k$ obtaining 
\begin{equation}
\frac{-i\omega}{2}\sum_{k}\left(M-2|k|\right)\left[S_{x},\tilde{\tau}_{k}\right]\label{eq: Hamiltonian part reduced density matrix single qubit}
\end{equation}
As for the decoherent part we have to first write the term $\sigma_{x}^{h}\rho\sigma_{x}^{h}$.
One has that 
\begin{eqnarray*}
O_{h}\Pi_{k}O_{h}^{\dagger} & = & \Pi_{k'(h)}\\
O_{h}\Pi_{k'(h)}O_{h}^{\dagger} & = & \Pi_{k}
\end{eqnarray*}
and therefore for each $h<M$ 
\begin{eqnarray*}
\tilde{\tau}_{k}(\omega,t)\tilde{\otimes}\Pi_{k} & \rightarrow & \tilde{\tau}_{k}(\omega,t)\tilde{\otimes}\Pi_{k'(h)}\\
\tilde{\tau}_{k'(h)}(\omega,t)\tilde{\otimes}\Pi_{k'(h)} & \rightarrow & \tilde{\tau}_{k'(h)}(\omega,t)\tilde{\otimes}\Pi_{k}
\end{eqnarray*}
Therefore, one gets 
\begin{eqnarray*}
\sigma_{x}^{h}\rho\sigma_{x}^{h} & = & \sum_{k}\tilde{\tau}_{k}(\omega,t)\tilde{\otimes}\Pi_{k'(h)}=\sum_{k}\tilde{\tau}_{k'(h)}(\omega,t)\tilde{\otimes}\Pi_{k}
\end{eqnarray*}
The effect is therefore to reshuffle the original state upon exchanging
$\tilde{\tau}_{k}(\omega,t)\leftrightarrow\tilde{\tau}_{k'(h)}(\omega,t)$.
For $h=M$ one has

\begin{eqnarray*}
\sigma_{x}^{M}\rho\sigma_{x}^{M} & = & \sum_{k}S_{z}\tilde{\tau}_{k'(M)}S_{z}\tilde{\otimes}\Pi_{k}\\
\sigma_{x}^{M}\rho\sigma_{x}^{M} & = & \sum_{k}S_{z}\tilde{\tau}_{k}S{}_{z}\tilde{\otimes}\Pi_{k'(M)}
\end{eqnarray*}
The effect is thus to reshuffle the original state upon exchanging
$\tilde{\tau}_{k}(\omega,t)\leftrightarrow\tilde{\tau}_{k'(M)}(\omega,t)$
and apply a $S_{z}$ rotation. Taking everything into account, the
differential equation for a single $\tilde{\tau}_{k}(\omega,t)$ can
be written as 
\begin{eqnarray*}
\partial_{t}\tilde{\tau}_{k}(\omega,t) & = & -\frac{i\omega\left(M-2|k|\right)}{2}\left[S_{x},\tilde{\tau}_{k}\right]-\frac{\gamma}{2}\left(M\tilde{\tau}_{k}-\left[S_{z}\tilde{\tau}_{k'(M)}S_{z}+\sum_{h=1}^{M-1}\tilde{\tau}_{k'(h)}\right]\right)
\end{eqnarray*}
To obtain the evolution of the reduced state $\xi\left(t,\omega\right)$,
one should take the trace over $\mathcal{H}_{2^{M-1}}$ which just
corresponds to summing over $k$. Since$\sum_{k}\tilde{\tau}_{k}(\omega,t)=\xi\left(t,\omega\right)$
but also $\sum_{k}\tilde{\tau}_{k'(h)}=\xi\left(t,\omega\right)$,
and $\sum_{k'}\left(\sum_{h=1}^{M-1}\tilde{\tau}_{k'(h)}\right)=\left(M-1\right)\xi\left(t,\omega\right)$
we finally get

\begin{eqnarray*}
\partial_{t}\xi\left(t,\omega\right) & = & \frac{-i\omega}{2}\left[S_{x},\sum_{k}\left(M-2|k|\right)\tilde{\tau}_{k}\right]+\frac{\gamma}{2}\left(\xi\left(t,\omega\right)-S_{z}\xi\left(t,\omega\right)S_{z}\right)
\end{eqnarray*}

\section{Coherence and QPTs\label{sec: Appendix Coherence-and-QPTs}}

We consider a family of states $|0_{\lambda}\rangle$ that are the
ground states of the generic Hamiltonian $H_{\lambda}=H_{0}+\lambda V$
labeled by a continuous parameter $\lambda$. In the same notation
of Methods xx, we can write, to first order in $\delta\lambda$, $\ket{0^{\lambda+\delta\lambda}}=\ket{0^{\lambda}}+\delta\lambda\ket{v}$
where $\ket{v}$ is the first order correction one can obtain with
standard perturbative analysis \cite{Cohen}: 
\[
|v\rangle=|v^{\perp}\rangle=\sum_{n\neq0}\frac{\langle0^{\lambda}|V|n^{\lambda}\rangle}{(E_{n}^{\lambda}-E_{0}^{\lambda})}|n^{\lambda}\rangle
\]
with $|n^{\lambda}\rangle,E_{n}^{\lambda}$ eigenvectors and eigenvalues
of $H_{\lambda}$. It holds $\left\langle v|v\right\rangle =\sum_{n\neq0}\frac{\left|\langle0^{\lambda}|V|n^{\lambda}\rangle\right|^{2}}{\left(E_{n}^{\lambda}-E_{0}^{\lambda}\right)^{2}}$
and we define $\ket{\hat{v}}=\ket{v}/\sqrt{\left\langle v|v\right\rangle }$;
by construction $\left\langle 0^{\lambda}|\hat{v}\right\rangle =0$.
As eigenbasis of the SLD we can choose $\mathbf{\mathcal{B}_{\boldsymbol{\mathbf{\alpha}}}=}\left\{ |\alpha_{\pm}\rangle=\frac{1}{\sqrt{2}}(|0^{\lambda}\rangle\pm|v\rangle)\right\} \bigcup\left\{ |2\rangle,\dots|N\rangle\right\} $
with the only requirement that $\langle\alpha_{\pm}|n\rangle=0\ \forall n\geq2$.
By again using the same notations of Methods xx we obtain the measurement
probabilities 
\[
p_{\pm}^{\lambda+\delta\lambda}=\left|\left\langle 0^{\lambda+\delta\lambda}|\alpha_{\pm}\right\rangle \right|^{2}=\frac{1}{2}\left(1\pm2|v|\delta\lambda\right)+\mathcal{O}(\delta\lambda^{2})
\]
and 
\[
p_{n}^{\lambda}=\left|\left\langle 0^{\lambda+\delta\lambda}|n\right\rangle \right|^{2}=\mathcal{O}(\delta\lambda^{3})\ \forall n\geq2
\]
.Consequently we obtain the desired result 
\begin{eqnarray*}
QFI & = & -\left(\partial_{\delta\lambda}^{2}Coh_{\mathcal{B}_{\boldsymbol{\mathbf{\alpha}}}}\right)_{\delta\lambda=0}\\
 & = & \sum_{i=\pm}\frac{\left(\partial_{\delta\lambda}p_{i}^{\delta\lambda}\right)_{\delta\lambda=0}}{p_{i}^{0}}=4|v|^{2}=\\
 & = & 4\sum_{n\neq0}\frac{|\langle0^{\lambda}|V|n^{\lambda}\rangle|^{2}}{(E_{n}^{\lambda}-E_{0}^{\lambda})^{2}}=4g_{\lambda}^{FS}.
\end{eqnarray*}
Notice that, although the choice of $\left\{ |2\rangle,\dots|N\rangle\right\} $
is not unique the result holds for any of the possible choices as
long as $\langle\alpha_{\pm}|n\rangle=0\ \forall n\geq2$. The scaling
properties when $\lambda\rightarrow\lambda_{c}$ follow from the those
of $g_{\lambda}^{FS}$. 
\end{document}